\documentclass{article}
\usepackage[margin=1in,footskip=0.2in]{geometry}
\usepackage{caption} 

\usepackage{color, colortbl}
\usepackage{amsmath}
\usepackage{graphicx}
\usepackage[square,sort,comma,numbers]{natbib}
\usepackage{url} 

\usepackage[utf8]{inputenc} 
\usepackage{ifthen}  
\usepackage{multicol}   
\usepackage{multirow}

\usepackage{booktabs}
\usepackage{filecontents}
\usepackage{rotating}
\usepackage{comment}
\usepackage{amsmath}
\usepackage{mathtools}
\usepackage{graphics}
\usepackage{array}
\newcolumntype{C}[1]{>{\centering\let\newline\\\arraybackslash\hspace{0pt}}m{#1}}

\usepackage{threeparttable}
\usepackage{mathrsfs}
\usepackage{cprotect}
\usepackage{arydshln}
\usepackage{ dsfont }
 \usepackage{amsthm}
 \usepackage{comment}
\usepackage[normalem]{ulem}
\usepackage{bm}
\usepackage{authblk}

\newtheorem{theorem}{Theorem}
\newtheorem{remark}{Remark}
\newtheorem{proposition}{Proposition}
\newtheorem{lemma}{Lemma}
\newtheorem{corollary}{Corollary}

\providecommand{\keywords}[1]{\textbf{\textit{Keywords:}} #1}

\begin{document}

\bibliographystyle{unsrtnat}

\title{Tuning in ridge logistic regression to solve separation}

\author[1]{Hana \v{S}inkovec}
\author[1]{Angelika Geroldinger}
\author[1]{Georg Heinze}
\author[2]{Rok Blagus}
\affil[1]{Section for Clinical Biometrics, 
Center for Medical Statistics, Informatics and Intelligent Systems, 
Medical University of Vienna}
\affil[2]{Institute for Biostatistics and Medical Informatics, 
University of Ljubljana,
Faculty of Medicine}

\date{}

\markboth%
{H. \v{S}inkovec, A. Geroldinger, G. Heinze, R. Blagus}
{Tuning in ridge logistic regression to solve separation}

\maketitle

\begin{abstract}
{Separation in logistic regression is a common problem causing failure of the iterative estimation process when finding maximum likelihood estimates. Firth's correction (FC) was proposed as a solution, providing estimates also in presence of separation. In this paper we evaluate whether ridge regression (RR) could be considered instead, specifically, if it could reduce the mean squared error (MSE) of coefficient estimates in comparison to FC. In RR the tuning parameter determining the penalty strength is usually obtained by minimizing some measure of the out-of-sample prediction error or information criterion. However, in presence of separation tuning these measures can yield an optimized value of zero (no shrinkage), and hence cannot provide a universal solution. We derive a new bootstrap based tuning criterion $B$ that always leads to shrinkage. Moreover, we demonstrate how valid inference can be obtained by combining resampled profile penalized likelihood functions. Our approach is illustrated in an example from oncology and its performance is compared to FC in a simulation study. Our simulations showed that in analyses of small and sparse datasets and with many correlated covariates $B$-tuned RR can yield coefficient estimates with MSE smaller than FC and confidence intervals that approximately achieve nominal coverage probabilities.}
\end{abstract}
\keywords{Logistic regression; Firth's correction; Separation; Ridge regression; Tuning; Inference; Mean squared error.}

\section{Introduction}
\label{intro}
Logistic regression is often used to describe the association of a set of covariates with a binary outcome, to estimate the effect of one exposure of interest on the binary outcome adjusted for covariates, or to  predict the probability of one of the levels of the outcome. The regression coefficients can be interpreted as log odds ratios. They are usually estimated by the method of maximum likelihood (ML). ML estimators are asymptotically unbiased and have optimal properties under some regularity conditions \citep{Kosmidis}.

However, in the analyses of small data sets or in situations where the outcomes or exposures are rare, covariates are strongly correlated with each other or strongly associated with the outcome, ML estimates of the regression coefficients may be biased away from zero or may even not exist \citep{Albert, Angelika}, i.e., the iterative estimation process does not converge. This situation has been termed \textit{separation} as the two outcome groups are separated by the values of a covariate or a linear combination of covariates  \citep{Heinze02}. Therefore, in the presence of sepration the fitted model yields perfect predictions for some (quasi-complete separation) or for all observations (complete separation).
A procedure that detects separation in the data was proposed by \citet{safeBinaryRegression}. In practice simpler alternatives, e.g. by observing  standard errors of standardized coefficients that become extremly large, are often used \citep{Heinze02}.

Although separation could arise as a feature of a data-generating process, we usually do not assume that one or several covariates are sufficient to  perfectly predict the outcome in the underlying population. Instead we consider separation as a consequence of random data variation. Firth's correction \citep{Firth93} (FC) was proposed as a solution to separation as it removes the $O(n^{-1})$ bias of the coefficients even in situations where the prevalence of separation is high \citep{Heinze02}. With FC the likelihood function is multiplied by a penalty term leading to shrinkage of the maximizer of the penalized likelihood towards zero. The penalty strength relative to the likelihood is determined by the number and correlation structure of covariates, hence no tuning is required.

In this paper we introduce an alternative solution to separation based on ridge regression (RR) and compare its performance in terms of mean squared error (MSE) of coefficients to the one of FC. RR adds a quadratic penalty to the likelihood and has been repeatedly reported to perform well in situations with correlated covariates \citep{flicflac, hdr}. Its usage is generally recommended whenever interest lies in obtaining coefficients and predicted probabilities with small MSE, but inference is not required \citep{flicflac}. In RR an additional parameter must be specified which determines the relative weight of the penalty in the penalized likelihood. This parameter is usually found by tuning criteria based on prediction performance, often estimated by cross-validation. In standard software tuning by cross-validated deviance, mean squared prediction error or mean classification error is commonly implemented \citep{glmnet, Hans92}. Tuning has a potential of finding an optimal penalty strength such that coefficient estimates have reduced MSE. However, determining the tuning parameter based on data introduces additional variability which has to be taken into account when constructing confidence intervals (CI); na\"{i}ve approaches, e.g., deducing Wald CI from the penalized covariance matrix, cannot be recommended as they do not achieve the nominal coverage levels due to the combination of bias and reduced variance \citep{flicflac}. Moreover, in the presence of separation tuning a measure of prediction performance might lead to no shrinkage and poor inference. We explain why this happens when using the commonly used tuning criteria and derive a new tuning criterion, denoted as B-tuning, that always leads to shrinkage. We also propose an associated valid procedure for CI estimation that takes into account the additional variabilty introduced by tuning.

The paper is organized accordingly. In the following section we summarize ML and penalized ML logistic regression and review commonly used criteria for optimization of the tuning parameter in RR together with our new criterion. Our proposal for interval estimation is also described in that section. Later we report on a comprehensive simulation study evaluating the performance of RR tuned by different criteria in comparison to FC, focusing on scenarios where separation occurs frequently. We exemplify the described methods in the analysis of a real data set. Finally, we summarize the most important findings of this work.

\section{Methods}
\label{meth}

Suppose we have $i=1,\ldots,N$, independent Bernoulli-distributed observations $y_i$, where $y_i=1$ denotes an event occurring in the $i$th subject and let $\pi_i$ denote the corresponding event probability. In the logistic regression model the logit of $\pi_i$ (i.e., the log-odds of the event) is modeled as a linear combination of the values of $K$ covariates,  $X=\{X_1,\ldots,X_K\}$, $K<N$, so that
\begin{equation}
\label{mod0}
\log\frac{\pi_i}{1-\pi_i}=\beta_0+\beta_1x_{i1}+\ldots+\beta_Kx_{iK}\mbox{, }i=1,\ldots,N,
\end{equation}
where $\beta_0$ is an intercept, and $\beta_k$, $k=1,\ldots,K$, are the logistic regression coefficients for $X_k$ that can be interpreted as the log-odds ratios, corresponding to one unit differences in $X_k$. The coefficients of the model can be obtained by the ML method, maximizing the likelihood function,
$$L(\bm\beta)=\prod_{i=1}^n \pi_i^{y_i}(1-\pi_i)^{1-y_i}, $$
or equivalently its natural logarithm, $l(\bm\beta)=\log L(\bm\beta)$, where $\bm\beta=(\beta_0,\ldots,\beta_K)^T$.

\subsection{Firth's correction}
\label{firth}
With FC the likelihood function is penalized by a term equivalent to Jeffreys' invariant prior so that the penalized log-likelihood becomes
\begin{equation}
\label{modFC}
l^P (\bm\beta)=l(\bm\beta)+\frac{1}{2} \log|I(\bm\beta)|,
\end{equation}
with $I(\bm\beta)$ denoting the Fisher information matrix evaluated at $\bm\beta$. To obtain valid inference, the use of profile penalized likelihood CI was recommended \citep{Heinze02}.


\subsection{Ridge regression}
\label{ridge}
In RR the penalized log-likelihood is given by
\begin{equation}
\label{modRR}
l^\lambda (\bm\beta)=l(\bm\beta)-\frac{\lambda}{2} \sum_{k=1}^K\beta_k^2,
\end{equation}
for some  $\lambda>0$, where a quadratic constraint is imposed on the coefficients so that their estimates are shrunken towards zero. Let $\bm{\hat\beta}^\lambda=(\hat\beta_0^\lambda,\ldots,\hat\beta_K^\lambda)^T$ be the penalized maximum likelihood estimate obtained as a solution to
 $$\frac{\partial l^\lambda (\bm\beta)}{\partial \bm\beta}=0 .$$
In contrast to FC an intercept $\beta_0$ is usually excluded from the penalty function which results in on average unbiased predicted probabilities for new observations \citep{Elgmati2015,Greenland15}.

From the Bayesian perspective, in RR univariate normal priors are imposed on the coefficients. These normal priors have mean zero and variance $\sigma^2_{prior}$ inversely proportional to $\lambda$. By using data augmentation it is possible to approximately fit RR by any statistical software where ML fitting and weighting of the observations is possible \citep{sullivan13}. Briefly, two pseudo-observations (event and non-event) are added for each covariate; the values for this covariate corresponding to the prior are set to $1/s$, where $s$ is a scaling factor which improves the approximation (usually set to $s=10$), and to zero for all other covariates. The weights for the augmented data used in ML fitting are then set to 1 for the original observations and to $2s^2/\sigma^2_{prior}$ for the pseudo-observations. This approach has been used to fit RR throughout this paper.

The consequence of using univariate priors is that, unlike in FC, the invariance to linear transformations of the parameter space is lost. Using a common penalty for all coefficients requires standardization of the covariates such that the coefficients approximately share their interpretation and are represented on the same scale. Usually standardization to zero mean and unit variance is applied. For a setting where both binary and continuous covariates are present \citet{gelman2008} proposed to first center both types of covariates to a mean of 0; continuous covariates are then scaled to a standard deviation of 1/2, while binary covariates are scaled so that their range is 1. The standardized estimates (of the continuous covariates) are then transformed back to the original scale to facilitate interpretation.

\subsubsection{Overview of commonly used tuning criteria}
\label{tuneD}
When optimizing $\lambda$ a grid of tuning parameter values is usually considered and a corresponding model is evaluated at each value of the grid. The selected tuning parameter is the one from the model with the smallest prediction error \citep{Hans92}. A value of $\lambda$ optimized at zero indicates that penalization does not reduce prediction error compared to the ML model. The prediction error is most often estimated by cross-validation (CV), and may be defined in different ways. The most commonly used criteria are, e.g.:

\begin{itemize}
\item minus twice log-likelihood error, i.e. deviance (\textit{D})
\begin{equation}
\label{dev}
D=-2\sum_{i=1}^n \left(y_i\log \hat\pi_{(-i)}+(1-y_i)\log(1-\hat\pi_{(-i)})  \right),
\end{equation}
\item mean squared prediction error (\textit{E})
\begin{equation}
\label{pe}
E=\frac{1}{n}\sum_{i=1}^n (y_i-\hat\pi_{(-i)})^2,
\end{equation}
\item mean classification error (\textit{C})
\begin{equation}
\label{ce}
C=\frac{1}{n}\sum_{i=1}^n\left[y_i \mbox{I}\left(\hat\pi_{(-i)}<\frac{1}{2}\right)+(1-y_i)\mbox{I}\left(\hat\pi_{(-i)}>\frac{1}{2}\right)+\frac{1}{2}\mbox{I}\left(\hat\pi_{(-i)}=\frac{1}{2}\right) \right],
\end{equation}
where $\mbox{I}(\cdot)$ is the indicator function, i.e. $\mbox{I}(\cdot)=1$, if the condition inside the brackets holds, and  $\mbox{I}(\cdot)=0$ otherwise.
\end{itemize}

In the definitions above
$$\hat\pi_{(-i)}=\frac{1}{1+\exp\left(-(\hat\beta_0^{(-i)}+\hat\beta_1^{(-i)}X_{i1}+\ldots+\hat\beta_K^{(-i)}X_{iK})\right)}, $$
denotes the estimated event probability for subject $i$, based on the (penalized) estimates of $\beta_k$, i.e. $\hat\beta_k^{(-i)}$, $k=0,\ldots,K$ where this subject is omitted. Note that $C$ does not generally have a unique optimum in $\lambda$.

Instead of CV an independent validation set, if available, may be used to estimate the prediction error.

Alternatively, one may use information criteria to avoid resampling when optimizing $\lambda$, e.g. the Akaike's information criterion (\textit{A}) \citep{AIC}. For RR it is given by

\begin{equation}
\label{aic}
A(\lambda)=-2\left(l^\lambda (\hat{\bm\beta})-df_e \right)
\end{equation}

with $df_e$ denoting an estimate of the effective degrees of freedom, given by

$$df_e=\mbox{trace}\left[ \frac{\partial^2 l}{\partial^2 \bm\beta}\left(\bm{\hat\beta}^\lambda\right) \left(\frac{\partial^2 l^\lambda }{\partial^2 \bm\beta}\left(\bm{\hat\beta}^\lambda\right)\right)^{-1} \right], $$

where the first term is the information matrix ignoring the penalty function, and the second term is the covariance matrix computed by inverting the information matrix of RR \citep{Harrell}. Note that for $\lambda=0$, $df_e=K+1$, the number of parameters being estimated, while for $\lambda>0$, $df_e<K+1$ holds \citep{Harrell}.

\subsubsection{Point and interval estimation for tuned ridge regression}
\label{tune}

For a saturated model with categorical predictors we theoretically explored the behavior of the criterion functions described above in cases of data separation (see Additional file 1). The criterion functions described in section \ref{tuneD} are all derived from a prediction perspective. In terms of prediction separation causes only little problems. For example, with complete separation the estimated probabilities equal the binary outcome when $\lambda=0$ and the (cross-validated) prediction error is as small as it can be - equal to zero. Thus, the criteria depending on the estimated probabilities are overoptimistic (as perfect predictions are not assumed), but could still work well enough for prediction purposes \citep{flicflac}. With $\lambda=0$, however, estimation of coefficients fails. Therefore, if interest lies in interpretability of regression coefficients, a tuning criterion should target inference instead, e.g., by directly taking into account the MSE of the coefficients \citep{msemin}. For that purpose we introduce a second-order moment matrix given by
$$\mathbf{m}(\lambda)=\mathbf{v}(\lambda)+\mathbf{b}(\lambda)\mathbf{b}(\lambda)^T, $$
where $\mathbf{v}(\lambda)$ and $\mathbf{b}(\lambda)$ are the variance and the bias of the RR estimator. Resorting to asymptotic results, the bias and variance of the RR estimator are given by
$$\mathbf{b}(\lambda)=-\lambda (I(\bm\beta)+\lambda\mathbf{I})^{-1}\bm\beta $$
and
$$\mathbf{v}(\lambda)= (I(\bm\beta)+\lambda\mathbf{I})^{-1} I(\bm\beta)  (I(\bm\beta)+\lambda\mathbf{I})^{-1},$$
respectively, where $\mathbf{I}$ is the identity matrix \citep{Hans92}. Note that $\mathbf{b}(\lambda)$ and $\mathbf{v}(\lambda)$ depend on the unknown true value of the parameter vector, $\bm\beta$. A meaningful estimate $\bm{\tilde\beta}$ has to be plugged in  for the approximation to be of any practical use. Let
\begin{equation}\label{est.mse}\widehat{\mathbf{m}}(\lambda)=\widehat{\mathbf{v}}(\lambda)+\widehat{\mathbf{b}}(\lambda)\widehat{\mathbf{b}}(\lambda)^T \end{equation}
be the estimated second-order moment matrix which is obtained by replacing $\bm{\beta}$  by $\bm{\tilde\beta}$. Optimization of $\widehat{\mathbf{m}}(\lambda)$ is then possible by using numerical methods, defining some scalar target function directly from $\widehat{\mathbf{m}}(\lambda)$, e.g, $\mbox{trace}(\mathbf{B}\widehat{\mathbf{m}}(\lambda))$, for a suitable non-zero positive semi-definite matrix $\mathbf{B}$. Using $\mathbf{B}=\mathbf{I}$ gives equal importance to all coefficients, while other choices may focus on  weighting the importance of each coefficient differently; zero weights are also possible for some coefficients. Our theoretical results to Theorem 1 show that using this criterion, termed \textit{B}-tuning criterion, yields an optimal tuning parameter different from zero even with complete separation given that an appropriate estimator $\bm{\tilde\beta}$ is used. The proof of the theorem is given in Additional file 1.

\begin{theorem}
Let $\bm{\tilde\beta}$ be some estimator of the $K+1$-parameter vector $\bm\beta$ such that 
$\bm{\tilde\beta}^T\bm{\tilde\beta}<\infty$ (condition 1). Let $\widehat{\mathbf{m}}(\lambda)$ be the estimate of $\mathbf{m}(\lambda)$ which is obtained by plugging in $\bm{\tilde\beta}$ for $\bm{\beta}$. Then $\mbox{trace}\left(\mathbf{B}\widehat{\mathbf{m}}(\lambda)\right)$ is minimized at some $\lambda>0$ for any non-zero, positive semi-definite matrix $\mathbf{B}$.
\end{theorem}

In the presence of separation ML estimator does not fulfill condition 1. In contrast, the FC estimator satisfies this condition even with complete separation; if using FC to estimate the second-order moment matrix, the optimal value of $\lambda$ will be strictly greater than zero, thus, providing finite coefficient estimates. We favor FC for its good small sample properties but in principle any other estimator satisfying condition 1 could be plugged in.

Using an estimate of $\bm\beta$ when estimating $\mathbf{m}(\lambda)$ introduces additional variability which should be taken into an account if one wishes to construct valid CI. For this reason and to reduce the variance in estimating $\lambda$, we propose a combination of bootstrap and a procedure to combine likelihood profiles (CLIP), originally presented in the context of multiple imputation \citep{GeorgCLIP}. In brief, our proposed approach consists of the following steps:

\renewcommand{\theenumi}{\roman{enumi}}%
\begin{enumerate}
  \item bootstrap samples $b=1,\ldots,B$ are drawn with replacement and $\bm{\tilde\beta}_b$ is obtained on each sample;
  \item the scalar target function $\mbox{trace}(\mathbf{B}\widehat{\mathbf{m}}(\lambda))$ is numerically optimized for each bootstrap sample using $\bm{\tilde\beta}_b$ to obtain $\lambda_b^*$, $b=1,\ldots,B$;
  \item the optimal penalty parameter is determined as the mean of $\lambda_b^*$;
  \item a penalized likelihood profile is computed on the original sample by fixing $\lambda$ at $\lambda_b^*$;
  \item $B$ penalized likelihood profiles are pooled using CLIP as described in \citet{GeorgCLIP}.
\end{enumerate}

As an alternative of taking the mean of $\lambda_b^*$, the mode of the combined profile likelihood function can be used as a point estimate in tuned RR. 

\section{Simulation study}
\label{sims}

The simulation study is described following recent recommendations \citep{Morris2019}.

\textbf{Aims:} We performed a comprehensive simulation study to investigate the properties of point and interval estimates obtained by the proposed procedures in comparison with default solutions to fit logistic regression models, focusing on scenarios where separation occurs frequently.

\textbf{Data-generating mechanisms:} We considered a data generation scheme capturing a plausible biomedical context similar to the one described in \citet{Binder}, allowing a fair comparison between new and existing approaches \citep{BoulesteixBiomJ2018}. Specifically, we simulated data sets of size $N$ consisting of a binary outcome variable $Y$ and $K$ covariates $X_1,\ldots,X_k$ of mixed types. These covariates were obtained by first sampling $Z_1 ,\ldots, Z_k$ from a standard multivariate normal distribution with a zero mean vector and correlation matrix $\Sigma$ defined in Table ~\ref{tabsimX}. To obtain the desired $X_k$, we applied transformations to $Z_k$ as summarized in Table \ref{tabsimX}. Covariates $X_1, X_2, X_8, X_9$ were binary, $X_{10}$ and $X_{11}$ were ordinal with 3 levels each, and $X_3 ,\ldots, X_{7}$ and $X_{12} ,\ldots, X_{15}$ were continuous. To avoid too extreme values, continuous covariates were truncated at the third quartile plus five times the interquartile distance of the corresponding distribution.

\begin{table}[h]
\centering
\caption{Covariate structure and effect sizes applied in the simulation study based on \citet{Binder}. $[\cdot]$ denotes that the non-integer part of the argument is removed. $I(\cdot)$ is the indicator function, which takes the value of 1 if the argument is true, and 0 otherwise. $\beta_1\in\{0.69, 1.39, 2.08\}$.
\label{tabsimX}}
\scriptsize
\begin{tabular}{llllll}
  \hline
$z_{ik}$ & Correlation of $z_{ik}$ & Type & $x_{ik}$ & $\mbox{E}(x_{ik})$ & $\beta$ \\
  \hline
$z_{i1}$ & $z_{i2}(0.5)$, $z_{i6}(0.5)$, $z_{i12}(0.5)$ & binary & $x_{i1}=I(z_{i1}<0.84)$  & 0.8 & $\beta_1$ \\
$z_{i2}$ & $z_{i1}(0.5)$, $z_{i6}(0.3)$, $z_{i11}(-0.3)$ & binary & $x_{i2}=I(z_{i2}<-0.35)$ & 0.36 & 0.693 \\
$z_{i3}$ & $z_{i4}(0.3)$, $z_{i11}(-0.5)$ & continuous & $x_{i3}=\exp(0.4z_{i3}+3)$ & 21.8 & 0 \\
$z_{i4}$ & $z_{i3}(0.3)$, $z_{i7}(0.5)$, $z_{i13}(-0.3)$ & continuous & $x_{i4}=\exp(0.5z_{i4}+1.5)$ & 5.1 & 0 \\
$z_{i5}$ & - & continuous & $x_{i5}=0.01[100(z_{i5}+4)^2]$ & 17 & 0 \\
$z_{i6}$ & \begin{tabular}[x]{@{}c@{}}$z_{i1}(0.5)$, $z_{i2}(0.3)$, $z_{i9}(0.5)$, \\ $z_{i13}(0.5)$, $z_{i14}(0.3)$\end{tabular} & continuous & $x_{i6}=[10z_{i6}+55]$ & 54.5 & 0 \\
$z_{i7}$ & $z_{i4}(0.5)$ & continuous & $x_{i7}=[10z_{i7}+55]$ & 54.5 & 0 \\
$z_{i8}$ & - & binary & $x_{i8}=I(z_{i8}<0)$ & 0.5 & 0.347 \\
$z_{i9}$ & $z_{i6}(0.5)$, $z_{i13}(0.3)$, $z_{i14}(0.3)$  & binary & $x_{i9}=I(z_{i9}<0)$  & 0.5 & 0 \\
$z_{i10}$ & $z_{i13}(-0.5)$, $z_{i14}(-0.3)$ & ordinal & $x_{i10}=I(z_{i10}\geq -1.2)+I(z_{i10}\geq 0.75)$ & 1.11 & 0.693\\
$z_{i11}$ & \begin{tabular}[x]{@{}c@{}}$z_{i2}(-0.3)$, $z_{i3}(-0.5)$, $z_{i12}(0.3)$, \\ $z_{i13}(0.3)$\end{tabular} & ordinal & $x_{i11}=I(z_{i11}\geq 0.5)+I(z_{i11}\geq 1.5)$ & 0.37 & 0.347 \\
$z_{i12}$ &  $z_{i1}(0.5)$, $z_{i11}(0.3)$ & continuous & $x_{i12}=[10z_{i12}+55]$ & 54.5 & 0.036 \\
$z_{i13}$ & \begin{tabular}[x]{@{}c@{}}$z_{i4}(-0.3)$, $z_{i6}(0.5)$, $z_{i9}(0.3)$, \\ $z_{i10}(-0.5)$, $z_{i11}(0.3)$, $z_{i14}(0.5)$\end{tabular} & continuous & $x_{i13}=[\mbox{max}(0,100\exp(z_{i13})-20]$ & 138.58 & 0.003 \\
$z_{i14}$ & \begin{tabular}[x]{@{}c@{}}$z_{i6}(0.3)$, $z_{i9}(0.3)$, $z_{i10}(-0.3)$, \\ $z_{i13}(0.5)$\end{tabular}  & continuous & $x_{i14}=[\mbox{max}(0,80\exp(z_{i14})-20]$  & 106.97 & 0.004 \\
$z_{i15}$ & - & continuous & $x_{i15}=[10z_{i15}+55]$ & 54.5 & 0.036\\
\hline
\end{tabular}
\end{table}

We considered all the combinations of simulation parameters and investigated 27 data generating mechanisms defined by the number of covariates $K\in\{7, 10, 15\}$, the sample size $N\in\{80, 200, 500\}$ and the value of $\beta_1\in\{0.69, 1.39, 2.08\}$. The effects of other covariates were held constant; the true regression coefficients for binary covariates were $\beta_2=0.69$, $\beta_8=0.35$, $\beta_9=0$, for ordinal covariates were equal to $\beta_{10}=0.69$, $\beta_{11}=0.35$ and for continuous covariates $X_{12},...,X_{15}$ $\beta_k$ was 
chosen such that the log odds ratio between the first and the fifth sextile of the corresponding distribution was 0.69. The true effects of the other five continuous covariates $X_3,\ldots,X_{7}$ were set to zero. An intercept term $\beta_0$ was determined for each scenario such that on average the desired proportion of events of 0.1 was obtained. Given simulated values of covariates $x_{i1},\ldots, x_{iK}$ for a subject $i$, $i=1,\ldots,N$, a linear predictor was obtained as $\eta_i = \beta_0 + \sum_{k=1}^{K} x_{ik} \beta_k$, from which $\pi_i$ was obtained as $\pi_i=1/(1+\exp(-\eta_i))$. Finally, the value of the outcome variable was sampled from a Bernoulli distribution as $y_i\sim \mbox{Bern}(\pi_i)$. For each data generating mechanism 1000 data sets were generated.

\textbf{Methods:} Each simulated data set was analyzed by fitting two models using the methods outlined in \ref{meth}. The first model was fitted without $X_3,\ldots,X_{7}$ that had no real effect on the outcome, while the second model also included those noise covariates. For brevity, here we focus on the methods presented in Table ~\ref{tabMethods}, i.e. FC and \textit{D}-, \textit{A}- and \textit{B}-tuned RR. See Additional file 2 for detailed results also including other methods described in \ref{tuneD}.

\begin{table}[h]
\centering
\caption{Overview of methods of which results are presented in the main paper; see Additional file 2 for results of other methods described in \ref{tuneD}. LOOCV, leave-one-out cross-validation; AIC, Akaike's information criterium; MSE, mean squared error; CLIP, combined likelihood profile.
\label{tabMethods}}
\scriptsize
\begin{tabular}{lccccc}
  \hline
  Acronym & Section & Tuning criterion & Point estimate & Interval estimation & R package  \\\hline
  FC & (\ref{firth}) & - & Maximizing (\ref{modFC}) & \begin{tabular}[x]{@{}c@{}}Profile penalized \\ likelihood \end{tabular} & \textbf{logistf} \\ [0.5cm]
  \textit{D} & (\ref{tuneD}) & \begin{tabular}[x]{@{}c@{}} LOOCV \\ deviance (\ref{dev}) \end{tabular} & Maximizing (\ref{modRR}) & \begin{tabular}[x]{@{}c@{}}Profile penalized \\ likelihood at $\lambda^*$ \end{tabular} & \begin{tabular}[x]{@{}c@{}}\textbf{penalized} for selection of $\lambda^*$, \\ \textbf{logistf} for estimation \\ via data augmentation  \end{tabular} \\ [0.5cm]
  \textit{A} & (\ref{tuneD}) & AIC (\ref{aic}) & Maximizing (\ref{modRR}) & \begin{tabular}[x]{@{}c@{}}Profile penalized \\ likelihood at $\lambda^*$ \end{tabular} & \begin{tabular}[x]{@{}c@{}}\textbf{logistf} for estimation \\ via data augmentation \end{tabular} \\ [0.5cm]
  \textit{B} & (\ref{tune}) &  MSE (\ref{est.mse}) & Maximizing (\ref{modRR}) & \begin{tabular}[x]{@{}c@{}} According to \\ CLIP \citep{GeorgCLIP} \end{tabular} & \begin{tabular}[x]{@{}c@{}}\textbf{logistf} for estimation \\ via data augmentation  \end{tabular} \\
\hline
\end{tabular}
\end{table}

All RR models were computed after first standardizing covariates. For RR in combination with \textit{D}- (with leave-one-out CV, LOOCV) and \textit{A}-tuning we applied standardization to zero mean and unit variance. The tuning paramter was selected from a fixed grid of 200 log-linearly equidistant values ranging from $10^{(-6)}$ to $20$. For \textit{B}-tuning we applied standardization as recommended by \citet{gelman2008}. We generated $B=100$ bootstrap resamples to obtain the point estimates and CI. The target function was defined as $\mbox{trace}(\mathbf{I}\widehat{\mathbf{m}}(\lambda))$. The unknown $\beta$ was replaced by the FC estimates. 
To evaluate the impact of using an estimate of $\beta$, we also considered the optimization of $\widehat{\mathbf{m}}$ by plugging in the true coefficients instead. We abbreviate this as \textit{T}-tuned RR.

Separation was detected by an algorithm \citep{safeBinaryRegression} implemented in the \textbf{brglm2} package \citep{brglm2}. R version 3.4.3 \citep{R} with the packages \textbf{penalized} \citep{penalized} and \textbf{logistf} \citep{logistf} was used for computations (Table ~\ref{tabMethods}).

\textbf{Estimands:} The true regression coefficients $\beta_1$ and $\beta_2$ were the estimands in our investigation.

\textbf{Performance measures:} We evaluated the root mean square error (RMSE) of coefficient estimates as well as the coverage, power and the width of the 95\% CI.

\subsection{Results}
\label{res}

In this section we report the simulation study results by means of nested loop plots \citep{Nested, loopR}. First, we describe the RMSE of the point estimates (see Additional file 2 for results on bias and variance of the point estimates). Second, we discuss the coverage, power and the width of the two-sided interval estimators (the coverage of the one-sided CI of FC and \textit{B}-tuned RR is shown in Additional file 2).
We focus here mainly on the results for $\beta_1$, the effect estimates of the fairly unbalanced covariate $X_1$ $(\mbox{E}(X_1)=0.8)$, causing separation in a combination with rare outcome $\mbox{E}(Y)=0.1$. The prevalence of separation was higher in scenarios with smaller sample sizes and larger effects of $\beta_1$, and it varied from 0.2\% to 82.2\%. Figure ~\ref{fig1Res} shows the prevalence of separation over all simulated scenarios.



\subsubsection{Point estimation}
\label{est}

Figure ~\ref{fig1Res} shows the RMSE of $\beta_1$ estimated by the methods described in Table ~\ref{tabMethods}. The results obtained by \textit{D}- and \textit{A}-tuned RR were poor in some scenarios and their RMSE increased with the prevalence of separation. The performance of \textit{A}-tuned RR worsened if noise covariates $X_3,\ldots,X_{7}$ were included in the model. Interestingly, the performance of \textit{D}-tuned RR improved if including noise, as this caused more shrinkage, which reduced the variance. In contrast, although not necessarily always better than \textit{D}- and \textit{A}-tuned RR, the estimates obtained by FC were reliable across all simulated scenarios even if the prevalence of separation was high. Generally, the RMSE of FC increased with $K$, if including noise and with larger effects of $\beta_1$, and it decreased with larger $N$. Similar to FC, \textit{B}-tuned RR was robust to separation; its RMSE increased with larger effects of $\beta_1$  and decreased with larger $N$. Moreover, in scenarios of size $N=80$ with the worst performance of FC the RMSE of \textit{B}-tuned RR was acceptable and smaller compared to FC. This was due to a sligtly larger bias towards zero and larger variance of $\hat{\beta_1}$ obtained by FC (see Additional file 2). However, in scenarios with $N=200$ or $N=500$ FC resulted in a smaller RMSE; although the bias of $\hat{\beta_1}$ obtained by both approaches was very similar and close to zero, the variance of $\hat{\beta_1}$ obtained by FC was slightly smaller. Applying the hypothetical \textit{T}-tuned RR additionally reduced the RMSE of $\beta_1$ compared to \textit{B}-tuning and generally yielded the best performance. 

\begin{figure}[h!]
\centering
\includegraphics[width=0.9\textwidth]{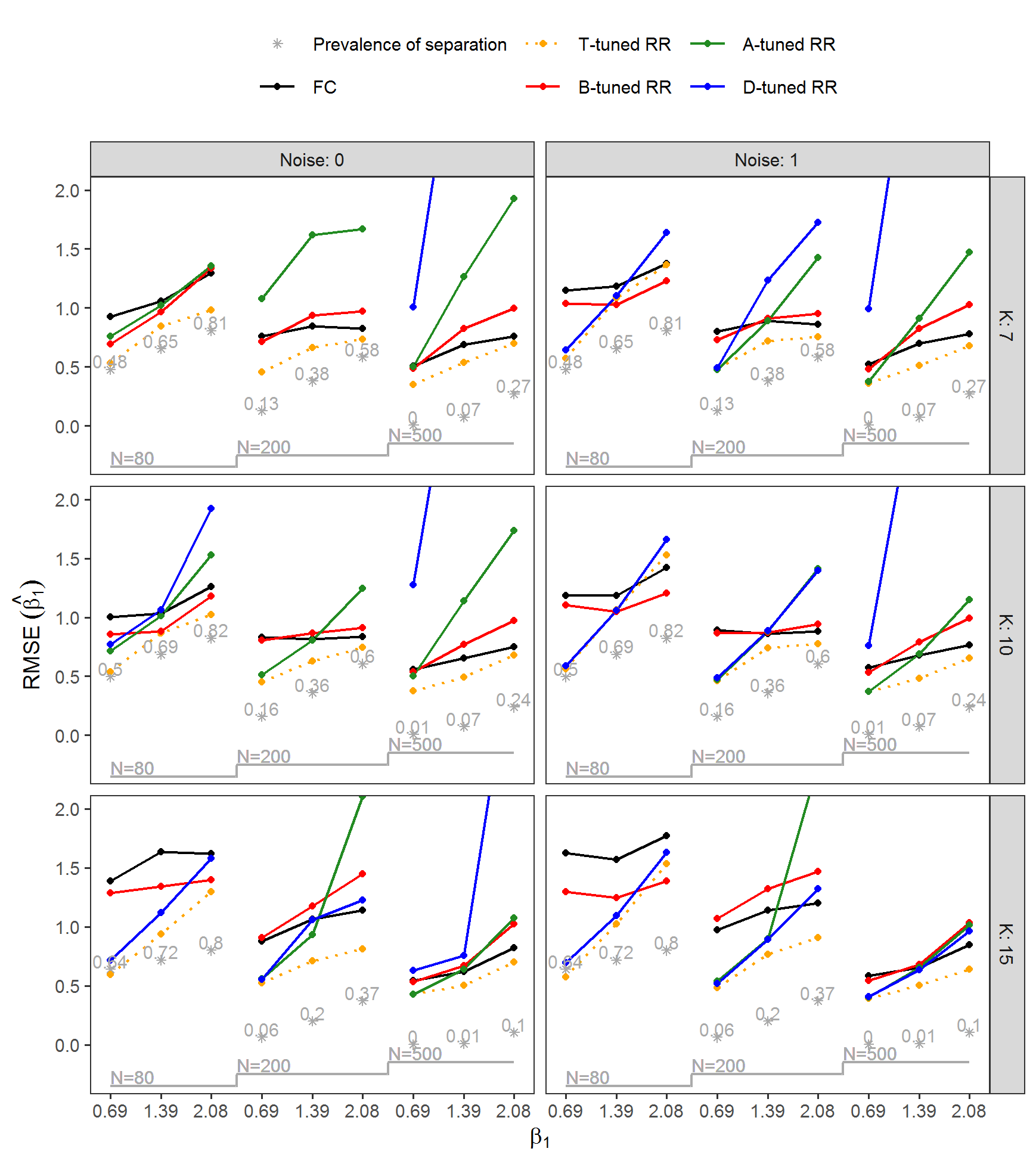}
  \caption{Root mean squared error of $\beta_1$. Nested loop plot of root mean squared error (RMSE) of $\beta_1$ by the value of $\beta_1\in\{0.69, 1.39, 2.08\}$, the sample size $N\in \{80, 200, 500\}$ and the number of covariates included into data-generating process $K\in\{7, 10, 15\}$ for all simulated scenarios. Two models were fitted by either excluding or including noise covariates $X_3,\ldots,X_{7}$ and the model parameters were obtained by different estimators. FC, Firth's correction; RR, ridge regression optimized by \textit{B}-tuning; \textit{D}-tuning, leave-one-out-cross-validated deviance; \textit{A}-tuning, Akaike's information criterium.}
  \label{fig1Res}
\end{figure}

In terms of RMSE of $\beta_2$ \textit{D}-tuned RR generally yielded best results, and for $N\neq 80$ even outperformed \textit{T}-tuned RR. The performance of \textit{D}-tuned RR improved if noise was added to the model. \textit{A}-tuning generally performed very similarly to \textit{D}-tuning except in scenarios with small samples and many covariates included, e.g. in all noise scenarios with $N=80$, where the RMSE of \textit{A}-tuned RR was large. Across all simulation scenarios the RMSE of $\beta_2$ obtained by FC and \textit{B}-tuned RR was comparable but slightly larger than the RMSE of \textit{D}- and \textit{A}-tuned RR.

\subsubsection{Interval estimation}
\label{ci}

The empirical coverage levels and the median width of the two-sided 95\% CI for $\beta_1$ are shown in Figure ~\ref{fig2Res} for FC and \textit{B}-tuned RR only as the two-sided coverage for other methods based on a profile likelihood CI was very unreliable, ranging from 26.6\% to 98.4\% for \textit{D}- and 19\% to 98.3\% for \textit{A}-tuned RR. Although FC outperformed \textit{B}-tuned RR, the coverage of the two-sided CI of \textit{B}-tuned RR for $\beta_1$ was satisfactory and always relatively close to the nominal level. The coverage rates of profile penalized likelihood CI according to FC were outside the 'plausible' interval, defined as $0.95\pm 2.57\sqrt(0.95*0.05/1000)$, in 11 out of 54 fitted models, always being too conservative with a maximum coverage of 97.9\%. The coverages of \textit{B}-tuned RR were outside the plausible range half of the time, in particular, 16-times being too conservative with maximum coverage of 98.7\% and 11-times below the nominal level with minimum coverage of 89.7\%. The median coverage (and the median width) over all scenarios was 96\% (4.69) for FC and 95.2\% (4.11) for \textit{B}-tuned RR if noise covariates were excluded from the fitting process, and 96.1\% (4.8) for FC and 95.8\% (4.59) for \textit{B}-tuned RR if including noise. Profile likelihood CI for FC were usually wider than those for \textit{B}-tuned RR. However, in some scenarios, especially with many covariates included in the model and with small sample sizes, CI for \textit{B}-tuned RR were substantially wider than the profile likelihood CI of FC due to a large upper limit. Briefly, one-sided 97.5\% CI for both FC and \textit{B}-tuned RR revealed overcoverage of the lower limit (extending to $-\infty$) and undercoverage of the upper limit (extending to $+\infty$), which was more apparent for \textit{B}-tuned RR (see Additional file 2). In terms of power to detect $\beta_1 \neq 0$ both methods performed very similarly; the median power was 23.3\% for FC, ranging from 0.9\% to 93\%, and 22.6\% for \textit{B}-tuned RR, ranging from 0.6\% to 92.7\%.

\begin{figure}[h!]
\centering
\includegraphics[width=0.9\textwidth]{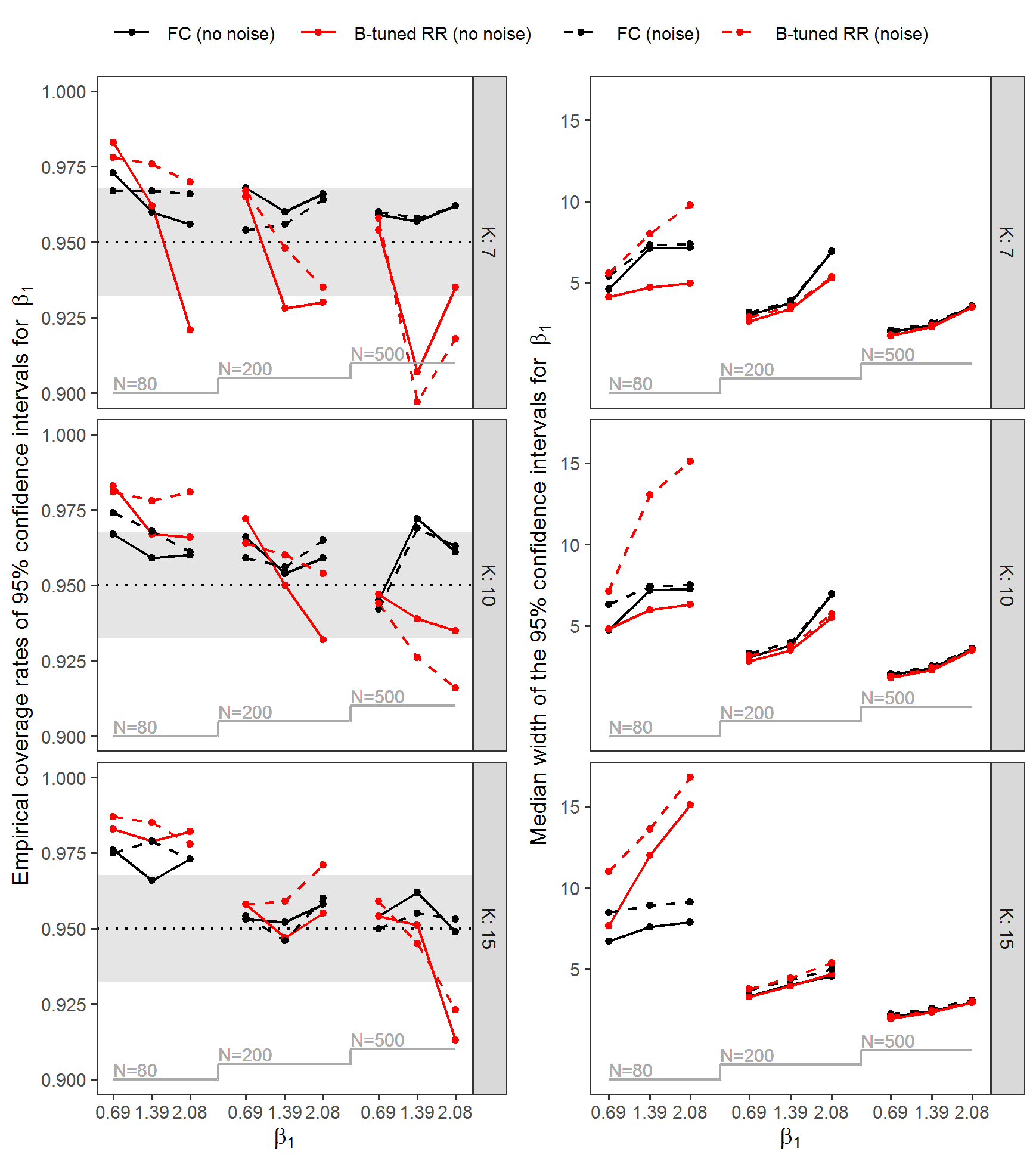}
  \caption{Coverage and median width of 95\% confidence intervals for $\beta_1$. Nested loop plot of empirical coverage and median width of the 95\% confidence intervals (CI) of $\beta_1$ by the value of $\beta_1\in\{0.69, 1.39, 2.08\}$, the sample size $N\in\{80, 200, 500\}$ and the number of covariates included into data-generating process $K\in\{7, 10, 15\}$ for all simulated scenarios. Two models were fitted by either excluding or including noise covariates $X_3,\ldots,X_{7}$. FC, Firth's correction; \textit{B}-tuned RR, tuned ridge regression where CI were obtained by combining $B=100$ penalized likelihood profiles.}
 \label{fig2Res}
\end{figure}

Although profile likelihood CI for $\beta_1$ by \textit{D}- and \textit{A}-tuned RR did not reach their nominal coverage levels, the coverage of 95\% CI for $\beta_2$ was satisfactory for all methods. FC yielded the best results, being slightly too conservative and outside the plausible range in three scenarios with a maximum coverage of 97.8\%. The CI of \textit{B}-tuned RR were too conservative in 6 cases with maximum coverage of 97\% and once fell below the nominal level with a minimum coverage of 93\%. One-sided 97.5\% CI for FC and \textit{B}-tuned RR were also satisfactory for both limits across all simulation scenarios (see Additional file 2). The coverage of \textit{D}-tuning (\textit{A}-tuning) was outside the plausible range in 9 (16) cases, and ranging from 90.9\% (86.6\%) to 97.1\% (96.4\%). The median widths of 95\% CI were comparable for FC, \textit{B}-tuned RR, ranging from 1.2 to 5.1 for FC, 1.2 to 6.7 for \textit{B}-tuned and from 1.2 to 6.2 for \textit{A}-tuned RR. The profile likelihood CI of \textit{D}- and \textit{A}-tuned RR were in median narrower, ranging from 1 to 2.8 and 1 to 2.7, respectively. However, considering the mean rather than the median, the width for \textit{D}- and \textit{A}-tuned RR could also be extreme, e.g. 32.5 and 30, respectively.

\section{Example}
\label{real}

As an example we consider a study of 79 primarily diagnosed cases of endometrial cancer with 30 and 49 patients being classified by histological grading (HG) as low or high, respectively \citep{Heinze02}. Three covariates, neovascularization (NV, binary), pulsality index of arteria uterina (PI, continuous) and endometrium height (EH, continuous) were used to explain the outcome. Quasi-complete separation occured in the data as all the patients with NV present had high HG.

In the analysis we fitted both a univariable model with NV only and a multivariable model with NV, PI and EH. The point estimates and two-sided 95\% CI are presented in Table ~\ref{tab.real} for FC and \textit{D}-, \textit{E}-, \textit{C}-, \textit{A}- and \textit{B}-tuned RR. For \textit{D}-, \textit{E}-, \textit{C}- and \textit{A}-tuned RR  we standardized the covariates to unit variance. For \textit{B}-tuned RR we used standardization recommended by \citet{gelman2008}, $B=250$ and $\mathbf{B}=\mathbf{I}$; CI were pooled from $B=250$ penalized likelihood profiles. For all the other tuning approaches we report in Table ~\ref{tab.real} the na{\"i}ve 95\% profile penalized likelihood CI. In addition, the optimal tuning parameter values are also shown in Table ~\ref{tab.real}.

\begin{table}[h]
\centering
\caption{Real data example: log odds ratio estimates $\hat{\beta}$ and two-sided 95\% confidence intervals (CI) for the covariates of endometrial cancer study obtained from the univariable models with NV only and multivariable models with NV, PI and EH by applying either Firth's correction (FC) or ridge regression (RR) optimized by \textit{B}-tuning, leave-one-out-cross-validated deviance (\textit{D}), mean squared prediction error (\textit{E}) and mean classification error (\textit{C}), and by Akaike's information criterion (\textit{A}). In addition, optimized tuning parameter values $\lambda$ are also shown where applicable.
\label{tab.real}}
\scriptsize
\begin{tabular}{lcccccc}
\hline
\multicolumn{1}{l}{Method}  & \multicolumn{2}{c}{Univariable} & \multicolumn{4}{c}{Multivariable}\\
  \cmidrule(lr){2-3}\cmidrule(lr){4-7}
 & $\lambda$ & $\hat{\beta}_{NV}$ (95\% CI) & $\lambda$  & $\hat{\beta}_{NV}$ (95\% CI) & $\hat{\beta}_{PI}$ (95\% CI) & $\hat{\beta}_{EH}$ (95\% CI) \\
  \hline
FC & - & 4.34 (2.23, 9.21) & - & 2.93 (0.61, 7.85) & -0.03 (-0.12, 0.04) & -2.6 (-4.37, -1.23) \\
  \textit{B}-tuned RR & 0.51 & 4.68 (2.33, 9.48) & 0.1 & 3.36 (0.68, 11.01) & -0.03 (-0.12, 0.04) & -2.66 (-4.4, -1.31) \\
  \textit{D}-tuned RR & 0 & - & 1.55 & 2.79 (0.84, 5.56) & -0.02 (-0.09, 0.04) & -2.14 (-3.39, -1.05) \\
  \textit{E}-tuned RR & 0.02 & 7.51 (2.78, 40.11) & 0 & - & -0.04 (-0.14, 0.04) & -2.9 (-4.79, -1.44) \\
  \textit{C}-tuned RR & 0 & - & 0.41 & 3.8 (1.02, 9.48) & -0.03 (-0.12, 0.04) & -2.62 (-4.23, -1.3) \\
  \textit{A}-tuned RR & 0.03 & 7.23 (2.76, 34.01) & 0.46 & 3.71 (1.01, 9.04) & -0.03 (-0.11, 0.04) & -2.59 (-4.18, -1.29) \\
   \hline
\end{tabular}
\end{table}

When fitting the univariable model we observe that existing tuning approaches resulted in no or only little shrinkage, yielding $\hat{\beta}_{NV}$ that either diverged to $\infty$ (\textit{D}-tuning) or was very large (odds ratio estimate exceeding 1000) with extremely wide 95\% penalized profile likelihood CI (\textit{A}- and \textit{E}-tuning). In contrast, the point and interval estimates obtained by \textit{B}-tuned RR were similar to those of FC. \textit{C}-tuned RR had no unique optimum in the interval $\lambda=[0, 20]$, neither for the univariable nor for the multivariable model, and hence the smallest optimal $\lambda$ was considered. In the univariable model the optimal tuning parameters covered the whole range of the interval, while for the multivariable model zero was not included in the range of optimal values. For the multivariable model, \textit{E}-tuning resulted in $\lambda=0$, therefore, providing no shrinkage, while \textit{B}-, \textit{D}- and \textit{A}-tuning yielded reasonable point estimates comparable to those obtained by FC. In contrast to \textit{D}- and \textit{A}-tuning, however, plausible CI were obtained by \textit{B}-tuning that could adequatly account for additional variability that came from specifying $\lambda$ from the data. Compared to the penalized likelihood profile CI of FC the CI of \textit{B}-tuned RR had a larger upper limit.

\section{Discussion}
\label{disc}

The present paper considers tuned RR as an alternative solution to separation in logistic regression and compares its performance to FC. It shows that tuning RR by the established approaches, i.e., by using criteria that depend on CV estimates of the prediction error or the Akaike information criterion, can be problematic in the presence of separation; tuning might result in an optimized value of zero and hence provide no shrinkage. Intuitively, in case of separation ML yields seemingly perfect predictions for some or for all observations. Therefore, tuning based on minimizing the (out-of-sample) prediction error might not improve predictive accuracy compared to ML. For a saturated model with categorical predictors we derived the conditions under which a tuning parameter equals zero; it can be shown that with complete separation the optimal value of the penalty parameter is always zero, but it can be non-zero with quasi-complete separation. Though obtaining a tuning parameter equal to zero is not necessarily problematic in terms of prediction, valid inference cannot be obtained since at least one of the coefficient estimates does not converge at the tuned solution. For this reason we derived a new tuning criterion targeting the MSE directly. We demonstrated that this criterion is robust to separation and that it provides parameter estimates with reasonable MSE even in the presence of complete separation. Moreover, the criterion can be optimized by using numerical methods and its flexibility allows a researcher to specify only particular coefficients of which MSE should be optimized. In agreement with \citet{sullivan13}, we advise against penalizing the intercept to obtain an average predicted probability that equals an observed event rate. A drawback of the new approach however is, that it requires bootstrapping and is thus computationally expensive.

The results from a simulation study confirmed our theoretical derivations also for non-saturated models and showed that using CV-deviance, mean squared prediction error and mean classification error or the Akaike information criterion can be problematic in the presence of separation. Mean classification error is anyhow a questionable criterion as it rarely provides a unique solution for the tuning parameter. For a saturated model in case of quasi-complete sepration CV mean squared prediction error has an advantage over deviance as it cannot result in a tuned solution of zero; however, in realistic settings it tends to overshrink. Further on, the criteria most commonly used in practice, i.e. the CV deviance and the Akaike information criterion, are also not reliable if separation is present in the data; both may yield extremly poor performance as coefficient estimates obtained for a covariate causing separation often have a very large RMSE. Software, e.g. \textbf{glmnet} \citep{glmnet}, may automatically specify a grid of tuning parameter values to search for optimal $\lambda$, and a danger is that a researcher overlooks that the optimal value is simply the smallest $\lambda$ from that grid. Interestingly, however, including noise covariates into the model improves the performance of CV deviance, suggesting a justified use of this criterion in a high-dimensional setting. Counter to the established criteria, \textit{B}-tuning based on minimizing the estimated MSE is robust to separation and yields coefficient parameter estimates with reasonable RMSE even in situations where the prevalence of separation is high. In addition, it may even outperform FC in scenarios with small sample sizes. Our results when using true coefficients to obtain $\widehat{\mathbf{m}}(\lambda)$ showed that the approach is still not optimal and plugging in another estimator instead of FC could additionally reduce the MSE of regression coefficients. In contrast to FC, however, RR lacks the invariance property, and an appropriate standardization of the covariates is required prior to tuning to assure that regression coefficients are comparable estimates of the impact of a covariate on the prediction. For \textit{B}-tuning we used standardization suggested by  \citet{gelman2008} as this improved the performance of this approach in comparison to zero mean and unit variance standardization.

Another important aspect regarding inference is interval estimation. In RR the estimation of CI is problematic; shrinkage potentially reduces the MSE of the ML but introduces bias and additional variability that comes from estimating the tuning parameter \citep{rok2020}. It was shown previously that na{\"i}ve approaches, e.g. Wald CI, do not reach their nominal coverage levels and could, therefore, not be recommended  \citep{flicflac}. In our study we na{\"i}vely estimated CI based on penalized likelihood profiles and could observe poor coverage for the binary covariate causing separation. Interestingly, the two-sided coverage for the other binary covariate was on average relatively close to the nominal level.

An alternative to tuning would be pre-specifying the value of the tuning parameter; \citet{sullivan13} suggest to choose the tuning parameter according to a plausible 95\% prior interval for a covariate's odds ratio. The width of such an interval on the log scale can be converted into a variance of a normal prior, and the penalty parameter in RR then corresponds to the reciprocal of the normal prior variance. In their approach, standardization of covariates is not considered, and different regression coefficients may have different prior variances and hence different penalty parameters. Alternatively, the penalty strength could be determined by fixing the effective degrees of freedom that one would like associated with the model fit, and solving for penalty parameter. \citet{Harrell} suggests to restrict the effective degrees of freedom to $m/15$ where $m$ is the frequency of the less frequent level of the outcome. 

For \textit{B}-tuned RR we demonstrated how to obtain interval estimates by combining penalized likelihood profiles obtained on bootstrap resamples \citep{GeorgCLIP}. Similarly as in FC, in situations where separation occurs penalized likelihood profile becomes highly asymmetric, therefore, potentially providing satisfactory coverage that is on average close to the nominal level \citep{Heinze02}. Using bootstrap resamples accounts for the additional source of variability that comes from estimating tuning parameter from the data. Moreover, it reduces the variance of tuning parameter, therefore, reducing the variance of penalized coefficient estimates. Our simulation study suggests that the obtained CI have satisfactory two-sided coverage even for the covariate causing separation. Despite the asymmetry of the CI that might yield a large upper limit, the one-sided coverage usually results in overcoverage of the lower limit and slight undercoverage of the upper limit. The CI are mostly of comparable width as the profile penalized likelihood CI obtained by FC, although they may also be wider due to a large upper limit. Nonetheless, the power of both approaches is comparable. One might also consider to obtain a point estimate by using the mode of the combined profile likelihood function. Our results suggest, however, that this results in a point estimate with a slightly larger RMSE.

All the proposed approaches can be easily implemented in any standard statistical software and the R-code is available from a github repository \citep{RRcode}. An example taken from a biomedical context illustrates practical application of the novel approach in comparison to the established methods. In this example \textit{B}-tuned RR yields plausible coefficient parameter estimates and  CI. This extensive investigation shows there is a small gain of \textit{B}-tuned RR over FC regarding the MSE of coefficient estimates in the analysis of small and sparse data sets prone to separation, however, at the cost of computational burden. FC, nevertheless, is transformation-invariant and mostly provides coefficient estimates with acceptable MSE, although its performance decreases in small sample size settings with many highly correlated covariates and included noise.

\section{Conclusions}
\label{con}

In logistic regression analyses of small and sparse data sets established criteria for tuning the amount of penalization in RR should be applied with caution as these measures can yield an optimized value of zero (no shrinkage) or a value which is the smallest from an arbitrary grid of values, and hence cannot be considered as a universal solution to separation. In contrast, our new tuning criterion \textit{B} leads to shrinkage even in case of separation. Moreover, by bootstrap resampling the computation of confidence intervals is facilitated which accounts for the variability that arises from the determination of a tuning parameter. If a data analyst is willing to accept additional computational burden, \textit{B}-tuned RR can be considered as an alternative to FC as a solution to separation in logistic regression having attractive properties in small and sparse samples with many highly correlated covariates.

\section{Software}
\label{soft}

Software in the form of R code, together with a sample input data set is available from a github repository \citep{RRcode}.

\section*{Acknowledgements}

We thank Nina Ru\v{z}i\v{c} for her help and comments on an earlier draft of this paper.

\section*{Funding}

This research was funded by the Austrian Science Fund (FWF), grant number I 2276, and the Slovenian Research Agency (Predicting rare events more accurately, N1-0035; Methodology for data analysis in medical sciences, P3-0154).

\bibliography{bibliography}

\pagebreak
\begin{flushleft}
\textbf{\LARGE Supplementary material for: Tuning in ridge logistic regression to solve separation}

\vspace{5mm}

\textbf{\LARGE Additional file 1: Theoretical results on the convergence of tuned ridge regression}

\end{flushleft}

\setcounter{equation}{0}
\setcounter{figure}{0}
\setcounter{table}{0}
\setcounter{page}{1}
\setcounter{section}{0}
\makeatletter
\renewcommand{\theequation}{S\arabic{equation}}
\renewcommand{\thefigure}{S\arabic{figure}}
\renewcommand{\thetable}{S\arabic{table}}
\renewcommand{\bibnumfmt}[1]{[S#1]}
\renewcommand{\citenumfont}[1]{S#1}

\section{Proof of Theorem 1}

According to Theorem 1 in Theobald \citep{S_RefA}, we need to show that there exists some $\lambda>0$ such that $\widehat{\mathbf{m}}(0)-\widehat{\mathbf{m}}(\lambda)$ is positive definite. Simple calculation yields
$$\widehat{\mathbf{m}}(0)-\widehat{\mathbf{m}}(\lambda)=\lambda(I(\bm\tilde\beta)+\lambda\mathbf{I})^{-1}\left[ 2\mathbf{I}+\lambda I(\bm\tilde\beta)^{-1} -\lambda \bm{\tilde\beta}\bm{\tilde\beta}^T \right](I(\bm\tilde\beta)+\lambda\mathbf{I})^{-1},  $$
$\lambda>0$. By condition (1) we need to show that there is some $\lambda>0$ such that $2\mathbf{I}-\lambda \bm{\tilde\beta}\bm{\tilde\beta}^T$ is positive definite, since (1) implies that $I(\bm\tilde\beta)$ is positive definite. Now since $K$ eigenvalues of $\lambda \bm{\tilde\beta}\bm{\tilde\beta}^T$ are zero and one is $\lambda \bm{\tilde\beta}^T\bm{\tilde\beta}$ the condition is $2-\lambda \bm{\tilde\beta}^T\bm{\tilde\beta}>0 $. By condition (1), there exists some $\lambda>0$ such that the condition holds, thus proving the theorem.

\section{Tuning ridge regression in the presence of separation}

Next we will derive theoretical conditions under which the optimal penalty parameter will be non-zero. We consider a particular model, i.e. a saturated model with categorical predictors. A motivating example is given first. This is then followed by an informal discussion of the existing criteria for tuning the penalty parameter in the presence of \textit{separation}. Finally, formal statements and proofs are given in the end.

\section{A motivating example}

As a motivating example we consider a data set available within R's \citep{S_RefB} \textbf{logistf} package \citep{S_RefC}. The data consists of 130 sexually active college women that suffered from urinary tract infection (UTI) and 109 controls, together with the covariate information on the use of contraceptives (yes and no) and age ($\geq 24$, $<24$ years).

First, consider estimating the log-odds for different categories of reported use of diaphragm (DIA) by using the model defined in Eq. (1) in the main text omitting the intercept - Model 1.  Only 7 women (3\%) reported the use of diaphragm (DIA) and all of them suffered from UTI, hence quasi-complete separation occurs as a consequence of the combination of rare exposure and strong association to the outcome. Thus, the parameter estimates cannot be obtained by ML. To estimate the parameters by RR (using the R's \textbf{penalized} package \citep{S_RefD}), the unknown penalty parameter needs to be chosen. The optimization of the LOOCV deviance for a grid of 200 log-linearly equidistant penalty parameter values ranging from $10^{-6}$ (-13.82 on the log scale) to 10 (2.3 on the log scale) yields a penalty parameter that equals the smallest value of the pre-specified grid. No matter how small the minimum value of the pre-specified grid, the LOOCV deviance is always optimized at the smallest value, i.e. if allowed at zero, which results in non-existing parameter estimate for a group of women using DIA. In contrast, optimizing the penalty parameter by LOOCV mean prediction error yields a value larger than the smallest value of a pre-specified grid, i.e. larger than zero. Still, the amount of penalization is small, and the parameter estimates are much larger than those obtained by FP (Table \ref{tab.mot1}).

Second, the goal is to estimate the log-odds for each combination of age group ($\geq 24$, $<24$ years), use of lubricated condom (LC) and oral contraceptive (OC) by again using the model defined in Eq. (1) omitting the intercept - Model 2. Even if there is no \textit{separation} when estimating univariable models, quasi-complete separation occurs in the fully saturated model, since all women with 24 years of age or older who did not use UCL and OC suffered from UTI. Like in the previous example, the LOOCV deviance is minimized at the smallest value of the grid, while LOOCV mean prediction error is not. 

In both examples tuned ridge regression when using the LOOCV deviance is, in terms of obtaining finite parameter estimates, not a solution for the \textit{separation}, while it is when using the LOOCV mean prediction error. However, when using the later criterion, the point estimates are (unreasonably) large. 
In the next section we present theoretically why this occurs in this particular model.

\begin{table}[ht]
\caption{Coefficient estimates for the urinary tract infection data obtained by ML, FP and RR where the optimal penalty parameter ($\lambda$) was obtained by minimizing LOOCV deviance (D) or LOOCV squared prediction error (PE).}
\label{tab.mot1}
\setlength{\tabcolsep}{6pt}
\centering
 \begin{threeparttable}
\begin{tabular}{llrrrr}
  \hline
& & ML & FP & \multicolumn{2}{c}{RR} \\
& &   &   &  D &  PE \\
  \hline
Model 1&  &  &  & -13.82\tnote{b} & -7.66\tnote{b} \\
 \multirow{2}{*}{DIA} &no  & 0.12 & 0.12 & 0.12 & 0.12 \\
  &yes &  -\tnote{a} & 2.71 & 13.18 & 7.58 \\\hline
 Model 2 & &  &  & -13.82\tnote{b} & -3.53\tnote{b} \\
 \multirow{8}{*}{age:LC:OC} &$<24$:no:no&  1.50 & 1.41 & 1.50 & 1.49 \\
  &$\geq24$:no:no & -\tnote{a} & 2.20 & 12.66 & 3.61 \\
  &$<24$:yes:no & -0.12 & -0.12 & -0.12 & -0.12 \\
  &$\geq24$:yes:no & -1.10 & -0.96 & -1.10 & -1.08 \\
  &$<24$:no:yes & 0.39 & 0.39 & 0.39 & 0.39 \\
  &$\geq24$:no:yes & -0.69 & -0.66 & -0.69 & -0.69 \\
  &$<24$:yes:yes & -0.10 & -0.09 & -0.10 & -0.09 \\
  &$\geq24$:yes:yes & -\tnote{a} & -1.61 & -12.02 & -3.06 \\
   \hline

\end{tabular}
\begin{tablenotes}
\item[a] the estimate does not exist \item[b] optimal penalty parameter (log scale)
\end{tablenotes}
\end{threeparttable}
\end{table}

\section{When is tuned ridge regression a potential solution when using existing criterion functions?}
\label{sec1}
In this section we show under which conditions could tuned ridge regression be considered a potential solution to \textit{separation}, when applied to a model, which was considered in a previous section. First, we introduce some notation. Assume a binary outcome $Y=\{0,1\}$ and a categorical covariate $X$ with $K$ categories, $X=\{1,\ldots,K\}$. $X$ can also be defined from a set of categorical covariates $X_j=\{1,\ldots,K_j\}$, $j=1,\ldots,L$, by joining all possible combinations of their levels, which yields $K=\prod_{j=1}^L K_j$ (fully saturated model). We observe $n$ independent realizations of pairs $(Y_i,X_i)$, denoted as $(y_i,x_i)$, where $Y_i\sim B(\pi_i)$ and $\pi_i=P(Y_i=1|X_i)$, where $B(\cdot)$ denotes Bernoulli distribution. The data can be summarized in the following $K\times 2$ contingency table,
\begin{center}\begin{tabular}{ll|cc|l}
&&\multicolumn{2}{c|}{$Y$}&\\
&&0&1&\\\hline
\multirow{ 4}{*}{$X$}&1& $a_{10}$&$a_{11}$&$a_{1\bullet}$\\
&$\ldots$&$\ldots$&$\ldots$&$\ldots$\\
&$k$&$a_{k0}$&$a_{k1}$&$a_{k\bullet}$\\
&$\ldots$&$\ldots$&$\ldots$&$\ldots$\\
&$K$&$a_{K0}$&$a_{K1}$&$a_{K\bullet}$\\\hline
&&$a_{\bullet 0}$&$a_{\bullet 1}$&$n$
\end{tabular}\end{center}

Then the model from the previous section is,
\begin{equation}
\label{mod1}
\log\frac{\pi_i}{1-\pi_i}=\beta_1 z_{i1}+\ldots+ \beta_K z_{iK}\mbox{, }i=1,\ldots,n,
\end{equation}
where $z_{ik}$ denotes the dummy covariate associated with the $k$th category, $k=1,\ldots,K$, i.e. $z_{ik}=1$ if $x_i=k$ and $z_{ik}=0$ otherwise. Maximum likelihood estimates are easily obtained and are given by
\begin{equation*}
\label{pks.1}
\hat\beta_k=\log\left[\frac{a_{k1}}{a_{k0}}\right]\mbox{, }k=1,\ldots,K
\end{equation*}
Observe that if any $a_{k1}$ or $a_{k0}$ are zero, the estimate is not defined, which is a well known consequence of \textit{separation} \citep{S_RefE, S_RefF}. To obtain finite estimates of $\beta_k$ even when $a_{k1}$ or $a_{k0}$ are zero (but not both simultaneously) one can for example apply the following ridge-type penalty
$$P(\lambda)=\frac{\lambda}{2}\sum_{k=1}^K \beta_k^2,$$
for some $\lambda\geq 0$. We show later that the approximate penalized estimates are then $\hat\beta_k(\lambda)=\log(\hat\pi_k(\lambda)/(1-\hat\pi_k(\lambda)))$, where
\begin{equation*}
\label{eq.pkl}
\hat\pi_k(\lambda)=\frac{a_{k1}+2\lambda}{a_{k\bullet}+4\lambda} \mbox{, }k=1,\ldots,K.
\end{equation*}
Note that this approximate solution is equivalent to penalizing the likelihood by the generalized Jeffrey's invariant prior and that $\lambda=1/4$ is the Jeffrey's invariant prior used by \citet{Firth93}. Any $\lambda>0$ will provide finite estimates with complete or quasi-complete separation. Observe that model (\ref{mod1}) omits the intercept; some results when the model includes the intercept are given later. In what follows, we consider optimizing the penalty parameter by using (i) leave-one-out cross-validation (LOOCV) or (ii) independent validation set. While LOOCV can in practice be computationally demanding it is the only CV scheme where theoretical investigation is feasible. In all the theoretical results we assume that $a_{k0},a_{k1}>1$ for all $k=1,\ldots,K$, unless stated otherwise.

\subsection*{Tuning the deviance}

Consider optimizing the LOOCV deviance which for our setting is given by
\begin{equation}
\label{eq.dev}
D=\sum_{k=1}^K D_k=-2\sum_{k=1}^K \left(a_{k0}\log\left(1-\hat\pi_{k0}(\lambda)\right)+a_{k1}\log \hat\pi_{k1}(\lambda)   \right),
\end{equation}
where \begin{equation*}
\label{pr.mod1.0}
\hat\pi_{k0}(\lambda)= \frac{a_{k1}+2\lambda}{a_{k\bullet}-1+4\lambda} \mbox{, }k=1,\ldots,K,
\end{equation*}
\begin{equation*}
\label{pr.mod1.1}
\hat\pi_{k1}(\lambda)= \frac{a_{k1}-1+2\lambda}{a_{k\bullet}-1+4\lambda} \mbox{, }k=1,\ldots,K.
\end{equation*}

We formally show later that when there is no \textit{separation} in the data, i.e. $|a_{k0}-a_{k1}|\neq a_{k\bullet}$ for $k=1,\ldots,K$, the LOOCV deviance is always minimized at some $\lambda>0$. Furthermore, whenever the difference between $a_{k0}$ and $a_{k1}$ is not large enough for all $k=1,\ldots,K$ the LOOCV deviance is optimized when using $\lambda=\infty$. Intuitively, whenever $a_{l0}=a_{l1}$, $a_{l0},a_{l1}>1$ holds, then $\hat\pi_{l0}(\lambda)\geq 1/2$ and $\hat\pi_{l1}(\lambda)\leq 1/2$, with the equalities applying, if and only if $\lambda=\infty$. Increasing $\lambda$ pools both event probabilities towards 1/2, therefore, both terms of $D_l$, i.e. $a_{l0}\log(1-\hat\pi_{l0}(\lambda))$ and $a_{l1}\log \hat\pi_{l1}(\lambda)$, are increasing functions of $\lambda$, which implies that $D_l$ is a decreasing function of $\lambda$ and is minimized at $\lambda=\infty$. However, when $a_{l0}$ and $a_{l1}$ are different enough
, then assuming without the loss of generality that $a_{l0}<a_{l1}$, $ 1-\hat\pi_{l0}^\lambda \leq 1/2 \leq \hat\pi_{l1}^\lambda$, with the equalities applying if and only if $\lambda=\infty$. Increasing $\lambda$ pools the two event probability estimates towards 1/2, hence $\log( 1-\hat\pi_{l0}(\lambda) )$ and $\log \hat\pi_{l1}(\lambda)$ are increasing and decreasing functions of $\lambda$, respectively. As the two terms in $D_l$ have different weights the gain of penalization up to some $\lambda$ for nonevents is larger than the loss for events and $D_l$ decreases. However, since $1-\hat\pi_{l0}(\lambda)$ approaches 1/2 more quickly than $\hat\pi_{l1}(\lambda)$, what can be seen from the derivatives with respect to $\lambda$, for a large enough $\lambda$ the loss for the events becomes larger as the gain for the nonevents, and the deviance becomes an increasing function of $\lambda$. Hence, the LOOCV deviance will be optimized for some $0<\lambda<\infty$.

However, with complete separation the LOOCV deviance is minimized at $\lambda=0$, leading to the ML solution and non-existent parameter estimates. In contrast, with quasi-complete separation, the LOOCV deviance can also be minimized at $\lambda>0$, but not always. Intuitively, when there are no events or nonevents for one category $k$ of $X$, e.g. in category $k$ there are no events, $a_{k1}=0$, then the term $a_{k1}\log \hat\pi_{k1}(\lambda)$ disappears from (\ref{eq.dev}) and $D_k$ becomes an increasing function of $\lambda$ as the ML estimate of the event probability in this category of $X$ is zero; hence $\log (1-\hat\pi_{k0}))=\log 1=0$, but for $\lambda>0$, $\log(1- \hat\pi_{k0}(\lambda))$ is a decreasing function of $\lambda$ and $0<\log(1- \hat\pi_{k0}(\lambda))\leq 1/2$ holds for all $\lambda>0$. Thus, penalization will always increase $D_k$ and whenever the gain of penalization for $D_l$, $l\neq k$ is smaller than the loss for $D_k$, $D$ will be optimized at $\lambda=0$. However, the gain of penalization can also be larger than the loss and the deviance will be optimized at $\lambda>0$. This occurs, if there are many categories with a small number of subjects within which no \textit{separation} exists. In a $2\times 2$ table however, the gain of penalization can exceed the loss only in exceptional cases (the exact conditions are given later).

When an independent validation set is available, the validation deviance, $D^V(\lambda)=D^V$, is
\begin{equation*}
\label{eq.dev.test1}
D^V=-2\sum_{k=1}^K D_k^V=-2\sum_{k=1}^K \left(a_{k0}^V\log(1-\hat\pi_{k}(\lambda))+a_{k1}^V\log \hat\pi_{k}(\lambda)   \right),
\end{equation*}
where $a_{k0}^V$ and $a_{k1}^V$ are the respective number of nonevents and events for category $k$ in the validation set. We formally show later that when there is complete separation in the original data and there is no \textit{separation} in the validation set then the validation deviance is minimized at $\lambda>0$, thus leading to finite parameter estimates even with a complete separation.

\subsection*{Tuning the squared prediction error}

The LOOCV squared prediction error is for our setting given by
\begin{equation*}
\label{eq.pm}
PE(\lambda)=\frac{1}{n}\sum_{k=1}^K \left( a_{k0}(\hat\pi_{k0}(\lambda))^2+a_{k1}(1-\hat\pi_{k1}(\lambda))^2  \right) .
\end{equation*}

For this criterion we show later, that it is minimized at $\lambda=0$ if and only if there is complete separation in the data and is hence superior to LOOCV deviance in a setting with a quasi-complete separation since it will in this case, in contrast to LOOCV deviance, always yield finite parameter estimates.

The validation squared prediction error is,
\begin{equation*}
\label{eq.pe.test1}
PE^V(\lambda)=\frac{1}{n_V}\sum_{k=1}^K \left( a_{k0}^V(\hat\pi_k(\lambda))^2+a_{k1}^V(1-\hat\pi_k(\lambda))^2  \right),
\end{equation*}
where $n_V$ is the number of subjects in the validation set. We later derive a similar result as for the validation deviance: even with complete separation in the original data, the validation prediction error will be minimized at $\lambda>0$ as long as there is no  \textit{separation} in the validation set.



\section{Formal derivation}

Here we derive the estimators used in Section \ref{sec1}. Theoretical results from that section are then formally stated. The proofs are then given in the end of this Additional file. We assume in all the theoretical results that $a_{k0},a_{k1}>1$ for all $k=1,\ldots,K$, unless stated otherwise. Throughout, if some function $\phi(x)$ for $x\in(0,\infty)$ reaches its infimum at $x=0$ or $x=\infty$, we define this as $\mbox{argmin}_{x\in(0,\infty)}\phi(x) =0$ and $\mbox{argmin}_{x\in(0,\infty)}\phi(x) =\infty$, respectively.

Observe that
\begin{equation}
\label{eq1}
 z_1+\ldots+z_K=1\mbox{; }
\end{equation}
\begin{equation}
\label{eq111}
 z_{j}^Tz_{k}=0\mbox{, for any } j\neq k,
\end{equation}
where $z_k=(z_{1k},\ldots,z_{nk})^T$ and $1$ is the identity vector of order $n$. Now assume the model
\begin{equation}
\label{mod2}
\log\frac{\pi_i}{1-\pi_i}=\gamma_0+\gamma_1 z_{i1}+\ldots+ \gamma_{K-1} z_{iK-1}\mbox{, }i=1,\ldots,n.
\end{equation}

It follows from (\ref{eq1}) and (\ref{eq111}) that
\begin{equation}
\label{eq2h}
\gamma_0=\beta_K\mbox{, }\gamma_k=\beta_k-\beta_K,
\end{equation}
hence the coefficients of the model (\ref{mod2}) can easily be obtained from the coefficients of the model (\ref{mod1}). The likelihood function for both models can be expressed as
$$L=\prod_{k=1}^K \pi_{k}^{a_{k1}} (1-\pi_{k})^{a_{k0}},  $$
where
$$\pi_k=P(Y_i=1|x_i=k)=P(Y_i=1|z_{ik}=1), $$
and the natural logarithm of $L$ is then simply,
$$l=\sum_{k=1}^K a_{k1}\log(\pi_{k})+a_{k0}\log(1-\pi_k)=\sum_{k=1}^K l_k.  $$

To estimate $\pi_k$ partial derivatives of $l$ are set to zero,
$$\frac{\partial l}{\partial \pi_k}=\frac{\partial l_k}{\partial \pi_k}=\frac{ -a_{k0}\pi_k+ a_{k1}(1-\pi_k)  }{ \pi_k(1-\pi_k)}=0\mbox{, }k=1,\ldots,K. $$

The partial derivatives are separable by $\pi_k$, hence the solutions are
\begin{equation}
\label{pks}
\hat\pi_k=\frac{a_{k1}}{a_{k\bullet}}\mbox{, }k=1,\ldots,K.
\end{equation}

Coefficient estimates of model (\ref{mod1}) and (\ref{mod2}) can easily be obtained from (\ref{pks}) and are given by
\begin{equation}
\label{pks.1}
\hat\beta_k=\log\left[\frac{a_{k1}}{a_{k0}}\right]\mbox{, }k=1,\ldots,K
\end{equation}
and
\begin{equation}
\label{pks.2}
\hat\gamma_0=\log\left[\frac{a_{K1}}{a_{K0}} \right]\mbox{, }\hat\gamma_k=\log\left[\frac{a_{k1}a_{K0}}{a_{k0}a_{K1}}\right]\mbox{, }k=1,\ldots,K-1,
\end{equation}
respectively. When for some category $k$ there are either no events or no nonevents, i.e. $|a_{k0}-a_{k1}|=a_{k\bullet}$ holds for some $k$, then the following results immediately follow from (\ref{pks.1}) and (\ref{pks.2}).
\begin{corollary}
\label{color1}
Let $A$ denote the set of indices for which
$$|a_{k0}-a_{k1}|=a_{k\bullet}>0,$$
holds. Then $\hat\beta_k$ does not exist for $k\in A$.
\end{corollary}
\begin{corollary}
\label{color2}
Let $A$ denote the set of indices for which
$$|a_{k0}-a_{k1}|=a_{k\bullet}>0,$$
holds. If $K\in A$, then $\hat\gamma_k$, $k=0,\ldots,K-1$ do not exist. If $K\notin A$, then $\hat\gamma_k$ does not exist for $k\in A$.
\end{corollary}

\subsection{Ridge penalty for the $\beta$ coefficients}

Introduce the RR penalty for the coefficients $\beta_k$, $k=1,\ldots,K$, i.e.
$$P(\lambda)=\frac{\lambda}{2}\sum_{k=1}^K \beta_k^2,$$
for some $\lambda\geq 0$, or equivalently in terms of $\pi_k$, $k=1,\ldots,K$ and $\gamma_k$, $k=0,\ldots,K-1$ as
\begin{equation}
\label{pen}
 P(\lambda)=\frac{\lambda}{2}\sum_{k=1}^K \log^2\left[ \frac{\pi_k}{1-\pi_k}  \right]=\frac{\lambda}{2}\left[\gamma_0^2+ \sum_{k=1}^{K-1}(\gamma_0+\gamma_k)^2 \right].
 \end{equation}

Applying penalty (\ref{pen}) to both models obviously leads to the same solution in terms of the event probability estimation, and therefore the relation (\ref{eq2h}) still applies. The penalized log likelihood for both models, as a function of $\pi_k$, is then
$$l^P= l-P(\lambda), $$
and the partial derivatives are
$$\frac{\partial l^P}{\partial \pi_k}=\frac{\partial l_k}{\partial \pi_k}-\frac{\lambda}{\pi_k(1-\pi_k)}\log \left[ \frac{\pi_k}{1-\pi_k}  \right]  \mbox{, }k=1,\ldots,K. $$

Setting the partial derivatives to zero again yields a system of equations, which are separable by $\pi_k$. The one step Newton solutions setting the initial values to 1/2 are then
\begin{equation}
\label{eq.pkl}
\hat\pi_k(\lambda)=\frac{a_{k1}+2\lambda}{a_{k\bullet}+4\lambda} \mbox{, }k=1,\ldots,K.
\end{equation}

Note that setting $\lambda$ to zero yields the ML solution, hence for $\lambda=0$ Corollary \ref{color1} and \ref{color2} apply. For $\lambda\rightarrow \infty$
$$\lim_{\lambda\rightarrow \infty}\hat\pi_k(\lambda)=\frac{1}{2}, $$
which is the fully penalized estimate for this example. Using $\lambda=1/4$ yields FP estimator while for $K=1$ setting $\lambda=1$ gives the Agresti-Coull estimator.

\subsubsection{Tuning the LOOCV deviance}

The LOOCV deviance, $D(\lambda)=D$, is
\begin{equation}
\label{eq.dev}
D=\sum_{k=1}^K D_k=-2\sum_{k=1}^K \left(a_{k0}\log\left(1-\hat\pi_{k0}(\lambda)\right)+a_{k1}\log \hat\pi_{k1}(\lambda)   \right),
\end{equation}
where $\hat\pi_{k0}(\lambda)$ and $\hat\pi_{k1}(\lambda)$ denote the penalized maximum likelihood estimates of the event probability when omitting subject $z_{ik}=1$,  $y_i=0$ and $z_{ik}=1$, $y_i=1$, respectively, i.e.
\begin{equation}
\label{pr.mod1.0}
\hat\pi_{k0}(\lambda)= \frac{a_{k1}+2\lambda}{a_{k\bullet}-1+4\lambda} \mbox{, }k=1,\ldots,K,
\end{equation}
\begin{equation}
\label{pr.mod1.1}
\hat\pi_{k1}(\lambda)= \frac{a_{k1}-1+2\lambda}{a_{k\bullet}-1+4\lambda} \mbox{, }k=1,\ldots,K.
\end{equation}

Observe that $D$ is not defined for $\lambda=0$, if there is a single event, $a_{k1}$, or nonevent, $a_{k0}$ (or both) for any $k$. If however, there are no events or no nonevents, then $D$ is well defined also for $\lambda=0$, since in this case the term $a_{k1}\log \hat\pi_{k1}(\lambda)$ or $a_{k0}\log\left(1-\hat\pi_{k0}(\lambda)\right)$ vanishes. Let
$$\lambda_D=\mbox{argmin}_{\lambda\geq0} D(\lambda), $$
be the point at which the LOOCV deviance defined in (\ref{eq.dev}) attains its minimum.

First, assume that there is no \textit{separation} in the data, i.e. $|a_{k0}-a_{k1}|\neq a_{k\bullet}$ for $k=1,\ldots,K$. Then we have the following results.
\begin{proposition}
\label{prop.D.ns}
If $a_{k0},a_{k1}>1$ for all $k=1,\ldots,K$, i.e. there is no \textit{separation} in the data, then $\lambda_D>0$.
\end{proposition}
\begin{lemma}
\label{lemma.D.ns.1}
Assume that $a_{k0},a_{k1}>1$ and $(a_{k0}-a_{k1})^2-a_{k\bullet}> 0 $ for all $k=1,\ldots,K$. Then $0<\lambda_D<\infty$.
\end{lemma}
\begin{lemma}
\label{lemma.D.ns.2}
Assume that $a_{k0},a_{k1}>1$ and $(a_{k0}-a_{k1})^2-a_{k\bullet}\leq 0 $ for all $k=1,\ldots,K$. Then $\lambda_D=\infty$.
\end{lemma}
\begin{remark}
$a_{k0}=a_{k1}$ is a special case of $(a_{k0}-a_{k1})^2-a_{k\bullet}\leq 0 $.
\end{remark}
\begin{lemma}
\label{lemma.D.ns.1.0}
Let $K=1$ and assume that $a_{10},a_{11}>1$, $(a_{10}-a_{11})^2-a_{1\bullet}>0$, then
\begin{equation}
\label{eq.opt.l}
\lambda_D=\frac{ a_{10}(a_{10}-1)+a_{11}(a_{11}-1) }{2(( a_{10}-a_{11} )^2-a_{1\bullet}  )} .
\end{equation}
\end{lemma}
\begin{remark}
Plugging (\ref{eq.opt.l}) into (\ref{eq.pkl}) gives a closed form approximation of RR penalized proportion estimate when optimizing the LOOCV deviance.
\end{remark}

When there is \textit{separation} in the data, i.e. for some $k$, $|a_{k0}-a_{k1}|=a_{k\bullet}$ holds, the following results are obtained.
\begin{proposition}
\label{prop.D.cs}
If $|a_{k0}-a_{k1}|=a_{k\bullet}>1$ for all $k=1,\ldots,K$, i.e. there is complete separation in the data, then $\lambda_D=0$.
\end{proposition}
\begin{lemma}
\label{lemma.D.qcs.K}
Assume that there is quasi-complete separation in the data. Let $S=\{1,2,\ldots,K\}$ denote the set of all indices. Let $A$ denote the set of indices for which
$$|a_{k0}-a_{k1}|=a_{k\bullet}>1, $$
holds and let $B$ denote the set of indices for which
$$a_{k0},a_{k1}>1,$$
holds. Assume that $A\cup B=S$. If
$$\sum_{k\in B} \frac{(a_{k0}^2+a_{k1}^2-a_{k\bullet})}{(a_{k1}-1)(a_{k0}-1)(a_{k\bullet}-1)}> \sum_{k\in A}\frac{a_{k\bullet}}{a_{k\bullet}-1}$$
holds, then $\lambda_D>0$.
\end{lemma}


Note that Lemma \ref{lemma.D.qcs.K} gives only one situation where the LOOCV deviance in case of quasi-complete separation is optimized at $\lambda>0$. We summarize here for $K=2$ a few more interesting examples with quasi-complete separation: Proposition \ref{prop.1.s} explains a situation where $|a_{k0}-a_{k1}|=a_{k\bullet}$, $a_{l0}=a_{l1}$ for $l\neq k$, and Proposition \ref{prop.2.s} a situation where $|a_{k0}-a_{k1}|=a_{k\bullet}$, $a_{l0}\neq a_{l1}$ for $l\neq k$.

\begin{proposition}
\label{prop.1.s}
Let $K=2$. Assume that $|a_{k0}-a_{k1}|=a_{k\bullet}>1$, $a_{l0}=a_{l1}>1$, $l\neq k$. Then the following exhaustive and mutually exclusive situations are derived,
\begin{itemize}
\item[1.] $\lambda_D=0$ if any of the following holds
\begin{itemize}
\item[(i.)] $a_{k\bullet}>2$, $a_{l\bullet}<a_{k\bullet}(a_{k\bullet}-1)$, $a_{k\bullet}(\frac{a_{l\bullet}}{2}-1)(a_{l\bullet}-1)\geq a_{l\bullet}(a_{k\bullet}-1)$
\item[(ii.)] $a_{k\bullet}>2$, $a_{l\bullet}\geq a_{k\bullet}(a_{k\bullet}-1)$
\item[(iii.)] $a_{k\bullet}=2$, $a_{l\bullet}>4$
\end{itemize}
\item[2.] $\lambda_D=\infty$ if $a_{k\bullet}=2$, $a_{l\bullet}=4$
 \item[3.] $0<\lambda_D<\infty$ if $a_{k\bullet}>2$, $a_{l\bullet}<a_{k\bullet}(a_{k\bullet}-1)$, $a_{k\bullet}(\frac{a_{l\bullet}}{2}-1)(a_{l\bullet}-1)< a_{l\bullet}(a_{k\bullet}-1)$.
\end{itemize}
\end{proposition}
\begin{remark}
Case 3. can occur if and only if $a_{l\bullet}=4$, $a_{k\bullet}>4$, hence case 1.(i.) can occur if and only if $a_{l\bullet}=4$, $2<a_{k\bullet}\leq4$ or $a_{l\bullet}>4$, $a_{k\bullet}>1$.
\end{remark}
\begin{proposition}
\label{prop.2.s}
Let $K=2$. Assume that $|a_{k0}-a_{k1}|=a_{k\bullet}>1$ and $a_{l0},a_{l1}>1$, $a_{l0}\neq a_{l1}$, for $l\neq k$. Then
\begin{itemize}
\item[1.] $\lambda_D=0$ if any of the following holds
\begin{itemize}
\item[(i.)] $(a_{l0}-a_{l1})^2<a_{l\bullet}$, $a_{k\bullet}(a_{k\bullet}-1)\leq a_{l\bullet}-(a_{l0}-a_{l1})^2$
 \item[(ii.)] $(a_{l0}-a_{l1})^2> a_{l\bullet}$, $a_{k\bullet}>2 $
 \item[(iii.)] $a_{l0}a_{l1}\leq \frac{1}{6}\left(a_{l0}^2+a_{l1}^2+5a_{l\bullet}-10\right)$, $a_{k\bullet}=2$
 \item[(iv.)] $(a_{l0}-a_{l1})^2< a_{l\bullet}$, $a_{k\bullet}(a_{k\bullet}-1)\geq a_{l\bullet}-(a_{l0}-a_{l1})^2$,
    $a_{k\bullet}(a_{k\bullet}-1)(3a_{l\bullet}-5)-2((a_{l0}-1)a_{l0}+(a_{l1}-1)a_{l1})+3(a_{k\bullet}-1)( (a_{l0}-a_{l1})^2-a_{l\bullet}  )  \geq 0$
 \item[(v.)] $(a_{l0}-a_{l1})^2= a_{l\bullet}$, $ a_{k\bullet}(a_{k\bullet}-1)(3a_{l\bullet}-5)\geq 4 a_{l0}a_{l1}$
\end{itemize}
 \item[2.] $\lambda_D=\infty$ if $a_{k\bullet}=2$, $a_{l\bullet}=5$.
\end{itemize}
\end{proposition}
\begin{remark}
{\color{white} Some remarks about theorem 1 are in order}
\begin{itemize}
\item[1.] Cases 1.(i.)-1.(v.) are mutually exclusive but do not include examples where analytical solution could only be obtained by directly investigating the LOOCV deviance. However, our extensive numerical evaluation suggests that $\lambda_D=0$ applies to a general example in such a setting.
 \item[2.] Note that 2. is a special case of 1.(i.), hence here the LOOCV deviance has two minima, one at $\lambda=0$ and the other at $\lambda=\infty$.
 \end{itemize}
\end{remark}



\subsubsection{Tuning the LOOCV squared prediction error}

The LOOCV squared prediction error which for our setting is given by
\begin{equation}
\label{eq.pm}
PE(\lambda)=\frac{1}{n}\sum_{k=1}^K \left( a_{k0}(\hat\pi_{k0}(\lambda))^2+a_{k1}(1-\hat\pi_{k1}(\lambda))^2  \right) ,
\end{equation}
where $\hat\pi_{k0}(\lambda)$ and $\hat\pi_{k1}(\lambda)$ are defined in (\ref{pr.mod1.0}) and (\ref{pr.mod1.1}). Let
$$\lambda_{PE}=\mbox{argmin}_{\lambda\geq0} PE(\lambda), $$
be the point at which the LOOCV squared prediction error defined in (\ref{eq.pm}) attains its minimum. The theoretical results for this criterion are as follows.
\begin{proposition}
\label{th1}
$\lambda_{PE}=0\mbox{, if and only if }$ $|a_{k0}-a_{k1}|=a_{k\bullet}>1$ for all $k=1,\ldots,K$, i.e. when there is complete separation in the data.
\end{proposition}

The following result is obtained for $K=1$.
 \begin{lemma}
 \label{lm1}
Let $K=1$. Assume that $a_{10},a_{11}>1$.
\begin{itemize}
\item[1.] if $(a_{10}-a_{11})^2-a_{1\bullet}\leq0$, then $\lambda_{PE}=\infty$
 \item[2.] if $(a_{10}-a_{11})^2-a_{1\bullet}>0$, then
\begin{equation}
\label{eq.pm.opt}
\lambda_{PE}=\frac{a_{10}a_{11}}{(a_{10}-a_{11})^2-a_{1\bullet}}.
\end{equation}
\end{itemize}
\end{lemma}
\begin{remark}
Plugging (\ref{eq.pm.opt}) into (\ref{eq.pkl}) gives a closed form approximation of the RR proportion estimate optimized by the LOOCV squared prediction error.
\end{remark}

\subsubsection{Tuning the validation deviance}

The validation deviance, $D^V(\lambda)=D^V$, is
\begin{equation}
\label{eq.dev.test1}
D^V=-2\sum_{k=1}^K D_k^V=-2\sum_{k=1}^K \left(a_{k0}^V\log(1-\hat\pi_{k}(\lambda))+a_{k1}^V\log \hat\pi_{k}(\lambda)   \right),
\end{equation}
where $a_{k0}^V$ and $a_{k1}^V$ are the respective number of nonevents and events for category $k$ in the validation set. Let
$$\lambda_{D^V}=\mbox{argmin}_{\lambda\geq0} D^V(\lambda), $$
 denote the point at which the validation deviance defined in (\ref{eq.dev.test1}) obtains its minimum. Then when assuming there is no \textit{separation} in the validation set we obtain the following result.
\begin{proposition}
\label{prop.D.valid}
Assume that $a_{k0}^V,a_{k1}^V>0$, for $k=1,\ldots,K$, i.e. there is no \textit{separation} in the validation set. If $|a_{k0}-a_{k1}|=a_{k\bullet}>0$, $k=1,\ldots,K$, i.e. there is complete separation in the original data set, then $\lambda_{D^V}>0$.
\end{proposition}

Examples where the original data are not completely separated, even with $K=1$, are too difficult to investigate analytically without imposing further assumptions. However, there are a few exceptions. E.g., let $A=\{k; a_{k0}=a_{k1} \}$ and $B=\{k; |a_{k0}-a_{k1}|=a_{k\bullet}\}$, $A\cup B=S=\{1,\ldots,K\}$. Then it can be shown that for $k\in A$, $\frac{dD_k^V}{d\lambda}=0$ holds and Proposition \ref{prop.D.valid} still applies. 


\subsubsection{Tuning the validation squared prediction error}

The validation squared prediction error is,
\begin{equation}
\label{eq.pe.test1}
PE^V(\lambda)=\frac{1}{n_V}\sum_{k=1}^K \left( a_{k0}^V(\hat\pi_k(\lambda))^2+a_{k1}^V(1-\hat\pi_k(\lambda))^2  \right),
\end{equation}
where $n_V$ is the number of subjects in the validation set. Let
$$\lambda_{PE^V}=\mbox{argmin}_{\lambda\geq0} PE^V(\lambda),$$
where $PE^V(\lambda)$ is defined in (\ref{eq.pe.test1}). Assuming there is no \textit{separation} in the validation set and there is complete separation in the original data set we obtain the following result.
\begin{proposition}
\label{pr4}
Assume that $|a_{k0}-a_{k1}|=a_{k\bullet}>0$ and $a_{k0}^V,a_{k1}^V>0$ for all $k=1,\ldots,K$, i.e. there is complete separation in the original data set and no \textit{separation} in the validation set. Then $\lambda_{PE^V}>0$.
\end{proposition}



\subsection{Ridge penalty for the $\gamma$ coefficients}

Here we apply RR penalty for the coefficients $\gamma_k$ where the intercept term is excluded from penalization, i.e.
$$P_1(\lambda)=\frac{\lambda}{2}\sum_{k=1}^{K-1}\gamma_k^2,$$
which can be expressed in terms of $\beta_k$ and $\pi_k$ as
\begin{equation}
\label{pen1.1}
P_1(\lambda)=\frac{\lambda}{2} \sum_{k=1}^{K-1}(\beta_k-\beta_K)^2=\frac{\lambda}{2} \sum_{k=1}^{K-1}\log^2\left(\frac{\pi_k}{1-\pi_k}\frac{1-\pi_K}{\pi_K} \right).
 \end{equation}

Applying the penalty term (\ref{pen1.1}) to the log-likelihood function, $l^{P_1}=l-P_1(\lambda)$, and taking partial derivatives with respect to $\pi_k$, $k=1,\ldots,K$ yields,
$$\frac{\partial l^{P_1}}{\partial \pi_k}=\frac{a_{k1}-\pi_ka_{k\bullet}-\lambda\log\frac{\pi_k}{1-\pi_k}\frac{1-\pi_K}{\pi_K}}{\pi_k(1-\pi_k)} \mbox{, } k=1,\ldots,K-1,  $$
$$\frac{\partial l^{P_1}}{\partial \pi_K}=\frac{a_{K1}-\pi_Ka_{K\bullet}+\lambda\sum_{k=1}^{K-1}\log\left(\frac{\pi_k}{1-\pi_k}\frac{1-\pi_K}{\pi_K}\right) }{\pi_K(1-\pi_K)} .$$

Setting the partial derivatives to zero, a single step Newton-Raphson solutions initializing at 1/2 are

$$\tilde\pi_{K}(\lambda)=\frac{a_{K1}+4\lambda\sum_{j=1}^{K-1}\frac{a_{j1}}{a_{j\bullet}+4\lambda}}{a_{K\bullet}+4\lambda\sum_{j=1}^{K-1}\frac{a_{j\bullet}}{a_{j\bullet}+4\lambda}} $$

$$\tilde\pi_{k}(\lambda)= \frac{a_{k1}+4\lambda \tilde\pi_{K}(\lambda)}{a_{k\bullet}+4\lambda}\mbox{, }k=1,\ldots,K-1,  $$

Observe that for $\lambda=0$ the above estimators simplify to the ML solution, hence Corollary \ref{color1} and \ref{color2} apply. For $\lambda\rightarrow\infty$ we have
$$\lim_{\lambda\rightarrow\infty}\tilde\pi_{k}(\lambda)=\frac{a_{\bullet1}}{n}\mbox{, }k=1,\ldots,K, $$
the obvious estimates for the fully penalized  model. Given the unpleasant form of $\tilde\pi_k(\lambda)$, $k=1,\ldots,K$, for a general $K$, henceforth assume that $K=2$ where
\begin{equation}
\label{estt.pk.mod2}
\tilde\pi_k(\lambda)= \frac{ a_{j\bullet}a_{k1}+4\lambda a_{\bullet 1}      }{ a_{k\bullet}a_{j\bullet}+4\lambda n  } \mbox{, } k=1,2\mbox{ and }j\neq k.
\end{equation}
The penalized ML estimates of the event probability when omitting subject $z_{ik}=1$,  $y_i=0$ and $z_{ik}=1$, $y_i=1$, respectively, are
\begin{equation}
\label{est.ridge2.1}
\tilde\pi_{k0}(\lambda)= \frac{a_{k1}a_{j\bullet}+4\lambda a_{\bullet 1}}{(a_{k\bullet}-1)a_{j\bullet}+4\lambda(n-1)} \mbox{, }k=1,2 \mbox{, }j\neq k ,
\end{equation}
\begin{equation}
\label{est.ridge2.2}
\tilde\pi_{k1}(\lambda)= \frac{(a_{k1}-1)a_{j\bullet}+4\lambda (a_{\bullet 1}-1)}{(a_{k\bullet}-1)a_{j\bullet}+4\lambda(n-1)} \mbox{, }k=1,2 \mbox{, }j\neq k .
\end{equation}

(Note that over-parameterizing the model (\ref{mod2}) so that shrinkage does not depend on the choice of the reference category, for $K=2$ yields the estimator that is identical as $\pi_k(\lambda/2)$, $k=1,2$. Hence the results presented hereafter would apply also for the over-parameterized model.)

\subsubsection{Tuning the LOOCV deviance}

Using the estimates of the event probability defined in (\ref{est.ridge2.1}) and (\ref{est.ridge2.2}) when calculating the LOOCV deviance (Eq. (\ref{eq.dev})) gives the following result.
\begin{proposition}
\label{pr5}
Let $K=2$ and let $\lambda_{D}$ denote the point at which the LOOCV deviance defined in (\ref{eq.dev}) utilizing the estimates of the event probability defined in (\ref{est.ridge2.1}) and (\ref{est.ridge2.2}) attains its minimum. Then $\lambda_{D}=0 \mbox{ if and only if } $ any of the following mutually exclusive conditions holds,
\begin{itemize}
\item[1.] $|a_{k0}-a_{k1}|= a_{k\bullet}>1\mbox{ for }k=1,2$, i.e. there is complete separation in the data,
\item[2.] assume without the loss of generality, $a_{10}=0$, $a_{20},a_{21}>1$, i.e. there is quasi-complete separation in the data
\begin{itemize}
\item[(i.)] $a_{21}>a_{20}$
\item[(ii.)] $2<a_{21}\leq a_{20}$
\item[(iii.)] $2=a_{21}\leq a_{20}$ and $a_{11}(a_{20}-1)-a_{20}(1+a_{20})+2\leq 0$.
\end{itemize}
\end{itemize}
\end{proposition}
\begin{remark}
When $a_{20},a_{21}=2$ the condition 2.(iii.) simplifies to $a_{11}\leq 4$. When $a_{11}=0$ just replace $a_{20}$ and $a_{21}$ with $a_{21}$ and $a_{20}$, respectively. 
\end{remark}

\subsubsection{Tuning the LOOCV squared prediction error}

When using (\ref{est.ridge2.1}) and (\ref{est.ridge2.2}) for calculating the LOOCV squared prediction error (Eq. (\ref{eq.pm})), the following result applies.
\begin{proposition}
\label{pr6}
Let $K=2$ and let $\lambda_{PE}$ denote the point at which the LOOCV squared prediction error defined in (\ref{eq.pm}) utilizing the estimates of the event probability defined in (\ref{est.ridge2.1}) and (\ref{est.ridge2.2}) attains its minimum. Then $\lambda_{PE}=0 \mbox{, if and only if } $ $|a_{k0}-a_{k1}|= a_{k\bullet}>1\mbox{ for }k=1,2$, i.e. there is complete separation in the data.
\end{proposition}

 \section{Proofs}

 \begin{proof}[Proof of Proposition \ref{prop.D.ns}, \ref{prop.D.cs}, \ref{prop.1.s} and \ref{prop.2.s} and Lemma \ref{lemma.D.ns.1}, \ref{lemma.D.ns.2}, \ref{lemma.D.ns.1.0}, and \ref{lemma.D.qcs.K}.]
 The first 
 derivative of $D$ with respect to $\lambda$ is


$$\frac{d D}{d \lambda}=-2\sum_{k=1}^K \frac{d D_k}{d \lambda},$$

where

\begin{eqnarray}
\label{der}
 -2\frac{d D_k}{d \lambda}&=&-4\left( \frac{a_{k1}(a_{k0}-a_{k1}+1)}{(a_{k1}-1+2\lambda)(a_{k\bullet}-1+4\lambda)}-\frac{a_{k0}(a_{k0}-a_{k1}-1)}{(a_{k0}-1+2\lambda)(a_{k\bullet}-1+4\lambda)}   \right) \nonumber\\
  &=&\frac{-4\left( a_{k1}(a_{k0}-1+2\lambda)(a_{k0}-a_{k1}+1)-a_{k0}(a_{k1}-1+2\lambda)(a_{k0}-a_{k1}-1)  \right) }{(a_{k1}-1+2\lambda)(a_{k0}-1+2\lambda)(a_{k\bullet}-1+4\lambda)}  \nonumber\\
  &=&\frac{-4\left( a_{k0}(a_{k0}-1)+a_{k1}(a_{k1}-1)-2\lambda( (a_{k0}-a_{k1})^2-a_{k1}-a_{k0}  )  \right)}{(a_{k1}-1+2\lambda)(a_{k0}-1+2\lambda)(a_{k\bullet}-1+4\lambda)},
\end{eqnarray}



Note that $\frac{d D_k}{d \lambda}$ and hence $\frac{d D}{d \lambda}$ are continuous on $0\leq\lambda<\infty$.

For Proposition \ref{prop.D.ns}, write

$$\frac{d D}{d \lambda}=-4\sum_{k=1}^K \frac{a_{k0}^2+a_{k1}^2-a_{k\bullet}-2\lambda((a_{k0}-a_{k1})^2-a_{k\bullet})}{(a_{k1}-1+2\lambda)(a_{k0}-1+2\lambda)(a_{k\bullet}-1+4\lambda)}. $$

It is obvious that for $\lambda=0$, $\frac{d D}{d \lambda}<0$, thus proving Proposition \ref{prop.D.ns}. Assuming that $(a_{k0}-a_{k1})^2-a_{k\bullet}\geq 0$ holds for all $k=1,...,K$, then we have for $0\leq\lambda<\lambda^-$, for some $\lambda^->0$, $\frac{d D}{d \lambda}<0$, while for $\lambda>\lambda^-$, $\frac{d D}{d \lambda}>0$, hence $D$ is minimized at $\lambda^-$, thus proving Lemma \ref{lemma.D.ns.1}.

For Proposition \ref{prop.D.cs} assume w.l.o.g. that $a_{k0}=0$ and $a_{k1}>1$, then

$$-2\frac{d D_k}{d \lambda}=\frac{4a_{k1}(a_{k1}-1)}{(a_{k1}-1+2\lambda)(a_{k1}-1+4\lambda)}>0, $$

hence $\frac{d D}{d \lambda}>0$ which proves Proposition \ref{prop.D.cs}.

Write

$$-2\frac{d D_k}{d \lambda}=-\frac{4a_{k\bullet}}{(\frac{a_{k\bullet}}{2}-1+2\lambda)(a_{k\bullet}-1+4\lambda)}<0\mbox{, for }k=1,...,K, $$

hence if $a_{k0}=a_{k1}$ holds for all $k=1,...,K$, then $\frac{d D}{d \lambda}<0$ which proves Lemma \ref{lemma.D.ns.2} for a special case, where $a_{k0}=a_{k1}$.

For a general situation considered in Lemma \ref{lemma.D.ns.2}, it is obvious from the last line of (\ref{der}) that $-2\frac{d D_k}{d \lambda}<0$, hence $\frac{d D}{d \lambda}<0$ which proves Lemma \ref{lemma.D.ns.2}.

For Lemma \ref{lemma.D.ns.1.0}, setting $\frac{d D}{d \lambda}$ to zero, yields

\begin{equation}
\label{eq.opt.l1}
\lambda=\frac{ -2a_{10}a_{11}+(a_{10}-a_{11}+1)a_{11}-(a_{10}-a_{11}-1)a_{10} }{ 2((a_{10}-a_{11}+1)a_{11}-(a_{10}-a_{11}-1)a_{10})  }=\frac{ a_{10}(a_{10}-1)+a_{11}(a_{11}-1) }{2(( a_{10}-a_{11} )^2-a_{10}-a_{11}  )} .
\end{equation}

from where it can be seen that under the setting of Lemma \ref{lemma.D.ns.1.0}, $0<\lambda<\infty$. Since for $\lambda=0$, $\frac{d D}{d \lambda}<0$, it follows that 
(\ref{eq.opt.l1}) is indeed the minimum, hence the deviance is optimized at (\ref{eq.opt.l1}).

For Lemma \ref{lemma.D.qcs.K} assume w.l.o.g. that for the first $K_1$ categories there are no nonevents,

$$a_{k0}=0\mbox{ for }k=1,...,K_1, $$

and that for the other $K_2=K-K_1$ categories there are at least two events and nonevents,

$$a_{k0},a_{k1}>1\mbox{ for }k=K_1+1,...,K .$$

For $\lambda=0$, $\frac{d D}{d \lambda}$, can be expressed as

\begin{eqnarray}
\label{kaj}
\frac{d D}{d \lambda}|\lambda=0&=&\frac{-4}{ \prod_{k=1}^{K_1}(a_{k1}-1)\prod_{k=K_1+1}^K (a_{k0}-1)(a_{k1}-1)(a_{k\bullet}-1)  }\left[ \right.\nonumber\\
&& \sum_{k=K_1+1}^K (a_{k0}^2+a_{k1}^2-a_{k\bullet})\prod_{l=1}^{K_1} (a_{l1}-1)\prod_{l=K_1+1,l\neq k}^K (a_{l1}-1)(a_{l0}-1)(a_{l\bullet}-1)\nonumber\\
&-&\left.\sum_{k=1}^{K_1}a_{k1}\prod_{l=1,l\neq k}^{K_1}(a_{l1}-1)\prod_{l=K_1+1}^K (a_{l1}-1)(a_{l0}-1)(a_{l\bullet}-1)  \right].
\end{eqnarray}

Observe that the denominator in (\ref{kaj}) is always positive, while the numerator can be expressed as given in Lemma \ref{lemma.D.qcs.K} by using simple algebra, hence under the condition of Lemma \ref{lemma.D.qcs.K} it follows that for $\lambda=0$, $\frac{d D}{d \lambda}<0$, hence proving Lemma \ref{lemma.D.qcs.K}.

To prove Proposition \ref{prop.1.s}, assume, w.l.o.g., $a_{10}=0$, $a_{11}>1$, and $ a_{20}=a_{21}>1$. In this case,

\begin{eqnarray*}
\frac{d D}{d \lambda}&=&4\left( \frac{a_{11}(a_{11}-1)(\frac{a_{2\bullet}}{2}-1+2\lambda)(a_{2\bullet}-1+4\lambda)-a_{2\bullet}(a_{11}-1+2\lambda)(a_{11}-1+4\lambda)}{(a_{11}-1+2\lambda)(a_{11}-1+4\lambda)(\frac{a_{2\bullet}}{2}-1+2\lambda)(a_{2\bullet}-1+4\lambda)} \right)\\
&=&\frac{4(\phi(\lambda)-\xi(\lambda))}{\psi(\lambda)},
 \end{eqnarray*}

for nonnegative continuous functions $\phi(\lambda),\xi(\lambda),\psi(\lambda)>0$ on $0\leq \lambda<\infty$. 

Rewrite the expression $\phi(\lambda)-\xi(\lambda)$ to

\begin{eqnarray}
\label{qv}
\phi(\lambda)&-&\xi(\lambda)=8\lambda^2( a_{11}(a_{11}-1)-a_{2\bullet} )\nonumber\\
&+&2\lambda(a_{11}-1)(a_{11}(2a_{2\bullet}-3)-3a_{2\bullet} )\nonumber\\
&+& a_{11}(a_{11}-1)(\frac{a_{2\bullet}}{2}-1)( a_{2\bullet} -1)-a_{2\bullet}(a_{11}-1)^2,
\end{eqnarray}

from where it can be seen that $\phi(\lambda)-\xi(\lambda)$ is a quadratic function of $\lambda$. The derivative of $\phi(\lambda)-\xi(\lambda)$ with respect to $\lambda$ is

 \begin{eqnarray}
\label{derq}
\frac{d(\phi(\lambda)-\xi(\lambda))}{d\lambda}&=&16\lambda ( a_{11}(a_{11}-1)-a_{2\bullet} )\nonumber\\
&+&2(a_{11}-1)(a_{11}(2a_{2\bullet}-3)-3a_{2\bullet} ).
\end{eqnarray}

Note that $\frac{d(\phi(\lambda)-\xi(\lambda))}{d\lambda}$ is continuous on $0\leq\lambda<\infty$.

Assume,

\begin{itemize}
\item [A1.]
$$a_{11}(a_{11}-1)-a_{2\bullet}> 0 \mbox{ and } a_{2\bullet}(2a_{11}-3)-3a_{11}\geq 0,$$

or equivalently

$$a_{11}>2 \mbox{ and }  a_{2\bullet}< a_{11}(a_{11}-1).$$

Under this condition $\phi(\lambda)-\xi(\lambda)$ is a strictly increasing function of $\lambda$. If we further assume

$$ a_{11}(\frac{a_{2\bullet}}{2}-1)( a_{2\bullet} -1)\geq a_{2\bullet}(a_{11}-1),  $$

then it follows that for $\lambda\geq 0$, $\phi(\lambda)-\xi(\lambda)\geq 0$ and hence $\frac{dD}{d\lambda}\geq 0$ which implies that the deviance is optimized at $\lambda^*=0$, which proves claim 1.(A).

If however,

$$ a_{11}(\frac{a_{2\bullet}}{2}-1)( a_{2\bullet} -1)< a_{2\bullet}(a_{11}-1),  $$

 then for $0<\lambda<\lambda^+$, for some $\lambda^+>0$, $\frac{dD}{d\lambda}< 0$, for $\lambda=\lambda^*$, $\frac{dD}{d\lambda}= 0$, and for $\lambda>\lambda^+$, $\frac{dD}{d\lambda}> 0$ which implies that the deviance is minimized for some $0<\lambda^*<\infty$, thus proving claim 3..

\item [A2.]
$$a_{11}(a_{11}-1)-a_{2\bullet}< 0  \mbox{ and }  a_{2\bullet}(2a_{11}-3)-3a_{11}< 0,$$

or equivalently

$$a_{11}=2 \mbox{ and } a_{2\bullet}<6.  $$

Here $\phi(\lambda)-\xi(\lambda)$ is a strictly decreasing function of $\lambda$ and since here

$$a_{11}(\frac{a_{2\bullet}}{2}-1)( a_{2\bullet} -1)> a_{2\bullet}(a_{11}-1) ,$$

 it can be seen that for  $0\leq\lambda<\lambda^+$, for some $\lambda^+>0$,  $\frac{d(\phi(\lambda)-\xi(\lambda))}{d\lambda}>0$  for $\lambda=\lambda^+$, $\frac{d(\phi(\lambda)-\xi(\lambda))}{d\lambda}=0$ holds, while for $\lambda>\lambda^+$, $\frac{d(\phi(\lambda)-\xi(\lambda))}{d\lambda}<0$ which implies that to obtain the point at which the deviance is optimized it is sufficient to compare the deviance when using $\lambda=0$ and $\lambda=\infty$. For $\lambda=0$ the deviance is

$$D(\lambda=0)=-2a_{2\bullet}\left(  \log\left( \frac{\frac{a_{2\bullet}}{2}-1}{a_{2\bullet}-1} \right)  \right),$$

while

$$\lim_{\lambda\rightarrow \infty} D(\lambda)=-2( a_{11}+a_{2\bullet} )\log\left(\frac{1}{2}\right) ,$$

and the difference is

\begin{eqnarray}
\label{difDev}
\Delta&=&D(\lambda=0)-\lim_{\lambda\rightarrow \infty} D(\lambda)\nonumber\\
&=&-2\left( a_{11}\log(2)+a_{2\bullet}\log\left(\frac{a_{2\bullet}-2}{a_{2\bullet}-1}\right)  \right).
\end{eqnarray}

In this setting it is easy to see that $D>0$ which implies that the deviance is optimized at $\lambda^*=\infty$ hence proving claim 2..

\item[A3.]
$$a_{11}(a_{11}-1)-a_{2\bullet}\geq 0  \mbox{ and }  a_{2\bullet}(2a_{11}-3)-3a_{11}< 0, $$
or equivalently
$$a_{2\bullet}\leq a_{11}(a_{11}-1)  \mbox{ and } a_{2\bullet}<\frac{3a_{11}}{2a_{11}-3} , $$

which can not hold since it implies that

$$a_{2\bullet}<3. $$

\item[A4.]
$$a_{11}(a_{11}-1)-a_{2\bullet}< 0  \mbox{ and }    a_{2\bullet}(2a_{11}-3)-3a_{11}\geq 0, $$
or equivalently
$$a_{11}=2   \mbox{ and }  a_{2\bullet}>4 $$
or
$$a_{11}>2  \mbox{ and }  a_{2\bullet}>a_{11}(a_{11}-1). $$
In this case it can be seen that for  $0\leq\lambda<\lambda^+$, for some $\lambda^+>0$,  $\frac{d(\phi(\lambda)-\xi(\lambda))}{d\lambda}>0$  for $\lambda=\lambda^+$, $\frac{d(\phi(\lambda)-\xi(\lambda))}{d\lambda}=0$ holds, while for $\lambda>\lambda^+$, $\frac{d(\phi(\lambda)-\xi(\lambda))}{d\lambda}<0$ which implies that to obtain the point at which the deviance is optimized it is sufficient to compare the deviance when using $\lambda=0$ and $\lambda=\infty$.

It can be seen that if $a_{11}=2   \mbox{ and }  a_{2\bullet}>4 $, then $\Delta<0$, for $\Delta$ defined in (\ref{difDev}), and the deviance is optimized at $\lambda^*=0$ which proves claim 1.(C).

Also, for $a_{11}>2  \mbox{ and }  a_{2\bullet}>a_{11}(a_{11}-1) $, then $\Delta<0$.

\item[A5.]
finally, assume that

$$a_{11}(a_{11}-1)-a_{2\bullet}= 0 \mbox{ and } a_{2\bullet}(2a_{11}-3)-3a_{11}\geq 0,$$

which is  equivalent to

$$a_{11}>2 \mbox{ and }a_{2\bullet}=a_{11}(a_{11}-1). $$

In this case  $\phi(\lambda)-\xi(\lambda)$ and hence $\frac{dD}{d\lambda}$ is a strictly increasing function of $\lambda$. Since in this setting it can be seen that
$$ a_{11}(\frac{a_{2\bullet}}{2}-1)( a_{2\bullet} -1)\geq a_{2\bullet}(a_{11}-1),  $$

then $\frac{dD}{d\lambda}>0$ and the deviance is optimized at $\lambda^*=0$, which together with 4. proves claim 1.(B).

\end{itemize}

For Proposition \ref{prop.2.s} assume w.l.o.g. that $a_{10}=0$, $a_{11}>1$, $1<a_{20}<a_{21}$. Write,
\begin{eqnarray*}
\frac{d D}{d \lambda}&=& \frac{4}{ (a_{11}-1+2\lambda)(a_{11}-1+4\lambda)(a_{21}-1+2\lambda)(a_{20}-1+2\lambda)(a_{2\bullet}-1+4\lambda)  } \left( \right.\\
&& a_{11}(a_{11}-1)(a_{20}-1+2\lambda)(a_{21}-1+2\lambda)(a_{2\bullet}-1+4\lambda)\\
&-&\left.(a_{20}(a_{20}-1)+a_{21}(a_{21}-1)-2\lambda( (a_{20}-a_{21})^2-a_{21}-a_{20}  ))(a_{11}-1+2\lambda)(a_{11}-1+4\lambda)  \right) \\
&=&4\left(\frac{\phi(\lambda)-\xi(\lambda)}{\psi(\lambda)}\right),
\end{eqnarray*}

for $\phi(\lambda)>0$ and $\psi(\lambda)>0$. The first and second derivatives of $\phi(\lambda)-\xi(\lambda)$ with respect to $\lambda$ are

\begin{eqnarray}
\label{egg}
\frac{d \phi(\lambda)}{d\lambda}&-&\frac{d \xi(\lambda)}{d\lambda}=48\lambda^2\left[a_{11}(a_{11}-1)+(a_{20}-a_{21})^2-a_{2\bullet}\right]\nonumber\\
&+&8\lambda\left[ a_{11}(a_{11}-1)(3a_{2\bullet}-5)-2\left((a_{20}-1)a_{20}+(a_{21}-1)a_{21}\right)+\right.\nonumber\\
&&\left.3(a_{11}-1)( (a_{20}-a_{21})^2-a_{2\bullet}  )   \right]\nonumber\\
&+&2(a_{11}-1) \left[   a_{11}((a_{2\bullet}-1)(a_{2\bullet}-2)+2(a_{20}-1)(a_{21}-1))+\right.\nonumber\\
&&\left. (a_{11}-1)( (a_{20}-a_{21})^2-a_{2\bullet} )-3((a_{20}-1)a_{20}+(a_{21}-1)a_{21})  \right],
\end{eqnarray}

\begin{eqnarray}
\label{egg2}
\frac{d^2 \left(\phi(\lambda)-\xi(\lambda)\right)}{d\lambda^2}&=&96\lambda\left[a_{11}(a_{11}-1)+(a_{20}-a_{21})^2-a_{2\bullet}\right]\nonumber\\
&+&8\left[ a_{11}(a_{11}-1)(3a_{2\bullet}-5)-2\left((a_{20}-1)a_{20}+(a_{21}-1)a_{21}\right)+\right.\nonumber\\
&&\left. 3(a_{11}-1)\left( (a_{20}-a_{21})^2-a_{2\bullet}  \right)   \right] .
\end{eqnarray}

Note that the functions $\phi(\lambda)$, $\xi(\lambda)$, $\psi(\lambda)$, $\frac{d \phi(\lambda)}{d\lambda}-\frac{d \xi(\lambda)}{d\lambda}$ and $\frac{d^2 \left(\phi(\lambda)-\xi(\lambda)\right)}{d\lambda^2}$ are continuous on $0\leq \lambda<\infty$ and

$$\lim_{\lambda\rightarrow\infty} \frac{d D}{d \lambda}=0.$$

Observe that

\begin{eqnarray}
\label{eq24}
a_{11}((a_{2\bullet}-1)(a_{2\bullet}-2)+2(a_{20}-1)(a_{21}-1))+&&\nonumber\\
(a_{11}-1)( (a_{20}-a_{21})^2-a_{2\bullet} )-&&\nonumber\\
3((a_{20}-1)a_{20}+(a_{21}-1)a_{21})&\geq&\nonumber\\
2((a_{2\bullet}-1)(a_{2\bullet}-2)+2(a_{20}-1)(a_{21}-1))+&&\nonumber\\
 ( (a_{20}-a_{21})^2-a_{2\bullet} )-&&\nonumber\\
3((a_{20}-1)a_{20}+(a_{21}-1)a_{21})&=&\nonumber\\
2(   3a_{20}a_{21}-4a_{2\bullet}+4  )&>&0,
\end{eqnarray}

\begin{eqnarray}
\label{eq241}
f(0)-g(0)&=& (a_{11}-1)( a_{11}(a_{20}-1)(a_{21}-1)(a_{2\bullet}-1)-(a_{11}-1)(a_{20}^2+a_{21}^2-a_{2\bullet})   )\nonumber\\
&>&(a_{11}-1)^2(  (a_{20}-1)(a_{21}-1)(a_{2\bullet}-1)- (a_{20}^2+a_{21}^2-a_{2\bullet})  ) \nonumber\\
&=&(a_{11}-1)^2( a_{20}^2(a_{21}-2)+a_{21}^2(a_{20}-2)+3(a_{2\bullet}-a_{20}a_{21} )-1   ),\nonumber\\
&\geq&0,
\end{eqnarray}

\begin{eqnarray}
\label{eq242}
 \Delta&=&D(\lambda=0)-\lim_{\lambda\rightarrow\infty}D(\lambda)\nonumber\\
 &=&-2\left( a_{11}\log(2)+a_{20}\log \left( \frac{2a_{20}-2}{a_{2\bullet}-1} \right) +a_{21}\log\left( \frac{2a_{21}-2}{a_{2\bullet}-1} \right) \right)\nonumber\\
 &\leq&-2\left( 2\log(2)+a_{20}\log \left( \frac{2a_{20}-2}{a_{2\bullet}-1} \right) +a_{21}\log\left( \frac{2a_{21}-2}{a_{2\bullet}-1} \right) \right)\nonumber\\
 &\leq&-2\left( 2\log(2)+2\log \left( \frac{2}{5-1} \right) +3\log\left( \frac{6-2}{5-1} \right) \right)=0.
\end{eqnarray}

\begin{itemize}
\item[A.]
Assume that $a_{11}(a_{11}-1)< a_{2\bullet}-(a_{20}-a_{21})^2$, which implies that $(a_{20}-a_{21})^2< a_{2\bullet}$. Then by (\ref{eq24}) and (\ref{eq241}) it follows that up to some $\lambda^+>0$, $\phi(\lambda)>\xi(\lambda)$, hence up to $0\leq\lambda<\lambda^+$, deviance is an increasing function of $\lambda$, for $\lambda=\lambda^+$, the deviance obtains its maximum, while for $\lambda>\lambda^+$ the deviance is a decreasing function of $\lambda$. Hence, it is sufficient to compare the deviance for $\lambda=0$ and $\lambda=\infty$. From (\ref{eq242}) it then follows that the deviance is minimized at $\lambda=0$. Note that $a_{11}=2$ and $a_{20}=2$, $a_{21}=3$ is a special case where the deviance has two minima, one at $\lambda=0$ and one at $\lambda=\infty$, which proves claim 2..

\item[B.]
Assume that $a_{11}(a_{11}-1)= a_{2\bullet}-(a_{20}-a_{21})^2$ which implies that $(a_{20}-a_{21})^2< a_{2\bullet}$. Then (\ref{egg}) is a linear function of $\lambda$. Assuming that

$$a_{11}(a_{11}-1)(3a_{2\bullet}-5)-2((a_{20}-1)a_{20}+(a_{21}-1)a_{21})+ 3(a_{11}-1)( (a_{20}-a_{21})^2-a_{\bullet}  )\geq 0,  $$

then it is obvious that $\frac{dD}{d\lambda}>0$, hence the deviance is minimized at $\lambda=0$, while for $$a_{11}(a_{11}-1)(3a_{2\bullet}-5)-2((a_{20}-1)a_{20}+(a_{21}-1)a_{21})+ 3(a_{11}-1)( (a_{20}-a_{21})^2-a_{\bullet}  )< 0,  $$
then by the same argument as in A. and by (\ref{eq242}) it follows that the deviance is minimized at $\lambda=0$. Items A. and B. hence prove claim 1. under (C1).

\item[C.]

Now assume that
\begin{equation}
\label{eq22g}
 (a_{20}-a_{21})^2> a_{2\bullet},
 \end{equation}

which implies that

\begin{equation}
\label{eq22.i}
2a_{20}a_{21}+a_{2\bullet}<a_{20}^2+a_{21}^2,
\end{equation}

\begin{equation}
\label{eq22.ii}
 a_{2\bullet}>6,
\end{equation}

 \begin{equation}
\label{eq22}
a_{11}(a_{11}-1)+(a_{20}-a_{21})^2-a_{2\bullet}>0,
\end{equation}

 \begin{equation}
\label{eq23}
\frac{1}{6}( a_{20}^2+a_{21}^2+5a_{2\bullet}-10 ) <  \frac{1}{2}(a_{20}^2+a_{21}^2 -a_{2 \bullet}).
\end{equation}

Further assume that $a_{11}>2$, then it follows from (\ref{eq22.i}) and (\ref{eq22.ii}) that

\begin{eqnarray}
\label{eq23}
a_{11}(a_{11}-1)(3a_{2\bullet}-5)-&&\nonumber\\
2((a_{20}-1)a_{20}+(a_{21}-1)a_{21})+&&\nonumber\\
3(a_{11}-1)( (a_{20}-a_{21})^2-a_{\bullet}  ) &\geq& \nonumber\\
6(3a_{2\bullet}-5)-&&\nonumber\\
2((a_{20}-1)a_{20}+(a_{21}-1)a_{21})+&&\nonumber\\
6( (a_{20}-a_{21})^2-a_{\bullet}  )&=&\nonumber\\
5(a_{2\bullet}-6)+4( a_{20}^2+a_{21}^2-(3a_{20}a_{21}-\frac{9}{4}a_{2\bullet}) )&>&0.\nonumber\\
\end{eqnarray}

This implies that for $\lambda\geq 0$, (\ref{egg}) is an increasing function of $\lambda$. Then by (\ref{eq24}) it follows that $\frac{d \phi(\lambda)}{d\lambda}-\frac{d \xi(\lambda)}{d\lambda}>0$ implying that $\phi(\lambda)-\xi(\lambda)$ is an increasing function of $\lambda$. By (\ref{eq241}) it then follows that $\phi(\lambda)>\xi(\lambda)$, implying that $\frac{d D}{d \lambda}>0$. This proves claim 1. under (C2).

For $a_{11}=2$ assume first that

$$2(3a_{2\bullet}-5)-2((a_{20}-1)a_{20}+(a_{21}-1)a_{21})+3( (a_{20}-a_{21})^2-a_{\bullet}  )\geq 0, $$
 or equivalently

 \begin{equation}
 \label{eh}
 6a_{20}a_{21}-5a_{2\bullet}+2\leq a_{20}^2+a_{21}^2-8,
 \end{equation}

 which together with (\ref{eq22g}) implies that

 $$a_{20}a_{21}\leq \frac{1}{6}\left(a_{20}^2+a_{21}^2+5a_{2\bullet}-10\right).$$

Then based on the same argument as before $\phi(\lambda)-\xi(\lambda)>0$, thus $\frac{d D}{d \lambda}>0$, which proves claim 1. under (C3).

Now assume that

$$2(3a_{2\bullet}-5)-2((a_{20}-1)a_{20}+(a_{21}-1)a_{21})+3( (a_{20}-a_{21})^2-a_{\bullet}  )< 0, $$
 or equivalently

 $$ 6a_{20}a_{21}-5a_{2\bullet}+10> a_{20}^2+a_{21}^2,$$

hence together with (\ref{eq22g})

\begin{equation}
 \label{eh1}
  \frac{1}{6}( a_{20}^2+a_{21}^2+5a_{2\bullet}-10 ) <a_{20}a_{21}< \frac{1}{2}(a_{20}^2+a_{21}^2 -a_{2 \bullet}).
  \end{equation}

Now, $\frac{d(\phi(\lambda)-\xi(\lambda))}{d\lambda}$ reaches its minimum at

$$\lambda^-=-\frac{a_{20}^2+a_{21}^2-6a_{20}a_{21}+5a_{2\bullet}-10  }{12(a_{20}^2+a_{21}^2-2a_{20}a_{21}-a_{2\bullet}+2)} >0,$$

where

$$\frac{d(\phi(\lambda)-\xi(\lambda))}{d\lambda}(\lambda^-)=4(3a_{20}a_{21}-4a_{2\bullet}+4)-\frac{(a_{20}^2+a_{21}^2-6a_{20}a_{21}+5a_{2\bullet}-10)^2}{3(a_{20}^2+a_{21}^2-2a_{20}a_{21}-a_{2\bullet}+2)} $$

Now, assuming that $\frac{d(\phi(\lambda)-\xi(\lambda))}{d\lambda}(\lambda^-)\geq 0$, then $\frac{d D}{d \lambda}>0$, however for $\frac{d(\phi(\lambda)-\xi(\lambda))}{d\lambda}(\lambda^-)< 0$, it is, since $\frac{d D}{d \lambda}$ can also be smaller then zero, necessary to directly investigate the deviance, which in this example is

$$D=-2\left( \log\frac{(1+2\lambda)^2 (a_{20}-1+2\lambda)^{a_{20}}(a_{21}-1+2\lambda)^{a_{21}}  }{ (1+4\lambda)^2(a_{2\bullet}-1+4\lambda)^{a_{2\bullet}} }  \right) .$$
 The investigation of the deviance is however possible only numerically, where it is sufficient to consider only

 $$0\leq\lambda<\frac{a_{20}^2+a_{21}^2-a_{2\bullet}}{a_{20}^2+a_{21}^2-2a_{20}a_{21}-a_{2\bullet}}, $$

 since it is obvious that for $\lambda\geq \frac{a_{20}^2+a_{21}^2-a_{2\bullet}}{a_{20}^2+a_{21}^2-2a_{20}a_{21}-a_{2\bullet}}$ the deviance is an increasing function of $\lambda$. We verified numerically that the deviance is also in this case minimized at $\lambda=0$. 

\item[D.]

Assume that $(a_{20}-a_{21})^2< a_{2\bullet}$ and $a_{11}(a_{11}-1)\geq a_{2\bullet}-(a_{20}-a_{21})^2$.


Assuming that

\begin{eqnarray*}
 && a_{11}(a_{11}-1)(3a_{2\bullet}-5)-2((a_{20}-1)a_{20}+(a_{21}-1)a_{21})+\\
&&3(a_{11}-1)( (a_{20}-a_{21})^2-a_{2\bullet}  )  \geq 0
\end{eqnarray*}

 then by the same argument as before, $\frac{dD}{d\lambda}>0$ thus proving claim 1. under (C4). Now for

 \begin{eqnarray*}
 && a_{11}(a_{11}-1)(3a_{2\bullet}-5)-2((a_{20}-1)a_{20}+(a_{21}-1)a_{21})+\\
&&3(a_{11}-1)( (a_{20}-a_{21})^2-a_{2\bullet}  )  <0,
\end{eqnarray*}

$\frac{d(\phi(\lambda)-\xi(\lambda))}{d\lambda}$ reaches its minimum at

$$\lambda^-= - \frac{   a_{11}(a_{11}-1)(3a_{2\bullet}-5)-2(a_{20}^2+a_{21}^2-a_{2\bullet})+3(a_{11}-1)( (a_{20}-a_{21})^2-a_{\bullet}  )    }{ 12 (a_{11}(a_{11}-1)+(a_{20}-a_{21})^2-a_{2\bullet}) } >0.$$

where
\begin{small}
\begin{eqnarray}
\frac{d(\phi(\lambda)-\xi(\lambda))}{d\lambda}(\lambda^-)&=&2(a_{11}-1) \left[   a_{11}((a_{2\bullet}-1)(a_{2\bullet}-2)+2(a_{20}-1)(a_{21}-1))+\right.\nonumber\\
&&\left. (a_{11}-1)( (a_{20}-a_{21})^2-a_{2\bullet} )-3((a_{20}-1)a_{20}+(a_{21}-1)a_{21})  \right] -\nonumber\\
&& \frac{  ( a_{11}(a_{11}-1)(3a_{2\bullet}-5)-2(a_{20}^2+a_{21}^2-a_{2\bullet})+3(a_{11}-1)( (a_{20}-a_{21})^2-a_{\bullet}  )   )^2 }{ 3 (a_{11}(a_{11}-1)+(a_{20}-a_{21})^2-a_{2\bullet}) }.\nonumber\\
\end{eqnarray}
\end{small}
For $\frac{d(\phi(\lambda)-\xi(\lambda))}{d\lambda}(\lambda^-)\geq 0$, it follows that $\frac{dD}{d\lambda}>0$, while for $\frac{d(\phi(\lambda)-\xi(\lambda))}{d\lambda}(\lambda^-)< 0$, $\frac{dD}{d\lambda}$ could also be negative hence it is necessary to investigate the deviance. We numerically verified that the deviance is also in this case minimized at $\lambda=0$.

\item [E.] Finally, let $(a_{20}-a_{21})^2= a_{2\bullet}$. Assume that

$$a_{11}(a_{11}-1)(3a_{2\bullet}-5)\geq 4 a_{20}a_{21}.  $$

Then by the same argument as before it follows that $\frac{dD}{d\lambda}>0$, which proves claim 1. under (C5). For

$$a_{11}(a_{11}-1)(3a_{2\bullet}-5)< 4 a_{20}a_{21},  $$

$\frac{d(\phi(\lambda)-\xi(\lambda))}{d\lambda}$ reaches its minimum at

$$\lambda^-= - \frac{ a_{11}(a_{11}-1)(3a_{2\bullet}-5)-4 a_{20}a_{21}  }{12a_{11}(a_{11}-1)} >0.$$

where

\begin{eqnarray}
\frac{d(\phi(\lambda)-\xi(\lambda))}{d\lambda}(\lambda^-)&=&-\frac{( a_{11}(a_{11}-1)(3a_{2\bullet}-5)-4 a_{20}a_{21} )^2}{3a_{11}(a_{11}-1)}\nonumber\\
&&+2(a_{11}-1)( a_{11}((a_{2\bullet}-1)(a_{2\bullet}-2)+2(a_{20}-1)(a_{21}-1))-6a_{20}a_{21} ).\nonumber\\
\end{eqnarray}

Similarly as in D., if $\frac{d(\phi(\lambda)-\xi(\lambda))}{d\lambda}(\lambda^-)\geq 0$ then $\frac{dD}{d\lambda}>0$, while for $\frac{d(\phi(\lambda)-\xi(\lambda))}{d\lambda}(\lambda^-)<0$, it is necessary to numerically investigate the deviance. We numerically verified that the deviance is also in this case minimized at $\lambda=0$.

 \end{itemize}

\end{proof}

\begin{proof}[Proof of Proposition \ref{th1} and Lemma \ref{lm1}.]
The derivative of $PE(\lambda)=PE$ with respect to $\lambda$ is

\begin{eqnarray}
\label{der.P}
\frac{dPE}{d\lambda}&=&\frac{8}{n}\sum_{k=1}^K \frac{ \lambda( (a_{k0}-a_{k1})^2-a_{k\bullet} )-a_{k0}a_{k1}  }{ ( a_{k\bullet}-1+4\lambda  )^3  }.
\end{eqnarray}
Note that $\frac{dPE}{d\lambda}$ is continuous on $0\leq\lambda< \infty$.

For $K=1$ it can easily be seen that, since $( a_{k\bullet}-1+4\lambda  )^3>0$, then for $(a_{k0}-a_{k1})^2-a_{k\bullet}> 0$ it follows that $\lambda( (a_{k0}-a_{k1})^2-a_{k\bullet} )-a_{k0}a_{k1}$ is a strictly increasing function of $\lambda$, hence $PE$ is minimized at

$$\frac{ a_{k0}a_{k1} }{ (a_{k0}-a_{k1})^2-a_{k\bullet} } ,$$

while for $(a_{k0}-a_{k1})^2-a_{k\bullet}\leq 0$, it follows that $\frac{dPE}{d\lambda}<0$, thus proving Lemma \ref{lm1}.

Now, for Proposition \ref{th1} assume w.l.o.g. that $a_{k0}=0$. Then assuming that $a_{k1}>1$, the numerator of (\ref{der.P}) for this $k$ is

$$\lambda a_{k1}(a_{k1}-1)\geq0, $$
 and if this holds for all $k=1,...,K$, then obviously $\frac{dPE}{d\lambda}\geq0$, with the equality applying if and only if $\lambda=0$, hence here $PE$ is minimized at $\lambda=0$.

 Now assume w.l.o.g that for the first $K_1$ categories, $a_{k0}=0$, $a_{k1}>1$, $k=1,...,K_1$ while for the remaining $K_2=K-K_1$ categories, $a_{k0},a_{k1}>1$, $k=K_1+1,...,K$. The derivative is then

 $$\frac{dPE}{d\lambda}=\frac{8}{n}\left[ \sum_{k=1}^{K_1} \frac{\lambda a_{k1}(a_{k1}-1) }{ (a_{k1}-1+4\lambda)^3 } +\sum_{k=K_1+1}^K \frac{\lambda( ( a_{k0}-a_{k1})^2-a_{k\bullet})-a_{k0}a_{k1}}{(a_{k\bullet}-1+4\lambda)^3}   \right]. $$

It is easily seen that for $\lambda=0$, $\frac{dPE}{d\lambda}<0$. Hence the proof of Proposition \ref{th1} is complete.

 \end{proof}

 \begin{proof}[Proof of Proposition \ref{prop.D.valid}]
  Assume that there is no \textit{separation} in the validation data, $a_{k0}^T,a_{k1}^T>0$ for $k=1,...,K$. The derivative of $D^T$ with respect to $\lambda$ is then

  $$ \frac{dD^T}{d\lambda}=-4\sum_{k=1}^K \frac{(a_{k\bullet}-2a_{k1})( a_{k1}^T(a_{k0}+2\lambda)-a_{k0}^T(a_{k1}+2\lambda)  )}{ (a_{k0}+2\lambda)(a_{k1}+2\lambda)(a_{k\bullet}+4\lambda) } . $$


  Now assume w.l.o.g that $a_{k1}=0$ for $k=1,...,K$. Note that the derivative is in this case continuous on $\lambda\in(0,\infty)$, while it is not defined for $\lambda=0$. The derivative, omitting the positive constants is

  $$-a_{k1}^T a_{k0}- 2\lambda(a_{k1}^T-a_{k0}^T). $$
   If $a_{k1}^T\geq a_{k0}^T$, the derivative for this category is negative for $\lambda\in(0,\infty)$, while if $a_{k1}^T< a_{k0}^T$ up to some $0<\lambda<\lambda^-$, $\lambda^->0$ the derivative for this category is negative, for $\lambda=\lambda^-$ the derivative for this category is equal to zero, while for $\lambda>\lambda^-$, the derivative is positive, which implies that $\lambda_{D^T}>0$, which proves the proposition.

  \end{proof}

\begin{proof}[Proof of Proposition \ref{pr4}.]
The derivative of $PE^T(\lambda)$ with respect to $\lambda$ is

$$\frac{dPE^T}{d\lambda}=\frac{4}{n_T}\sum_{k=1}^K \frac{(a_{k0}-a_{k1})}{(a_{k\bullet}+4\lambda)^3}\left(a_{k0}^Ta_{k1}-a_{k1}^Ta_{k0}+2\lambda(a_{k0}^T-a_{k1}^T)   \right). $$

For $\lambda\geq0$, $\frac{dPE^T}{d\lambda}$ is continuous. Assume w.l.o.g. that $a_{k0}=0$, $a_{k1}>0$, $k=1,...,K$. The derivative is then

$$ \frac{dPE^T}{d\lambda}= -\frac{4}{n_T}\sum_{k=1}^{K_1} \frac{a_{k1}}{(a_{k1}+4\lambda)^3}\left(a_{k0}^Ta_{k1}+2\lambda(a_{k0}^T-a_{k1}^T)   \right), $$

For $\lambda=0$ we have $ \frac{dPE^T}{d\lambda}<0$, which proves the proposition.
 \end{proof}

 \clearpage
 \newpage

 \begin{proof}[Proof of Proposition \ref{pr5}.]
 The derivative of the deviance with respect to $\lambda$ is


  $$\frac{dD}{d\lambda}=-8\sum_{k=1}^2 \frac{a_{j\bullet}}{(a_{k\bullet}-1)a_{j\bullet}+4\lambda(n-1)}\left(
 \frac{a_{k1}(  a_{\bullet 1}a_{k\bullet}-na_{k1}+a_{j0}   )}{a_{j\bullet}(a_{k1}-1)+4\lambda(a_{\bullet 1}-1)}
-\frac{a_{k0}(  a_{\bullet 1}a_{k\bullet}-na_{k1}-a_{j1} )  }{a_{j\bullet}(a_{k0}-1)+4\lambda(a_{\bullet 0}-1)} \right)$$

Observe that $\frac{dD}{d\lambda}$ is continuous on $\lambda\in[0,\infty)$. For $\lambda=0$ the derivative becomes

 $$\frac{dD}{d\lambda}|_{\lambda=0}=-8\sum_{k=1}^2\frac{ a_{k1}(a_{k0}-1)(a_{\bullet 1}a_{k\bullet}-na_{k1}+a_{j0})-a_{k0}(a_{k1}-1) (a_{\bullet 1}a_{k\bullet}-na_{k1}-a_{j1}) }{a_{j\bullet}(a_{k\bullet}-1)(a_{k1}-1)(a_{k0}-1)}. $$

 It can be verified numerically that


 $$a_{k1}(a_{k0}-1)(a_{\bullet 1}a_{k\bullet}-na_{k1}+a_{j0})> a_{k0}(a_{k1}-1) (a_{\bullet 1}a_{k\bullet}-na_{k1}-a_{j1}), $$

 holds for each $k$, hence $\frac{dD}{d\lambda}|_{\lambda=0}< 0$, which implies that $\lambda_D>0$.

In case of complete separation, assume w.l.o.g that $a_{11}=a_{20}=0$, $a_{10},a_{21}>1$. The derivative then becomes

$$\frac{dD}{d\lambda}=8a_{10}a_{21}\left( \frac{a_{21}(a_{10}-1)}{A_1}+\frac{(a_{21}-1)(n-a_{21})}{A_2}  \right), $$

 where

 $$A_1=( (a_{10}-1)a_{21}+4\lambda(n-1)   )( a_{21}(a_{10}-1)+4\lambda(a_{10}-1)   )>0 ,$$
  $$A_2=( (a_{21}-1)a_{10}+4\lambda(n-1)   )( a_{10}(a_{21}-1)+4\lambda(a_{21}-1)   )>0 .$$

  Obviously, $\frac{dD}{d\lambda}>0$, hence $\lambda_D=0$.

  In case of quasi complete separation, assume w.l.o.g. $a_{10}=0$, $a_{11},a_{20},a_{21}>1$. The derivative is then

 $$\frac{dD}{d\lambda}=-8a_{11}\frac{  d_1 a_0(\lambda)a_1(\lambda)a_2(\lambda)+c_1(\lambda)c_2(\lambda)(b_1+4\lambda b_2)  }{ c_1(\lambda)c_2(\lambda)a_0(\lambda)a_1(\lambda)a_2(\lambda)  } ,$$

  for some nonnegative functions

  $$ a_0(\lambda) = (a_{2\bullet}-1)a_{11}+4\lambda(n-1)>0 $$
  $$ a_1(\lambda) = (a_{21}-1)a_{11}+4\lambda(a_{\bullet 1}-1)>0 $$
  $$ a_2(\lambda) = (a_{20}-1)a_{11}+4\lambda(a_{20}-1)>0 $$
  $$ c_1(\lambda)= (a_{11}-1)a_{2\bullet} +4\lambda(n-1)>0 $$
   $$ c_2(\lambda)= (a_{11}-1)a_{2\bullet} +4\lambda(a_{\bullet 1}-1)>0 $$

  and constants (with respect to $\lambda$)

   \begin{small}
   $$b_1=a_{11}( a_{2\bullet}a_{\bullet 1}(a_{20}-a_{21})+na_{21}(a_{21}-a_{20})+a_{20}a_{11}(a_{21}-1)    )  $$
  \end{small}

    $$b_2=  a_{11}a_{20}(na_{21}-a_{\bullet 1}(a_{2\bullet}-1)-1) +a_{\bullet 1}a_{2\bullet}(a_{20}-a_{21})+na_{21}(a_{21}-a_{20}) $$
     $$d_1= -a_{2\bullet}a_{20}(  a_{11}-1)<0  $$

    The numerator (omitting the constants $-8a_{11}$) at $\lambda=0$ is

    $$d_1a_0a_1a_2+b_1c_1c_2, $$

    $$ a_0  = (a_{2\bullet}-1)a_{11} >0 $$
  $$ a_1  = (a_{21}-1)a_{11} >0 $$
  $$ a_2  = (a_{20}-1)a_{11} >0 $$
  $$ c_1 = (a_{11}-1)a_{2\bullet}  >0 $$
   $$ c_2 = (a_{11}-1)a_{2\bullet}  >0 $$

    which can be expressed, omitting the positive constants, as

    $$-a_{20}a_{11}^2(a_{20}-1)(a_{21}-1)(a_{2\bullet}-1)+(a_{11}-1)a_{2\bullet}( a_{2\bullet}a_{\bullet 1}(a_{20}-a_{21})+na_{21}(a_{21}-a_{20})+a_{20}a_{11}(a_{21}-1) ), $$

    or alternatively,

    \begin{eqnarray}
    \label{lk1}
    a_{11}(a_{21}-1)a_{20}( a_{2\bullet}(a_{11}-1)-a_{11}(a_{20}-1)(a_{2\bullet}-1) )+ \nonumber\\
    +  (a_{11}-1)a_{2\bullet}( a_{2\bullet}a_{\bullet 1}(a_{20}-a_{21})+na_{21}(a_{21}-a_{20})  ).
    \end{eqnarray}

    It can be verified numerically that when $a_{21}>a_{20}$ the above expression is always smaller then zero and in this case we verified numerically that $\frac{dD}{d\lambda}\geq 0$ implying that $\lambda_D=0$, while for $a_{21}<a_{20}$ expression (\ref{lk1}) can be larger then zero for a sufficiently large $a_{11}$, but only when $a_{21}=2$. For $a_{21}=2$, $a_{20}>2$ (\ref{lk1}), omitting the positive constants, becomes

    $$a_{11}(a_{20}-1)-a_{20}(1+a_{20})+2 .$$

    When the above expression is larger then zero, thence $\lambda_D>0$.

    When $a_{20}=a_{21}$ holds, the expression (\ref{lk1}) simplifies to (omitting the positive constants)

    $$5a_{11}a_{20}-2a_{20}( a_{11}a_{20}+1 )-a_{11}, $$

    which is (strictly) larger then zero if and only if

    $$a_{20}=a_{21}=2 \mbox{ and } a_{11}>4,$$

    hence in this case $\lambda_{D}>0$. Note that the above condition is the same as,

    $$a_{11}(a_{20}-1)-a_{20}(1+a_{20})+2>0,$$

    when $a_{20}=2$, hence $a_{20}=a_{21}=2$ is a special case.

     In the other cases $\frac{dD}{d\lambda}$ for $\lambda=0$ is larger (or equal) than zero and we verified numerically that $\frac{dD}{d\lambda}\geq 0$ implying that $\lambda_D=0$. Hence the proof is complete.




    \end{proof}

   \begin{proof}[Proof of Proposition \ref{pr6}.]

    The derivative of the mean prediction error with respect to $\lambda$ is

    $$\frac{dPE}{d\lambda}=\frac{8}{n}\sum_{k=1}^2\frac{a_{j\bullet}}{( (a_{k\bullet}-1)a_{j\bullet}+4\lambda(n-1)  )^3}\left( \right.  $$
    $$ a_{k0}(a_{k1}a_{j\bullet}+4\lambda a_{\bullet 1})(a_{\bullet 1}(a_{k\bullet}-1)-(n-1)a_{k1})$$
    $$\left.-a_{k1}(a_{k0}a_{j\bullet}+4\lambda a_{\bullet 0})((a_{\bullet 1}-1)(a_{k\bullet}-1)-(n-1)(a_{k1}-1))  \right).$$

    Obviously $\frac{dPE}{d\lambda}$ is continuous on $\lambda\in[0,\infty)$. Assuming that there is no \textit{separation}, for $\lambda=0$ the derivative for each $k$, omitting the positive constants, becomes

    \begin{equation}
    \label{eq.pe.e}
    (a_{k\bullet}-1)-(n-1)<0,
    \end{equation}

    hence it follows that $\lambda_{PE}>0$.

    When there is quasi complete separation, assume w.l.o.g. that $a_{10}=0$. For $\lambda=0$, the derivative for $k=1$ is equal to zero, and since (\ref{eq.pe.e}) still applies for $k=2$, it follows that $\lambda_{PE}>0$.

    When there is complete separation, assume w.l.o.g. that $a_{10},a_{21}=0$. The derivative then becomes,

    $$\frac{dPE}{d\lambda}=\frac{32}{n}\frac{(a_{11}^3a_{20}(a_{20}-1)\lambda+a_{11}a_{20}^3(a_{11}-1)\lambda)}{( (a_{k\bullet}-1)a_{j\bullet}+4\lambda(n-1)  )^3}. $$

    Obviously, $\frac{dPE}{d\lambda}\geq 0$ with the equality applying if and only if $\lambda=0$, hence in this case it follows that $\lambda_{PE}=0$. This completes the proof.
  \end{proof}

  \subsection{Validation deviance and validation mean prediction error when applying $P_1(\lambda)$}

   Assume that there is no \textit{separation} in the validation data, $a_{k0}^T,a_{k1}^T>0$ for $k=1,...,K$. It can be shown that if there is no \textit{separation} in the original data then the derivative of $D^T$ with respect to $\lambda$ is continuous on $\lambda\in[0,\infty)$ and is given by,

   $$\frac{dD^T}{d\lambda}=-8\sum_{k=1}^2\frac{a_{j\bullet}(a_{\bullet 1}a_{k\bullet}-na_{k1})}{a_{k\bullet}a_{j\bullet}+4\lambda n}\left(\right. $$
   $$ \left. -\frac{a_{k0}^T}{a_{j\bullet}a_{k0}+4\lambda a_{\bullet 0}}+\frac{a_{k1}^T}{a_{j\bullet}a_{k1}+4\lambda a_{\bullet 1}} \right)  $$

  or, omitting the positive constants,

  $$(a_{\bullet 1}a_{k\bullet}-na_{k1})(a_{k0}^T(a_{j\bullet}a_{k1}+4\lambda a_{\bullet 1})-a_{k1}^T(a_{j\bullet}a_{k0}+4\lambda a_{\bullet 0})) $$
Assume that there is no \textit{separation} in the validation data, $a_{k0}^T,a_{k1}^T>0$ for $k=1,...,K$. It can be shown that if there is no \textit{separation} in the original data then the derivative of $PE^T$ with respect to $\lambda$ is continuous on $\lambda\in[0,\infty)$ and is given by,

$$\frac{dPE^T}{d\lambda}=\frac{8}{n_T}\sum_{k=1}^2\frac{a_{j\bullet}(a_{\bullet 1}a_{k\bullet}-na_{k1})}{(a_{k\bullet}a_{j\bullet}+4\lambda n)^3}\left(\right. $$
$$\left. a_{k0}^T(a_{j\bullet}a_{k1}+4\lambda a_{\bullet 1})-a_{k1}^T(a_{j\bullet}a_{k0}+4\lambda a_{\bullet 0}) \right) , $$

or, omitting the positive constants,

$$(a_{\bullet 1}a_{k\bullet}-na_{k1})(a_{k0}^T(a_{j\bullet}a_{k1}+4\lambda a_{\bullet 1})-a_{k1}^T(a_{j\bullet}a_{k0}+4\lambda a_{\bullet 0})). $$

Observe that $\frac{dD^T_k}{d\lambda}$ and $\frac{dPE^T_k}{d\lambda}$ are identical up to some positive constant. This however, in general, does not hold for $\frac{dD^T}{d\lambda}$ and $\frac{dPE^T}{d\lambda}$.


\pagebreak
\begin{flushleft}
\textbf{\LARGE Supplementary material for: Tuning in ridge logistic regression to solve separation}

\vspace{5mm}

\textbf{\LARGE Additional file 2: Detailed simulation results for all considered scenarios and estimators.}

\end{flushleft}

\setcounter{equation}{0}
\setcounter{figure}{0}
\setcounter{table}{0}
\setcounter{page}{1}
\makeatletter
\renewcommand{\theequation}{S\arabic{equation}}
\renewcommand{\thefigure}{S\arabic{figure}}
\renewcommand{\thetable}{S\arabic{table}}
\renewcommand{\bibnumfmt}[1]{[S#1]}
\renewcommand{\citenumfont}[1]{S#1}

\begin{table}[ht]

\caption{Simulation results showing separation prevalence (SP) and root mean squared errors of $\beta_1$ across simulation scenarios that differed by the number of covariates included into data-generating process $K\in\{7, 10, 15\}$, the sample size $N\in\{80, 200, 500\}$, the value of $\beta_1\in\{0.69, 1.39, 2.08\}$ and noise excluded or included in the fitting process. FC, Firth's correction; RR, ridge regression; $B$, $B$-tuning; $T$, $T$-tuning; $D$, tuning by leave-one-out-cross-validated deviance; $A$, tuning by Akaike's information criterion; $E$, tuning by leave-one-out-cross-validated mean squared prediction error; $C$, tuning by leave-one-out-cross-validated mean classification error.}

\setlength{\tabcolsep}{4pt}
\renewcommand{\arraystretch}{0.88}
\centering
\begin{small}
\begin{tabular}{rrr|r|rrrrrrr|rrrrrrr}
  \hline
$K$ & $N$ & $\beta_1$ & SP  &  \multicolumn{7}{c}{No noise} & \multicolumn{7}{|c}{Noise} \\ 
 &  &  & (\%) & FC & \multicolumn{6}{c|}{RR} & FC & \multicolumn{6}{c}{RR} \\
 &  &  &  &  & $B$ & $T$ & $D$ & $A$ & $E$ & $C$ &  & $B$ & $T$ & $D$ & $A$ & $E$ & $C$\\ 
  \hline
7 & 80 & 0.69 & 48 & 0.92 & 0.69 & 0.53 & 6.35 & 0.76 & 0.90 & 0.55 & 1.15 & 1.04 & 0.57 & 0.64 & 5.20 & 0.81 & 0.66 \\ 
  7 & 80 & 1.39 & 65 & 1.06 & 0.97 & 0.84 & 7.16 & 1.02 & 1.12 & 1.15 & 1.19 & 1.03 & 1.08 & 1.10 & 3.46 & 1.16 & 1.20 \\ 
  7 & 80 & 2.08 & 81 & 1.30 & 1.34 & 0.98 & 7.88 & 1.35 & 1.78 & 1.79 & 1.37 & 1.23 & 1.37 & 1.64 & 3.90 & 1.85 & 1.78 \\ 
  7 & 200 & 0.69 & 13 & 0.76 & 0.72 & 0.45 & 4.89 & 1.08 & 0.57 & 0.46 & 0.80 & 0.73 & 0.49 & 0.49 & 0.47 & 0.48 & 0.47 \\ 
  7 & 200 & 1.39 & 38 & 0.85 & 0.94 & 0.66 & 7.97 & 1.62 & 0.80 & 0.94 & 0.89 & 0.91 & 0.72 & 1.24 & 0.89 & 0.95 & 0.96 \\ 
  7 & 200 & 2.08 & 58 & 0.82 & 0.97 & 0.74 & 9.48 & 1.67 & 1.05 & 1.53 & 0.86 & 0.95 & 0.75 & 1.72 & 1.43 & 1.48 & 1.53 \\ 
  7 & 500 & 0.69 & 0 & 0.51 & 0.48 & 0.35 & 1.01 & 0.50 & 0.43 & 0.33 & 0.52 & 0.48 & 0.36 & 0.99 & 0.37 & 0.35 & 0.35 \\ 
  7 & 500 & 1.39 & 7 & 0.69 & 0.82 & 0.53 & 3.81 & 1.26 & 0.67 & 0.69 & 0.70 & 0.82 & 0.51 & 3.82 & 0.91 & 0.70 & 0.72 \\ 
  7 & 500 & 2.08 & 27 & 0.76 & 1.00 & 0.70 & 6.93 & 1.92 & 0.78 & 1.19 & 0.78 & 1.03 & 0.68 & 6.95 & 1.47 & 1.17 & 1.22 \\ 
  10 & 80 & 0.69 & 50 & 1.00 & 0.85 & 0.54 & 0.77 & 0.72 & 0.83 & 0.68 & 1.18 & 1.11 & 0.56 & 0.59 & 5.96 & 0.80 & 0.65 \\ 
  10 & 80 & 1.39 & 69 & 1.03 & 0.88 & 0.87 & 1.07 & 1.01 & 1.29 & 1.16 & 1.19 & 1.05 & 1.06 & 1.06 & 8.02 & 2.96 & 2.07 \\ 
  10 & 80 & 2.08 & 82 & 1.26 & 1.18 & 1.02 & 1.92 & 1.53 & 2.08 & 1.76 & 1.42 & 1.21 & 1.53 & 1.66 & 8.13 & 1.75 & 1.75 \\ 
  10 & 200 & 0.69 & 16 & 0.83 & 0.81 & 0.45 & 3.45 & 0.51 & 0.54 & 0.47 & 0.89 & 0.86 & 0.46 & 0.49 & 0.47 & 0.50 & 0.56 \\ 
  10 & 200 & 1.39 & 36 & 0.82 & 0.87 & 0.63 & 5.23 & 0.81 & 1.02 & 0.94 & 0.86 & 0.87 & 0.74 & 0.88 & 0.88 & 0.91 & 0.95 \\ 
  10 & 200 & 2.08 & 60 & 0.84 & 0.91 & 0.74 & 6.71 & 1.25 & 1.35 & 1.52 & 0.88 & 0.94 & 0.77 & 1.40 & 1.41 & 1.45 & 1.51 \\ 
  10 & 500 & 0.69 & 1 & 0.56 & 0.54 & 0.38 & 1.28 & 0.50 & 0.41 & 0.36 & 0.57 & 0.53 & 0.37 & 0.76 & 0.37 & 0.38 & 0.38 \\ 
  10 & 500 & 1.39 & 7 & 0.66 & 0.77 & 0.49 & 3.80 & 1.14 & 0.61 & 0.68 & 0.68 & 0.79 & 0.48 & 3.12 & 0.69 & 0.67 & 0.70 \\ 
  10 & 500 & 2.08 & 24 & 0.75 & 0.97 & 0.68 & 6.53 & 1.74 & 1.00 & 1.20 & 0.76 & 0.99 & 0.66 & 5.66 & 1.15 & 1.14 & 1.21 \\ 
  15 & 80 & 0.69 & 64 & 1.39 & 1.29 & 0.60 & 0.71 & 11.10 & 1.57 & 0.98 & 1.62 & 1.30 & 0.57 & 0.70 & 22.08 & 1.32 & 0.78 \\ 
  15 & 80 & 1.39 & 72 & 1.63 & 1.34 & 0.94 & 1.12 & 12.40 & 1.81 & 1.70 & 1.57 & 1.25 & 1.03 & 1.09 & 25.17 & 2.06 & 2.56 \\ 
  15 & 80 & 2.08 & 80 & 1.62 & 1.40 & 1.30 & 1.58 & 10.02 & 1.88 & 1.74 & 1.77 & 1.39 & 1.53 & 1.63 & 22.69 & 2.55 & 2.29 \\ 
  15 & 200 & 0.69 & 6 & 0.88 & 0.91 & 0.52 & 0.55 & 0.56 & 0.61 & 0.61 & 0.97 & 1.07 & 0.48 & 0.52 & 0.54 & 0.59 & 0.56 \\ 
  15 & 200 & 1.39 & 20 & 1.06 & 1.18 & 0.71 & 1.06 & 0.93 & 1.42 & 1.05 & 1.14 & 1.32 & 0.77 & 0.89 & 0.90 & 1.11 & 0.99 \\ 
  15 & 200 & 2.08 & 37 & 1.14 & 1.45 & 0.81 & 1.23 & 2.11 & 1.64 & 1.39 & 1.20 & 1.47 & 0.91 & 1.32 & 2.37 & 1.57 & 1.41 \\ 
  15 & 500 & 0.69 & 0 & 0.54 & 0.53 & 0.43 & 0.63 & 0.43 & 0.43 & 0.43 & 0.58 & 0.54 & 0.39 & 0.41 & 0.40 & 0.41 & 0.43 \\ 
  15 & 500 & 1.39 & 1 & 0.62 & 0.67 & 0.50 & 0.76 & 0.64 & 0.62 & 0.64 & 0.66 & 0.68 & 0.50 & 0.63 & 0.66 & 0.63 & 0.65 \\ 
  15 & 500 & 2.08 & 10 & 0.82 & 1.02 & 0.70 & 3.45 & 1.08 & 1.32 & 1.01 & 0.85 & 1.03 & 0.64 & 0.96 & 1.01 & 1.00 & 1.01 \\ 
   \hline
\end{tabular}
\end{small}
\end{table}

\begin{sidewaystable}[ht]

\caption{Simulation results showing two-sided empirical coverage rates in percentages (and median widths) of 95\% confidence intervals for $\beta_1$ across simulation scenarios that differed by the number of covariates included into data-generating process $K\in\{7, 10, 15\}$, the sample size $N\in\{80, 200, 500\}$, the value of $\beta_1\in\{0.69, 1.39, 2.08\}$ and noise excluded or included in the fitting process. FC, Firth's correction; RR, ridge regression; $B$, $B$-tuning; $T$, $T$-tuning; $D$, tuning by leave-one-out-cross-validated deviance; $A$, tuning by Akaike's information criterion; $E$, tuning by leave-one-out-cross-validated mean squared prediction error; $C$, tuning by leave-one-out-cross-validated mean classification error.}

\setlength{\tabcolsep}{4pt}
\renewcommand{\arraystretch}{0.88}
\centering
\begin{small}
\begin{tabular}{lll|lllllll|lllllll}
  \hline
$K$ & $N$ & $\beta_1$ &  \multicolumn{7}{c}{No noise} & \multicolumn{7}{|c}{Noise} \\ 
 &  &   & FC & \multicolumn{6}{c|}{RR} & FC & \multicolumn{6}{c}{RR} \\
 &  &   &  & $B$ & $T$ & $D$ & $A$ & $E$ & $C$ &  & $B$ & $T$ & $D$ & $A$ & $E$ & $C$\\  
  \hline
7 & 80 & 0.69 & 97 (4.6) & 98 (4.1) & 99 (4.3) & 97 (2.5) & 98 (2.8) & 98 (2.2) & 98 (1.9) & 97 (5.4) & 98 (5.6) & 99 (3.3) & 98 (2.1) & 98 (2.5) & 98 (2) & 98 (2) \\ 
  7 & 80 & 1.39 & 96 (7.1) & 96 (4.7) & 98 (6.7) & 65 (4.2) & 75 (3.2) & 61 (2.5) & 21 (1.9) & 97 (7.3) & 98 (8) & 98 (4.7) & 50 (2.1) & 73 (2.6) & 41 (2) & 24 (2) \\ 
  7 & 80 & 2.08 & 96 (7.2) & 92 (5) & 98 (9.2) & 70 (5.9) & 72 (3.6) & 50 (2.8) & 0 (1.9) & 97 (7.4) & 97 (9.8) & 98 (7.5) & 30 (2.2) & 43 (2.7) & 18 (2.1) & 5 (2) \\ 
  7 & 200 & 0.69 & 97 (3) & 96 (2.6) & 97 (2.9) & 93 (1.9) & 94 (1.9) & 96 (1.9) & 96 (1.7) & 95 (3.2) & 97 (2.9) & 98 (2.5) & 96 (1.8) & 96 (1.8) & 96 (1.7) & 96 (1.7) \\ 
  7 & 200 & 1.39 & 96 (3.7) & 93 (3.4) & 97 (4) & 76 (2.4) & 76 (2.4) & 76 (2.6) & 49 (1.7) & 96 (3.9) & 95 (3.5) & 97 (3.7) & 58 (1.9) & 60 (1.9) & 56 (1.8) & 48 (1.7) \\ 
  7 & 200 & 2.08 & 97 (6.9) & 93 (5.3) & 98 (8.2) & 72 (>50) & 70 (8.1) & 72 (3) & 0 (1.8) & 96 (6.9) & 94 (5.4) & 99 (6.9) & 30 (1.9) & 19 (1.9) & 12 (1.8) & 2 (1.8) \\ 
  7 & 500 & 0.69 & 96 (2) & 95 (1.7) & 97 (1.9) & 96 (1.5) & 95 (1.5) & 95 (1.5) & 96 (1.4) & 96 (2) & 96 (1.8) & 98 (1.8) & 96 (1.4) & 96 (1.4) & 96 (1.4) & 96 (1.4) \\ 
  7 & 500 & 1.39 & 96 (2.4) & 91 (2.3) & 95 (2.5) & 77 (1.9) & 75 (1.9) & 86 (2) & 68 (1.5) & 96 (2.5) & 90 (2.3) & 97 (2.5) & 63 (1.5) & 65 (1.5) & 69 (1.5) & 67 (1.5) \\ 
  7 & 500 & 2.08 & 96 (3.5) & 94 (3.5) & 98 (3.9) & 82 (2.6) & 81 (2.8) & 84 (2.7) & 0 (1.5) & 96 (3.6) & 92 (3.5) & 98 (3.9) & 50 (1.9) & 32 (1.6) & 14 (1.6) & 1 (1.5) \\ 
  10 & 80 & 0.69 & 97 (4.7) & 98 (4.8) & 99 (4.2) & 98 (2.3) & 98 (2.7) & 98 (2.1) & 98 (1.9) & 97 (6.3) & 98 (7.1) & 99 (3.6) & 98 (2.1) & 98 (2.6) & 98 (2) & 98 (2) \\ 
  10 & 80 & 1.39 & 96 (7.2) & 97 (6) & 98 (6.1) & 62 (2.5) & 81 (2.9) & 57 (2.2) & 28 (1.9) & 97 (7.4) & 98 (13.1) & 98 (4.6) & 52 (2.1) & 76 (2.7) & 45 (2.1) & 30 (2) \\ 
  10 & 80 & 2.08 & 96 (7.2) & 97 (6.3) & 98 (8.7) & 45 (2.6) & 54 (3) & 32 (2.3) & 6 (2) & 96 (7.5) & 98 (15.1) & 98 (6.8) & 27 (2.2) & 42 (2.7) & 19 (2.1) & 10 (2) \\ 
  10 & 200 & 0.69 & 97 (3.1) & 97 (2.8) & 98 (3) & 91 (1.9) & 97 (1.9) & 97 (1.9) & 96 (1.7) & 96 (3.3) & 96 (3.2) & 98 (2.7) & 97 (1.8) & 97 (1.8) & 97 (1.8) & 97 (1.7) \\ 
  10 & 200 & 1.39 & 95 (3.8) & 95 (3.5) & 97 (3.9) & 73 (2.2) & 76 (2.1) & 72 (2) & 52 (1.7) & 96 (3.9) & 96 (3.8) & 96 (3.6) & 63 (1.9) & 66 (1.9) & 62 (1.8) & 54 (1.8) \\ 
  10 & 200 & 2.08 & 96 (6.9) & 93 (5.5) & 98 (7.9) & 60 (2.9) & 52 (2.3) & 40 (2.1) & 3 (1.8) & 96 (7) & 95 (5.7) & 98 (6.5) & 28 (2) & 23 (2) & 17 (1.9) & 9 (1.8) \\ 
  10 & 500 & 0.69 & 94 (2) & 95 (1.8) & 95 (1.9) & 95 (1.5) & 95 (1.4) & 95 (1.5) & 96 (1.4) & 94 (2.1) & 94 (1.9) & 96 (1.9) & 95 (1.4) & 95 (1.4) & 95 (1.4) & 95 (1.4) \\ 
  10 & 500 & 1.39 & 97 (2.4) & 94 (2.3) & 98 (2.5) & 77 (1.8) & 74 (1.6) & 82 (1.7) & 70 (1.5) & 97 (2.5) & 93 (2.4) & 98 (2.5) & 68 (1.5) & 69 (1.5) & 73 (1.5) & 70 (1.5) \\ 
  10 & 500 & 2.08 & 96 (3.5) & 94 (3.5) & 98 (3.9) & 72 (2.3) & 60 (2) & 57 (1.9) & 1 (1.5) & 96 (3.6) & 92 (3.5) & 98 (3.8) & 40 (1.8) & 24 (1.6) & 21 (1.6) & 7 (1.5) \\ 
  15 & 80 & 0.69 & 98 (6.7) & 98 (7.7) & 99 (4.9) & 98 (2.7) & 95 (3.2) & 98 (2.7) & 97 (2.2) & 97 (8.5) & 99 (11) & 99 (4.6) & 98 (2.4) & 90 (3.6) & 99 (2.5) & 98 (2.1) \\ 
  15 & 80 & 1.39 & 97 (7.6) & 98 (12) & 98 (5.9) & 67 (2.7) & 82 (3.2) & 64 (2.7) & 48 (2.2) & 98 (8.9) & 98 (13.6) & 99 (5.1) & 63 (2.5) & 78 (3.6) & 61 (2.5) & 47 (2.1) \\ 
  15 & 80 & 2.08 & 97 (7.9) & 98 (15.1) & 99 (7.5) & 47 (2.9) & 60 (3.4) & 46 (2.9) & 33 (2.2) & 97 (9.1) & 98 (16.8) & 98 (6.1) & 37 (2.6) & 56 (3.9) & 39 (2.6) & 30 (2.2) \\ 
  15 & 200 & 0.69 & 95 (3.3) & 96 (3.3) & 96 (3.3) & 96 (2.1) & 96 (2.1) & 96 (2.2) & 96 (2) & 95 (3.7) & 96 (3.8) & 97 (3.1) & 96 (2) & 96 (2.1) & 96 (2.1) & 96 (2) \\ 
  15 & 200 & 1.39 & 95 (4) & 95 (3.9) & 96 (4) & 75 (2.3) & 78 (2.3) & 76 (2.3) & 68 (2.2) & 95 (4.3) & 96 (4.4) & 97 (3.7) & 69 (2.1) & 75 (2.2) & 70 (2.2) & 65 (2) \\ 
  15 & 200 & 2.08 & 96 (4.5) & 96 (4.7) & 98 (4.7) & 58 (2.5) & 56 (2.5) & 56 (2.5) & 44 (2.3) & 96 (5) & 97 (5.4) & 98 (4.5) & 43 (2.3) & 46 (2.3) & 44 (2.3) & 40 (2.2) \\ 
  15 & 500 & 0.69 & 95 (2) & 95 (1.9) & 95 (2) & 95 (1.6) & 94 (1.5) & 95 (1.6) & 95 (1.6) & 95 (2.2) & 96 (2.1) & 97 (2.1) & 96 (1.6) & 96 (1.5) & 96 (1.6) & 96 (1.6) \\ 
  15 & 500 & 1.39 & 96 (2.4) & 95 (2.3) & 96 (2.4) & 86 (1.8) & 81 (1.7) & 86 (1.8) & 78 (1.7) & 96 (2.6) & 94 (2.4) & 97 (2.4) & 79 (1.7) & 76 (1.7) & 81 (1.7) & 81 (1.7) \\ 
  15 & 500 & 2.08 & 95 (2.9) & 91 (2.9) & 96 (3.1) & 68 (2.1) & 61 (2) & 68 (2.1) & 53 (1.8) & 95 (3.1) & 92 (3) & 96 (3) & 57 (1.9) & 49 (1.8) & 55 (1.9) & 49 (1.8) \\ 
   \hline
\end{tabular}
\end{small}
\end{sidewaystable}

\pagebreak

\begin{figure}
  \caption{Simulation results showing one-sided empirical coverage rates of 97.5\% confidence intervals for $\beta_1$ across simulation scenarios that differed by the number of covariates included into data-generating process $K\in\{7, 10, 15\}$, the sample size $N\in\{80, 200, 500\}$, the value of $\beta_1\in\{0.69, 1.39, 2.08\}$ and noise excluded or included in the fitting process. Bottom shows one-sided lower and top one-sided upper 97.5\% confidence limit. FC, Firth's correction; RR, ridge regression.}
  \includegraphics[width=\linewidth]{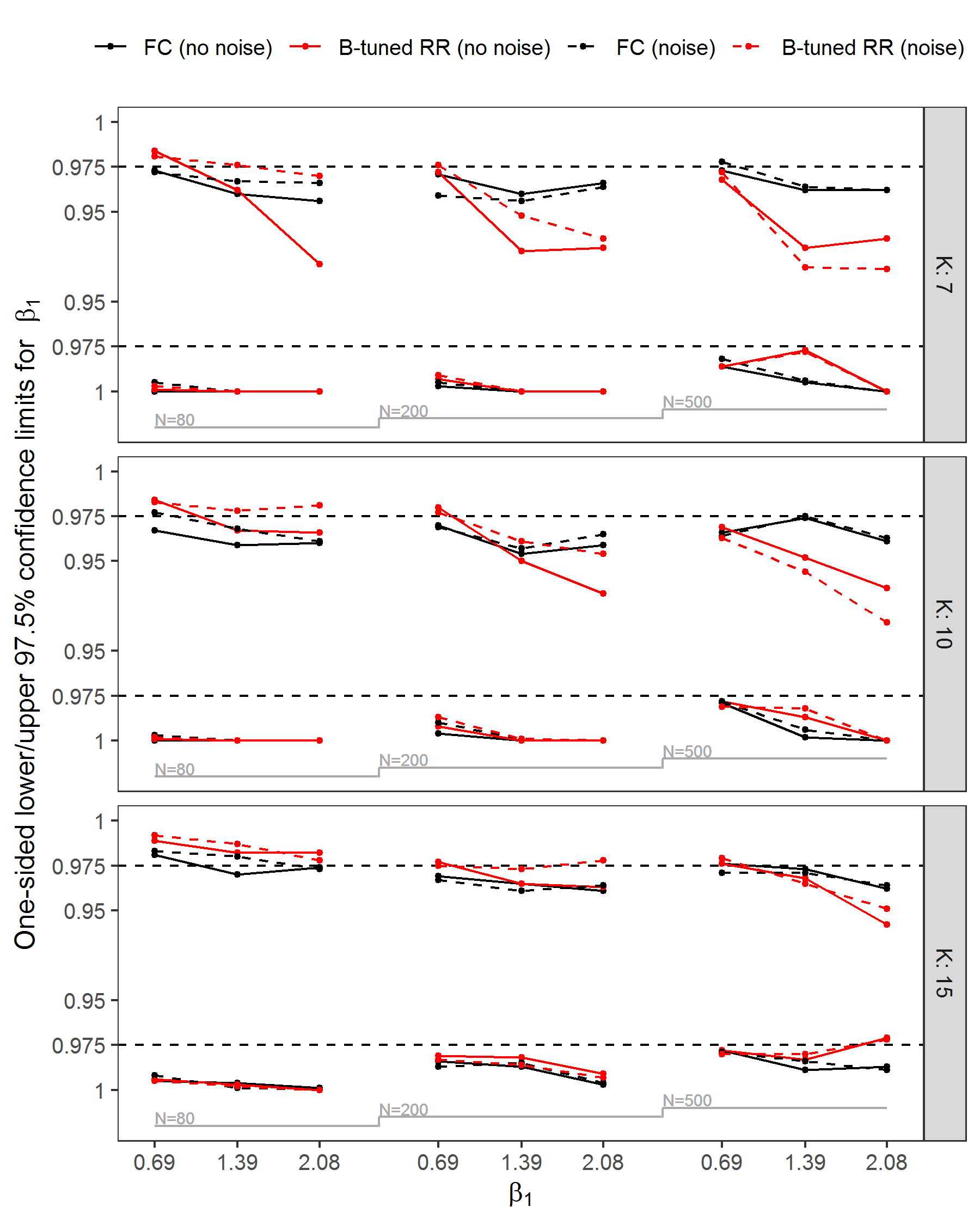}
\end{figure}

\pagebreak

\begin{sidewaysfigure}[ht]
  \caption{Simulation results showing bias (left) and variance (right) of $\hat{\beta_1}$ across simulation scenarios that differed by the number of covariates included into data-generating process $K\in\{7, 10, 15\}$, the sample size $N\in\{80, 200, 500\}$, the value of $\beta_1\in\{0.69, 1.39, 2.08\}$ and noise excluded or included in the fitting process. FC, Firth's correction; RR, ridge regression.}
  \includegraphics[width=\linewidth]{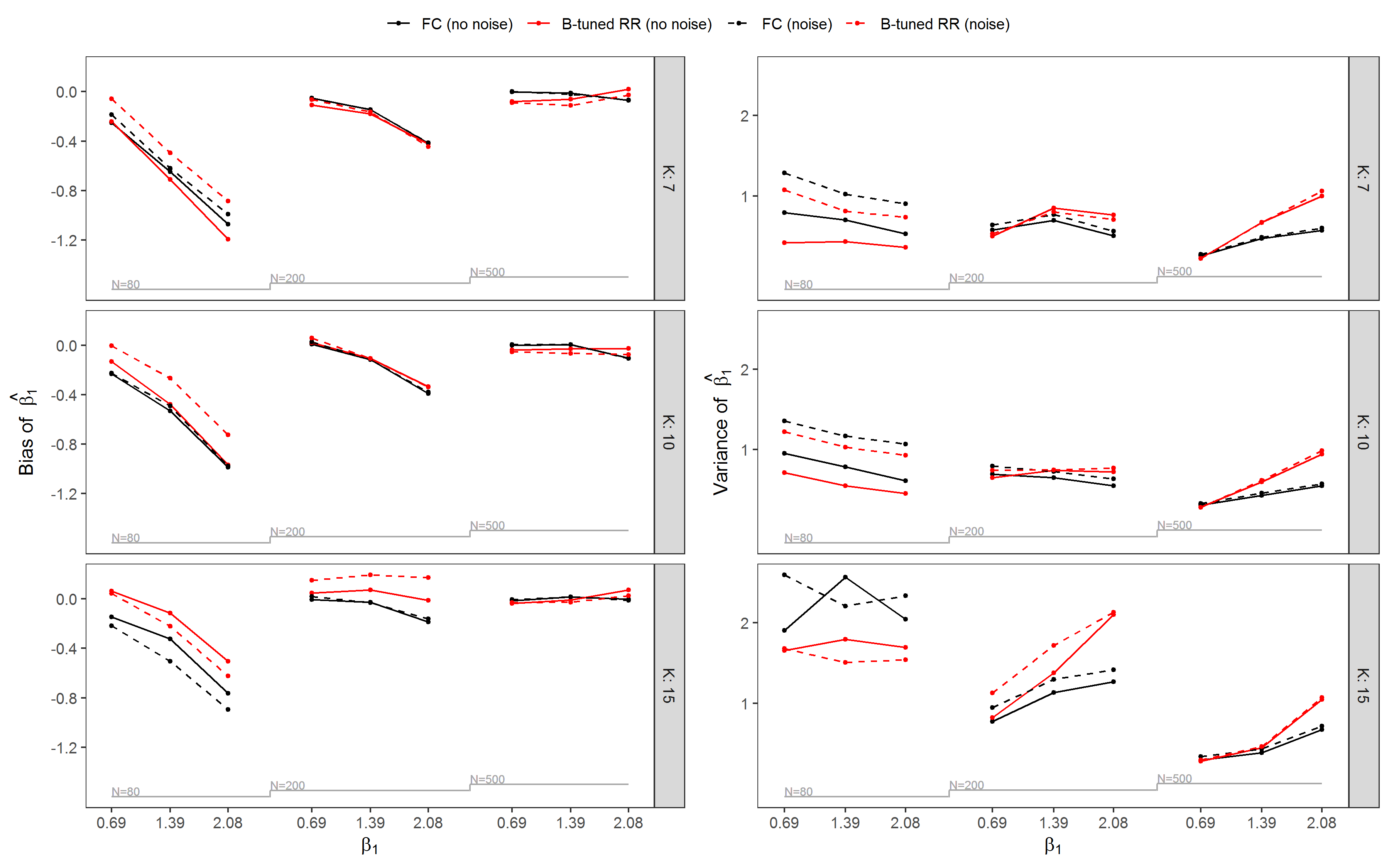}
\end{sidewaysfigure}

\pagebreak

\begin{table}[ht]

\caption{Simulation results showing separation prevalence (SP) and root mean squared errors of $\beta_2$ across simulation scenarios that differed by the number of covariates included into data-generating process $K\in\{7, 10, 15\}$, the sample size $N\in\{80, 200, 500\}$, the value of $\beta_1\in\{0.69, 1.39, 2.08\}$ and noise excluded or included in the fitting process. FC, Firth's correction; RR, ridge regression; $B$, $B$-tuning; $T$, $T$-tuning; $D$, tuning by leave-one-out-cross-validated deviance; $A$, tuning by Akaike's information criterion; $E$, tuning by leave-one-out-cross-validated mean squared prediction error; $C$, tuning by leave-one-out-cross-validated mean classification error.}

\setlength{\tabcolsep}{4pt}
\renewcommand{\arraystretch}{0.88}
\centering
\begin{small}
\begin{tabular}{rrrr|r|rrrrrrr|rrrrrrr}
  \hline
$K$ & $N$ & $\beta_1$ & $\beta_2$ & SP  &  \multicolumn{7}{c}{No noise} & \multicolumn{7}{|c}{Noise} \\ 
 &  & & & (\%) & FC & \multicolumn{6}{c|}{RR} & FC & \multicolumn{6}{c}{RR} \\
 &  & & &  &  & $B$ & $T$ & $D$ & $A$ & $E$ & $C$ &  & $B$ & $T$ & $D$ & $A$ & $E$ & $C$\\ 
  \hline
7 & 80 & 0.69 & 0.69 & 48 & 0.82 & 0.74 & 0.54 & 2.40 & 0.87 & 1.09 & 0.53 & 0.90 & 0.90 & 0.52 & 0.77 & 5.76 & 0.95 & 0.80 \\ 
  7 & 80 & 1.39 & 0.69 & 65 & 0.77 & 0.70 & 0.62 & 2.22 & 0.86 & 0.99 & 0.50 & 0.85 & 0.89 & 0.47 & 0.60 & 4.36 & 0.97 & 0.69 \\ 
  7 & 80 & 2.08 & 0.69 & 81 & 0.83 & 0.75 & 0.83 & 2.78 & 0.89 & 1.67 & 0.51 & 0.94 & 0.95 & 0.66 & 0.86 & 2.96 & 1.16 & 0.69 \\ 
  7 & 200 & 0.69 & 0.69 & 13 & 0.48 & 0.45 & 0.39 & 0.44 & 0.44 & 0.45 & 0.41 & 0.50 & 0.47 & 0.40 & 0.42 & 0.42 & 0.42 & 0.42 \\ 
  7 & 200 & 1.39 & 0.69 & 38 & 0.50 & 0.47 & 0.46 & 0.46 & 0.46 & 0.47 & 0.37 & 0.51 & 0.49 & 0.42 & 0.40 & 0.39 & 0.40 & 0.39 \\ 
  7 & 200 & 2.08 & 0.69 & 58 & 0.50 & 0.48 & 0.50 & 0.48 & 0.48 & 0.48 & 0.39 & 0.52 & 0.51 & 0.52 & 0.42 & 0.42 & 0.42 & 0.40 \\ 
  7 & 500 & 0.69 & 0.69 & 0 & 0.31 & 0.30 & 0.28 & 0.29 & 0.29 & 0.30 & 0.28 & 0.32 & 0.30 & 0.27 & 0.28 & 0.28 & 0.28 & 0.28 \\ 
  7 & 500 & 1.39 & 0.69 & 7 & 0.30 & 0.29 & 0.29 & 0.28 & 0.28 & 0.29 & 0.26 & 0.31 & 0.30 & 0.30 & 0.28 & 0.28 & 0.28 & 0.27 \\ 
  7 & 500 & 2.08 & 0.69 & 27 & 0.31 & 0.31 & 0.32 & 0.30 & 0.30 & 0.30 & 0.26 & 0.32 & 0.32 & 0.32 & 0.30 & 0.28 & 0.28 & 0.27 \\ 
  10 & 80 & 0.69 & 0.69 & 50 & 0.88 & 0.85 & 0.51 & 0.66 & 0.81 & 0.79 & 0.65 & 0.97 & 1.02 & 0.50 & 0.57 & 8.63 & 0.64 & 0.61 \\ 
  10 & 80 & 1.39 & 0.69 & 69 & 0.90 & 0.87 & 0.62 & 0.67 & 0.80 & 0.81 & 0.62 & 0.96 & 1.06 & 0.51 & 0.60 & 9.02 & 1.25 & 0.93 \\ 
  10 & 80 & 2.08 & 0.69 & 82 & 0.86 & 0.83 & 0.81 & 0.87 & 0.77 & 0.71 & 0.58 & 0.96 & 1.06 & 0.55 & 0.61 & 8.62 & 1.06 & 0.61 \\ 
  10 & 200 & 0.69 & 0.69 & 16 & 0.51 & 0.49 & 0.39 & 0.44 & 0.43 & 0.44 & 0.42 & 0.54 & 0.54 & 0.39 & 0.43 & 0.43 & 0.44 & 0.43 \\ 
  10 & 200 & 1.39 & 0.69 & 36 & 0.51 & 0.49 & 0.46 & 0.43 & 0.41 & 0.42 & 0.40 & 0.52 & 0.52 & 0.41 & 0.40 & 0.40 & 0.40 & 0.41 \\ 
  10 & 200 & 2.08 & 0.69 & 60 & 0.53 & 0.52 & 0.53 & 0.47 & 0.43 & 0.45 & 0.41 & 0.54 & 0.56 & 0.55 & 0.42 & 0.41 & 0.42 & 0.42 \\ 
  10 & 500 & 0.69 & 0.69 & 1 & 0.32 & 0.30 & 0.28 & 0.29 & 0.29 & 0.30 & 0.28 & 0.33 & 0.31 & 0.27 & 0.29 & 0.29 & 0.29 & 0.29 \\ 
  10 & 500 & 1.39 & 0.69 & 7 & 0.31 & 0.30 & 0.30 & 0.29 & 0.29 & 0.29 & 0.27 & 0.32 & 0.31 & 0.31 & 0.29 & 0.29 & 0.29 & 0.29 \\ 
  10 & 500 & 2.08 & 0.69 & 24 & 0.31 & 0.31 & 0.31 & 0.29 & 0.29 & 0.28 & 0.26 & 0.32 & 0.32 & 0.32 & 0.29 & 0.27 & 0.27 & 0.27 \\ 
  15 & 80 & 0.69 & 0.69 & 64 & 1.04 & 1.08 & 0.54 & 0.62 & 9.10 & 1.64 & 0.69 & 1.05 & 1.02 & 0.50 & 0.72 & 13.06 & 1.21 & 0.73 \\ 
  15 & 80 & 1.39 & 0.69 & 72 & 1.05 & 1.11 & 0.58 & 0.64 & 10.18 & 1.02 & 0.93 & 0.99 & 1.04 & 0.50 & 0.59 & 13.37 & 1.66 & 1.15 \\ 
  15 & 80 & 2.08 & 0.69 & 80 & 1.02 & 1.07 & 0.69 & 0.71 & 7.71 & 1.63 & 1.10 & 1.00 & 1.06 & 0.54 & 0.63 & 14.07 & 0.85 & 0.73 \\ 
  15 & 200 & 0.69 & 0.69 & 6 & 0.60 & 0.62 & 0.44 & 0.44 & 0.44 & 0.45 & 0.45 & 0.64 & 0.73 & 0.40 & 0.43 & 0.44 & 0.44 & 0.45 \\ 
  15 & 200 & 1.39 & 0.69 & 20 & 0.60 & 0.63 & 0.51 & 0.43 & 0.43 & 0.45 & 0.45 & 0.63 & 0.74 & 0.44 & 0.41 & 0.44 & 0.42 & 0.44 \\ 
  15 & 200 & 2.08 & 0.69 & 37 & 0.61 & 0.65 & 0.61 & 0.45 & 0.46 & 0.46 & 0.46 & 0.63 & 0.75 & 0.55 & 0.42 & 0.45 & 0.43 & 0.44 \\ 
  15 & 500 & 0.69 & 0.69 & 0 & 0.37 & 0.36 & 0.33 & 0.32 & 0.31 & 0.32 & 0.32 & 0.38 & 0.37 & 0.31 & 0.31 & 0.31 & 0.31 & 0.32 \\ 
  15 & 500 & 1.39 & 0.69 & 1 & 0.36 & 0.36 & 0.35 & 0.30 & 0.30 & 0.31 & 0.31 & 0.37 & 0.37 & 0.33 & 0.30 & 0.30 & 0.30 & 0.31 \\ 
  15 & 500 & 2.08 & 0.69 & 10 & 0.37 & 0.37 & 0.37 & 0.32 & 0.31 & 0.32 & 0.31 & 0.37 & 0.38 & 0.37 & 0.30 & 0.30 & 0.31 & 0.31 \\ 
   \hline
\end{tabular}
\end{small}
\end{table}

\begin{sidewaystable}[ht]

\caption{Simulation results showing two-sided empirical coverage rates in percentages (and median widths) of 95\% confidence intervals for $\beta_2$ across simulation scenarios that differed by the number of covariates included into data-generating process $K\in\{7, 10, 15\}$, the sample size $N\in\{80, 200, 500\}$, the value of $\beta_1\in\{0.69, 1.39, 2.08\}$ and noise excluded or included in the fitting process. FC, Firth's correction; RR, ridge regression; $B$, $B$-tuning; $T$, $T$-tuning; $D$, tuning by leave-one-out-cross-validated deviance; $A$, tuning by Akaike's information criterion; $E$, tuning by leave-one-out-cross-validated mean squared prediction error; $C$, tuning by leave-one-out-cross-validated mean classification error.}

\setlength{\tabcolsep}{4pt}
\renewcommand{\arraystretch}{0.88}
\centering
\begin{small}
\begin{tabular}{llll|lllllll|lllllll}
  \hline
$K$ & $N$ & $\beta_1$ & $\beta_2$ &  \multicolumn{7}{c}{No noise} & \multicolumn{7}{|c}{Noise} \\ 
 &  &  &  & FC & \multicolumn{6}{c|}{RR} & FC & \multicolumn{6}{c}{RR} \\
 &  & &  &  & $B$ & $T$ & $D$ & $A$ & $E$ & $C$ &  & $B$ & $T$ & $D$ & $A$ & $E$ & $C$\\  
  \hline
7 & 80 & 0.69 & 0.69 & 96 (3.1) & 96 (2.8) & 97 (2.9) & 91 (1.9) & 92 (2.1) & 91 (1.7) & 92 (1.6) & 96 (3.3) & 96 (3.3) & 98 (2.5) & 93 (1.6) & 93 (1.9) & 92 (1.6) & 92 (1.6) \\ 
  7 & 80 & 1.39 & 0.69 & 96 (3.1) & 97 (2.9) & 97 (3.1) & 95 (2.6) & 95 (2.3) & 95 (1.9) & 96 (1.6) & 96 (3.3) & 96 (3.4) & 98 (3) & 96 (1.7) & 96 (2) & 96 (1.6) & 95 (1.6) \\ 
  7 & 80 & 2.08 & 0.69 & 96 (3.2) & 97 (2.9) & 95 (3.3) & 94 (2.8) & 95 (2.5) & 94 (2.1) & 94 (1.6) & 95 (3.4) & 96 (3.4) & 95 (3.5) & 94 (1.7) & 95 (2.1) & 94 (1.6) & 94 (1.6) \\ 
  7 & 200 & 0.69 & 0.69 & 96 (1.9) & 96 (1.8) & 97 (1.9) & 93 (1.5) & 93 (1.5) & 93 (1.5) & 93 (1.3) & 96 (2) & 96 (1.9) & 97 (1.8) & 92 (1.4) & 92 (1.4) & 92 (1.4) & 92 (1.3) \\ 
  7 & 200 & 1.39 & 0.69 & 96 (1.9) & 97 (1.9) & 96 (1.9) & 96 (1.6) & 96 (1.6) & 96 (1.6) & 95 (1.3) & 96 (2) & 97 (2) & 97 (2) & 96 (1.4) & 95 (1.4) & 95 (1.4) & 95 (1.3) \\ 
  7 & 200 & 2.08 & 0.69 & 96 (1.9) & 96 (1.9) & 96 (1.9) & 95 (1.8) & 95 (1.8) & 95 (1.7) & 94 (1.3) & 95 (2) & 96 (2) & 95 (2) & 94 (1.4) & 95 (1.4) & 94 (1.4) & 94 (1.3) \\ 
  7 & 500 & 0.69 & 0.69 & 95 (1.2) & 96 (1.2) & 95 (1.2) & 94 (1.1) & 94 (1) & 94 (1.1) & 94 (1) & 95 (1.2) & 96 (1.2) & 96 (1.2) & 94 (1) & 94 (1) & 94 (1) & 94 (1) \\ 
  7 & 500 & 1.39 & 0.69 & 95 (1.2) & 96 (1.2) & 95 (1.2) & 95 (1.1) & 95 (1.1) & 95 (1.1) & 95 (1) & 96 (1.2) & 96 (1.2) & 96 (1.2) & 95 (1) & 95 (1) & 95 (1) & 95 (1) \\ 
  7 & 500 & 2.08 & 0.69 & 95 (1.2) & 95 (1.2) & 95 (1.2) & 96 (1.2) & 96 (1.2) & 95 (1.2) & 95 (1) & 96 (1.2) & 96 (1.2) & 95 (1.2) & 94 (1.1) & 95 (1.1) & 94 (1) & 95 (1) \\ 
  10 & 80 & 0.69 & 0.69 & 95 (3.2) & 96 (3.2) & 97 (2.9) & 94 (1.8) & 94 (2.1) & 94 (1.7) & 92 (1.6) & 96 (3.5) & 95 (3.9) & 98 (2.7) & 94 (1.7) & 94 (2) & 93 (1.6) & 93 (1.6) \\ 
  10 & 80 & 1.39 & 0.69 & 94 (3.3) & 95 (3.2) & 94 (3.2) & 93 (2) & 94 (2.2) & 93 (1.8) & 92 (1.6) & 96 (3.6) & 95 (4) & 97 (3) & 94 (1.7) & 92 (2.1) & 94 (1.6) & 93 (1.6) \\ 
  10 & 80 & 2.08 & 0.69 & 95 (3.2) & 95 (3.1) & 94 (3.4) & 94 (2) & 95 (2.2) & 94 (1.8) & 93 (1.6) & 96 (3.5) & 95 (3.9) & 95 (3.5) & 95 (1.7) & 93 (2.1) & 94 (1.6) & 93 (1.6) \\ 
  10 & 200 & 0.69 & 0.69 & 95 (2) & 96 (1.9) & 96 (1.9) & 93 (1.4) & 93 (1.4) & 92 (1.4) & 91 (1.3) & 95 (2.1) & 96 (2) & 97 (1.8) & 93 (1.4) & 93 (1.4) & 92 (1.4) & 92 (1.4) \\ 
  10 & 200 & 1.39 & 0.69 & 96 (2) & 96 (2) & 95 (2) & 95 (1.6) & 96 (1.5) & 95 (1.5) & 93 (1.3) & 96 (2.1) & 96 (2.1) & 96 (2) & 94 (1.4) & 94 (1.4) & 94 (1.4) & 93 (1.3) \\ 
  10 & 200 & 2.08 & 0.69 & 94 (2) & 95 (2) & 94 (2) & 95 (1.8) & 96 (1.6) & 95 (1.5) & 94 (1.3) & 95 (2.1) & 95 (2.1) & 94 (2.1) & 95 (1.4) & 95 (1.5) & 94 (1.4) & 94 (1.4) \\ 
  10 & 500 & 0.69 & 0.69 & 95 (1.2) & 95 (1.2) & 95 (1.2) & 93 (1.1) & 93 (1) & 93 (1.1) & 93 (1) & 95 (1.3) & 96 (1.2) & 96 (1.2) & 93 (1) & 93 (1) & 93 (1) & 93 (1) \\ 
  10 & 500 & 1.39 & 0.69 & 94 (1.2) & 94 (1.2) & 94 (1.2) & 95 (1.1) & 94 (1.1) & 93 (1.1) & 94 (1) & 94 (1.3) & 95 (1.2) & 94 (1.3) & 94 (1.1) & 93 (1) & 94 (1) & 93 (1) \\ 
  10 & 500 & 2.08 & 0.69 & 95 (1.2) & 95 (1.2) & 94 (1.2) & 96 (1.1) & 96 (1.1) & 96 (1.1) & 96 (1) & 95 (1.3) & 95 (1.3) & 94 (1.3) & 96 (1.1) & 96 (1) & 96 (1) & 95 (1) \\ 
  15 & 80 & 0.69 & 0.69 & 96 (4.2) & 96 (4.7) & 98 (3.6) & 96 (2.1) & 92 (2.5) & 95 (2.1) & 95 (1.7) & 97 (5) & 97 (6.7) & 98 (3.4) & 96 (2) & 88 (2.7) & 96 (2) & 96 (1.7) \\ 
  15 & 80 & 1.39 & 0.69 & 96 (4.1) & 95 (4.7) & 97 (3.7) & 97 (2.1) & 92 (2.5) & 96 (2.1) & 95 (1.7) & 98 (5.1) & 97 (6.7) & 98 (3.6) & 97 (2) & 87 (2.6) & 97 (2) & 96 (1.7) \\ 
  15 & 80 & 2.08 & 0.69 & 97 (4.2) & 96 (4.7) & 96 (4.1) & 97 (2.2) & 92 (2.5) & 96 (2.2) & 96 (1.7) & 97 (5) & 97 (6.7) & 98 (3.9) & 97 (2) & 88 (2.7) & 97 (2) & 96 (1.7) \\ 
  15 & 200 & 0.69 & 0.69 & 96 (2.3) & 95 (2.3) & 96 (2.3) & 96 (1.6) & 96 (1.6) & 96 (1.7) & 95 (1.6) & 95 (2.4) & 94 (2.5) & 97 (2.2) & 95 (1.5) & 96 (1.6) & 95 (1.6) & 95 (1.5) \\ 
  15 & 200 & 1.39 & 0.69 & 94 (2.3) & 94 (2.3) & 95 (2.3) & 96 (1.7) & 96 (1.7) & 96 (1.7) & 95 (1.6) & 95 (2.4) & 93 (2.5) & 95 (2.3) & 96 (1.6) & 96 (1.6) & 96 (1.6) & 96 (1.5) \\ 
  15 & 200 & 2.08 & 0.69 & 95 (2.3) & 94 (2.4) & 94 (2.4) & 96 (1.8) & 96 (1.7) & 96 (1.8) & 95 (1.6) & 95 (2.4) & 94 (2.6) & 95 (2.4) & 96 (1.7) & 96 (1.7) & 96 (1.7) & 96 (1.6) \\ 
  15 & 500 & 0.69 & 0.69 & 95 (1.4) & 95 (1.4) & 95 (1.4) & 95 (1.2) & 95 (1.2) & 95 (1.2) & 94 (1.2) & 94 (1.4) & 94 (1.4) & 96 (1.4) & 95 (1.2) & 95 (1.1) & 95 (1.2) & 94 (1.2) \\ 
  15 & 500 & 1.39 & 0.69 & 94 (1.4) & 94 (1.4) & 94 (1.4) & 95 (1.2) & 94 (1.2) & 95 (1.2) & 95 (1.2) & 94 (1.4) & 94 (1.4) & 94 (1.4) & 95 (1.2) & 95 (1.2) & 95 (1.2) & 95 (1.2) \\ 
  15 & 500 & 2.08 & 0.69 & 94 (1.4) & 95 (1.4) & 94 (1.4) & 96 (1.3) & 96 (1.2) & 95 (1.3) & 95 (1.2) & 95 (1.4) & 95 (1.5) & 95 (1.4) & 96 (1.2) & 96 (1.2) & 96 (1.2) & 96 (1.2) \\ 
   \hline
\end{tabular}
\end{small}
\end{sidewaystable}

\pagebreak

\begin{figure}
  \caption{Simulation results showing one-sided empirical coverage rates of 97.5\% confidence intervals for $\beta_2$ across simulation scenarios that differed by the number of covariates included into data-generating process $K\in\{7, 10, 15\}$, the sample size $N\in\{80, 200, 500\}$, the value of $\beta_1\in\{0.69, 1.39, 2.08\}$ and noise excluded or included in the fitting process. Bottom shows one-sided lower and top one-sided upper 97.5\% confidence limit. FC, Firth's correction; RR, ridge regression.}
  \includegraphics[width=\linewidth]{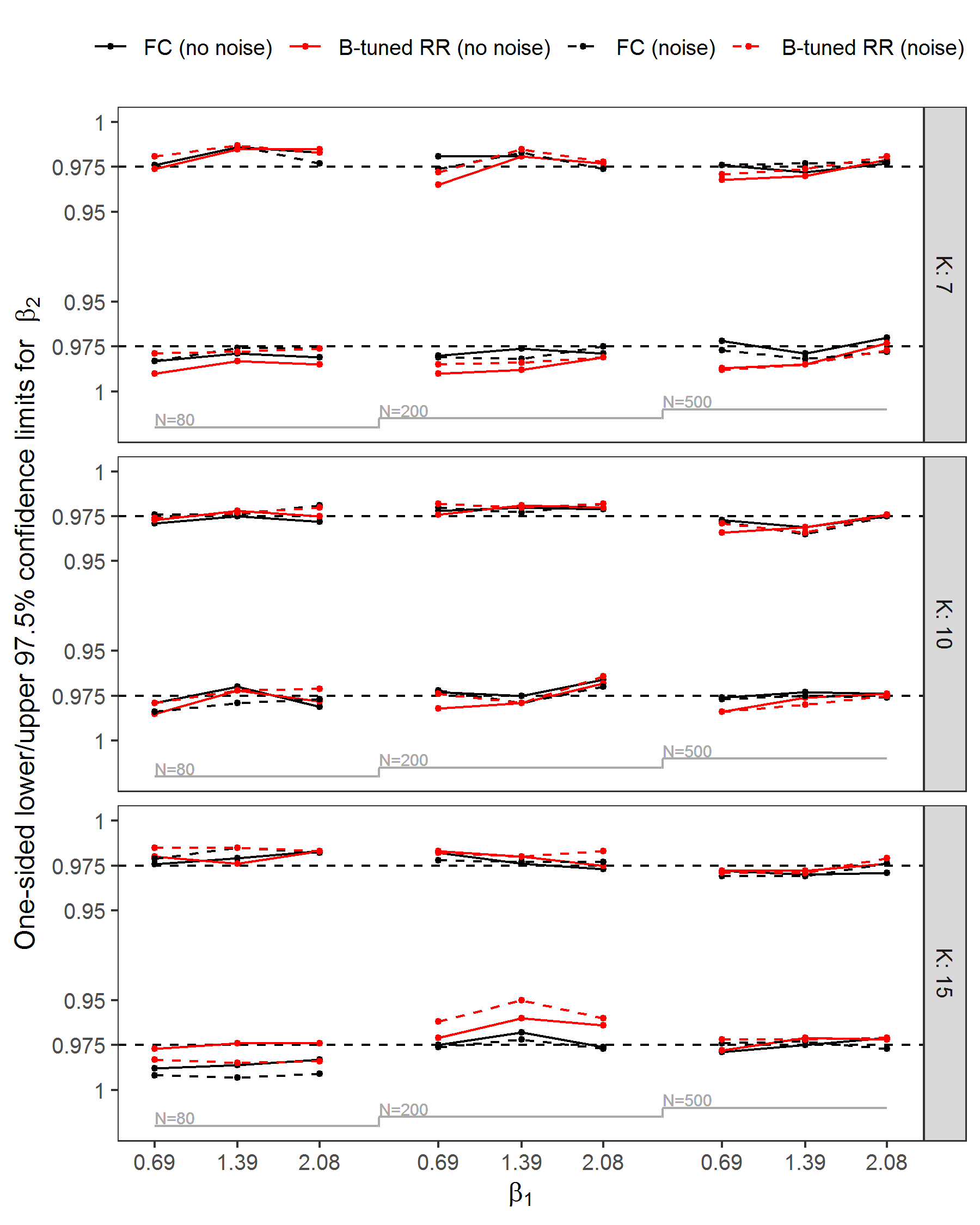}
\end{figure}

\pagebreak

\begin{sidewaysfigure}
  \caption{Simulation results showing bias (left) and variance (right) of $\hat{\beta_2}$ across simulation scenarios that differed by the number of covariates included into data-generating process $K\in\{7, 10, 15\}$, the sample size $N\in\{80, 200, 500\}$, the value of $\beta_1\in\{0.69, 1.39, 2.08\}$ and noise excluded or included in the fitting process. FC, Firth's correction; RR, ridge regression.}
  \includegraphics[width=\linewidth]{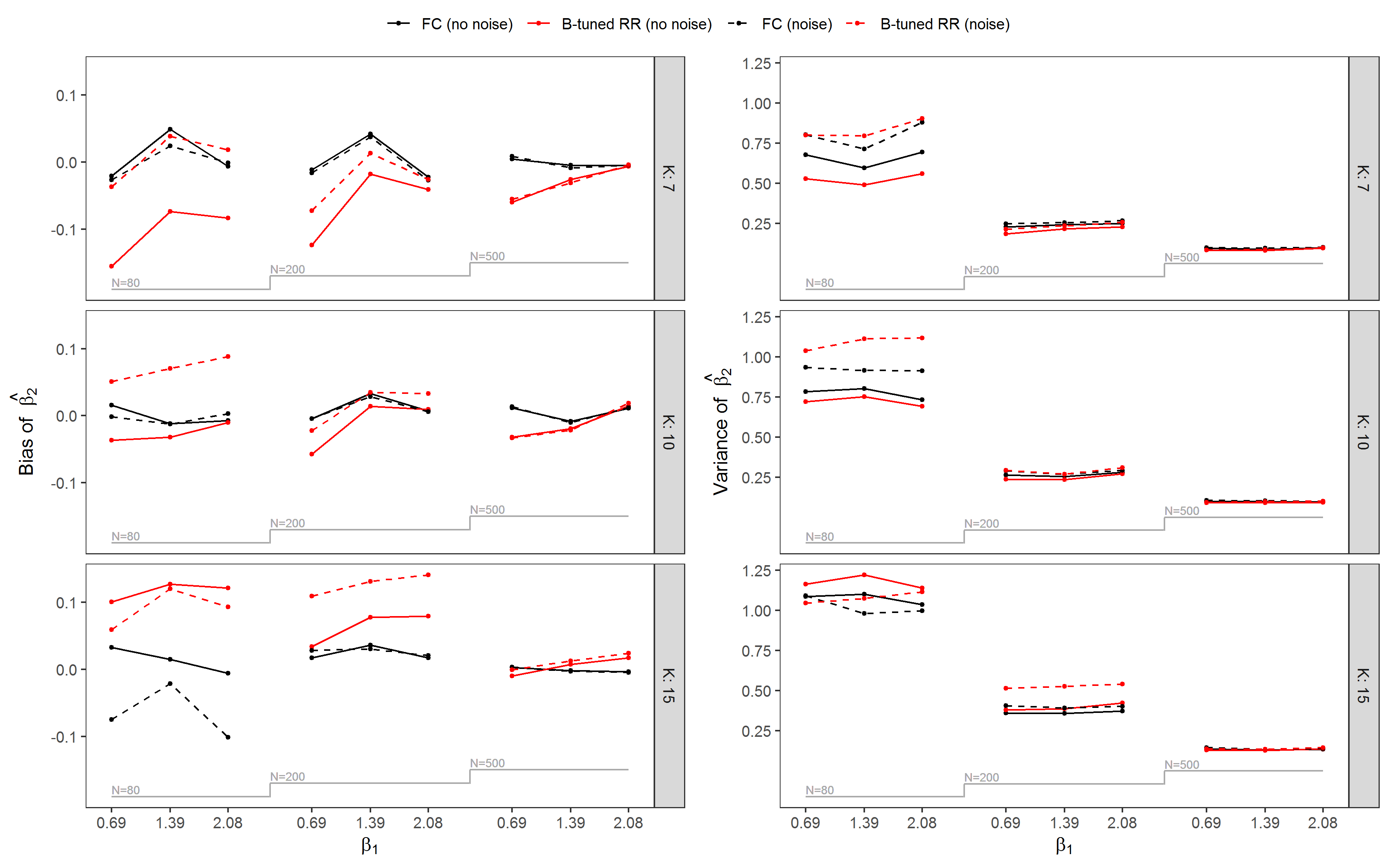}
\end{sidewaysfigure}

\end{document}